\documentclass[10pt]{article}
\usepackage{etex}
\usepackage{epsfig,amsmath,latexsym,amssymb,amsthm,amscd}
\usepackage{enumitem}
\usepackage{hyperref}
\usepackage{multirow}
\usepackage{graphicx}
\usepackage{float}
\usepackage{subcaption}
\usepackage{dsfont}
\usepackage{titlesec}
\titleformat{\subsubsection}[runin]{\normalfont\bfseries}{\thesubsubsection.}{3pt}{}[.]
\def\R{{\mathbb R}}
\def\N{{\mathbb N}}
\def\p{{\mathbb P}}

\def\E{{\mathbb E}}
\def\F{{\mathbb F}}

\newcounter{thm}
\newcounter{ex}
\newcounter{re}
\newtheorem{Theorem}[thm]{Theorem}

\newtheorem{Corollary}[thm]{Corollary}

\theoremstyle{definition}

\newtheorem{Assumption}[thm]{Assumption}

\newtheorem*{Remark*}{Remark}
\newtheorem*{FurtherRemarks*}{Further remarks}
\newtheorem*{KeyObservation}{Key observation}

\usepackage{environ}
\NewEnviron{myequation1}{%
\begin{equation}
\scalebox{0.85}{$\BODY$}
\end{equation}
}
\NewEnviron{myequation2}{%
\begin{equation}
\scalebox{0.92}{$\BODY$}
\end{equation}
}

\newcommand{\footremember}[2]{%
   \footnote{#2}
    \newcounter{#1}
    \setcounter{#1}{\value{footnote}}%
}
\newcommand{\footrecall}[1]{%
    \footnotemark[\value{#1}]%
} 

\author{Philipp Harms\thanks{Institute of Mathematics, Albert Ludwigs University of Freiburg, 79104 Freiburg, Germany.} \and David Stefanovits\footremember{ds}{Department of Mathematics, ETH Zurich, 8092 Zurich, Switzerland} \and Josef Teichmann\footrecall{ds} \and Mario V.~W\"uthrich\footrecall{ds}${}\;{}$\thanks{Swiss Finance Institute SFI, Walchestrasse 9, 8006 Zurich, Switzerland}}

\date{August 31, 2016}
\title{Consistent Re-Calibration of the Discrete-Time Multifactor Vasi\v{c}ek Model\thanks{We gratefully acknowledge support by ETH Foundation and SNF grant 149879. We thank Dr.\ Hansj\"org Furrer for supporting this SNF project. Correspondence: \href{philipp.harms@stochastik.uni-freiburg.de}{philipp.harms@stochastik.uni-freiburg.de}.}}

\begin{document}
\maketitle
\begin{abstract}
\noindent The discrete-time multifactor Vasi\v{c}ek model is a tractable Gaussian spot rate model. Typically, two- or three-factor versions allow one to capture the dependence structure between yields with different times to maturity in an appropriate way. In practice, re-calibration of the model to the prevailing market conditions leads to model parameters that change over time. Therefore, the model parameters should be understood as being time-dependent or even stochastic. Following the consistent re-calibration (CRC) approach, we construct models as concatenations of yield curve increments of Hull--White extended multifactor Vasi\v{c}ek models with different parameters. The CRC approach provides attractive tractable models that preserve the no-arbitrage premise. As~a numerical example, we fit Swiss interest rates using CRC multifactor Vasi\v{c}ek models.
\end{abstract}

\section{Introduction}
The tractability of affine models, such as the Vasi\v{c}ek \cite{Vasicek} and the Cox--Ingersoll--Ross \cite{CIR} models, has made them appealing for term structure modeling. Affine term structure models are based on a (multidimensional) factor process, which in turn describes the evolution of the spot rate and the bank account processes. No-arbitrage arguments then provide the corresponding zero-coupon bond prices, yield curves and forward rates. Prices in these models are calculated under an equivalent martingale measure for known static model parameters. However, model parameters typically vary over time as financial market conditions change. They may, for instance, be of a regime switching nature and need to be permanently re-calibrated to the actual financial market conditions. In practice, this re-calibration is done on a regular basis (as new information becomes available). This implies that model parameters are not static and, henceforth, may also be understood as stochastic processes. The re-calibration should preserve the no-arbitrage condition, which provides side constraints in the re-calibration. The aim of this work is to discuss these side constraints with the help of the discrete-time multifactor Vasi\v cek interest rate model, which is a tractable, but also flexible model. We show that re-calibration under the side constraints naturally leads to Heath--Jarrow--Morton \cite{HJM} models with stochastic parameters, which we call consistent re-calibration (CRC) models \cite{Harms}.

These models are attractive in financial applications for several reasons. In risk management and in the current regulatory framework \cite{BIS}, one needs realistic and tractable models of portfolio returns. Our approach provides tractable non-Gaussian models for multi-period returns on bond portfolios. Moreover, stress tests for risk management purposes can be implemented efficiently in our framework by selecting suitable models for the parameter process. While an in-depth market study of the performance of CRC models remains to be done, we provide in this paper some evidence of improved fits.

The paper is organized as follows. In Section~\ref{sec: hwe}, we introduce Hull--White extended discrete-time multifactor Vasi\v cek models, which are the building blocks for CRC in this work. We define CRC of the Hull--White extended multifactor Vasi\v cek model in Section~\ref{sec: crc}. Section~\ref{sec: real world dynamics} specifies the market price of risk assumptions used to model the factor process under the real-world probability measure and the equivalent martingale measure, respectively. In Section~\ref{sec: parameters}, we deal with parameter estimation from market data. In Section~\ref{sec: numerical example}, we fit the model to Swiss interest rate data, and in Section~\ref{sec: conclusion}, we conclude. All proofs are presented in Appendix~\ref{sec: proofs}.

\section[Discrete-Time Multifactor Vasicek Model and Hull--White Extension]{Discrete-Time Multifactor Vasi\v{c}ek Model and Hull--White Extension} \label{sec: hwe}
\subsection{Setup and Notation}

Choose a fixed grid size $\Delta>0$ and consider the discrete-time grid $\{0,\Delta, 2\Delta, 3\Delta, \ldots \} = \N_0\Delta$. For~example, a daily grid corresponds to $\Delta=1/252$ if there are 252 business days per year. Choose a (sufficiently rich) filtered probability space
$(\Omega,{\cal F},\F,\p^\ast)$ with discrete-time filtration $\F=({\cal F}(t))_{t\in\N_0}$, where $t\in \N_0$ refers to time point $t\Delta$. Assume that $\p^\ast$ denotes an equivalent martingale measure for a (strictly positive) bank account numeraire $(B(t))_{t\in\N_0}$. $B(t)$ denotes the value at time $t\Delta$ of an investment of one unit of currency at Time $0$ into the bank account (i.e., the risk-free rollover relative to $\Delta$).

We use the following notation. Subscript indices refer to elements of vectors and matrices. Argument indices refer to time points. We denote the $n\times n$ identity matrix by $\mathds{1}\in\R^{n\times n}$. We also introduce the vectors $\boldsymbol 1=\left(1,\ldots,1\right)^\top\in\R^n$ and $\boldsymbol e_1=(1,0,\ldots,0)^\top\in\R^n$.

\subsection[Discrete-Time Multifactor Vasicek Model]{Discrete-Time Multifactor Vasi\v{c}ek Model}
\label{discrete-time one-factor Vasicek model}
We choose $n\in\N$ fixed and introduce the $n$-dimensional $\F$-adapted factor process:
\[
\boldsymbol X=\left(\boldsymbol X(t)\right)_{t\in\mathbb N_0}=\left(X_1(t),\ldots,X_n(t)\right)^\top_{t\in\mathbb N_0},
\] 
which generates the spot rate and bank account processes as follows: 
\begin{equation}\label{eq: spot rate}
r(t)=\boldsymbol 1^\top\boldsymbol X(t)\quad\text{and}\quad B(t)=\exp\left\{\Delta\sum_{s=0}^{t-1}r(s)\right\},
\end{equation}
where $t\in\mathbb N_0$; empty sums are set equal to zero. The factor process $\boldsymbol X$ is assumed to evolve under $\p^\ast$ according to: 
\begin{equation}\label{eq: ARn}
\boldsymbol X(t)=\boldsymbol b+\beta\boldsymbol X(t-1)+\Sigma^{\frac{1}{2}}\boldsymbol\varepsilon^\ast(t),\quad t>0,
\end{equation}
with initial factor $\boldsymbol X(0)\in\R^n$, $\boldsymbol b\in\mathbb R^n$, $\beta\in\mathbb R^{n\times n}$, $\Sigma^{\frac{1}{2}}\in\mathbb R^{n\times n}$ and $(\boldsymbol \varepsilon^\ast(t))_{t\in\mathbb N}=(\varepsilon_1^\ast(t),\ldots,\varepsilon_n^\ast(t))^\top_{t\in\mathbb N}$ being $\mathbb F$-adapted. The following assumptions are in place throughout the paper. 
\begin{Assumption}\label{assumption}
We assume that the spectrum of matrix $\beta$ is a subset of $(-1,1)^n$ and that matrix $\Sigma^{\frac{1}{2}}$ is non-singular. Moreover, for each $t\in\N$, we assume that $\boldsymbol\varepsilon^\ast(t)$ is independent of $\mathcal F(t-1)$ under $\p^\ast$ and has standard normal distribution $\boldsymbol \varepsilon^\ast(t)\stackrel{\mathbb P^\ast}{\sim}\mathcal N(\boldsymbol 0, \mathds{1})$.
\end{Assumption}

\begin{Remark*}
In Assumption~\ref{assumption}, the condition on matrix $\beta$ ensures that $\mathds{1}-\beta$ is invertible and that the geometric series generated by $\beta$ converges. The condition on $\Sigma^{\frac{1}{2}}$ ensures that $\Sigma=\Sigma^{\frac{1}{2}}(\Sigma^{\frac{1}{2}})^\top$ is symmetric positive definite. Under Assumption~\ref{assumption}, Equation \eqref{eq: ARn} defines a stationary process; see \cite{TS}, Section 11.3. 	
\end{Remark*}

The model defined by Equations \eqref{eq: spot rate} and \eqref{eq: ARn} is called the discrete-time multifactor Vasi\v cek model. Under the above model assumptions, we have for $m>t$:
\begin{equation}\label{eq: AR conditional distribution}
\boldsymbol X(m)|\mathcal F(t)\stackrel{\mathbb P^\ast}{\sim}\mathcal N\left(\left(\mathds{1}-\beta\right)^{-1}\left(\mathds{1}-\beta^{m-t}\right)\boldsymbol b+\beta^{m-t}\boldsymbol X(t),\sum_{s=0}^{m-t-1}\beta^s\Sigma(\beta^\top)^s\right).
\end{equation}

\begin{Remark*} 
For $m>t$, the conditional distribution of $\boldsymbol X(m)$, given ${\cal F}(t)$, depends only on the value $\boldsymbol X(t)$ at time $t\Delta$ and on lag $m-t$. In other words, the factor process \eqref{eq: ARn} is a time-homogeneous Markov~process.
\end{Remark*}
 
At time $t\Delta$, the price of the zero-coupon bond (ZCB) with maturity date $m \Delta >t\Delta$ with respect to filtration $\F$ and equivalent martingale measure $\p^\ast$ is given by:
\begin{equation*}
P(t,m)=\E^\ast \left[\left.\frac{B(t)}{B(m)}\right|{\cal F}(t)\right]
=\E^\ast \left[\left. \exp \left\{-\Delta\sum_{s=t}^{m-1}\boldsymbol1^\top\boldsymbol X(s) \right\}\right|{\cal F}(t)\right].
\end{equation*}

For the proof of the following result, see Appendix \ref{sec: proofs}. 
\begin{Theorem}\label{theo: ARn prices}
The ZCB prices in the discrete-time multifactor Vasi\v cek Models \eqref{eq: spot rate} and \eqref{eq: ARn} with respect to filtration $\F$ and equivalent martingale measure $\p^\ast$ have an affine term structure:
\[
P(t,m)=e^{A(t,m)-\boldsymbol B(t,m)^\top\boldsymbol X(t)},\quad m>t,
\]
with $A(m-1,m)=0$, $\boldsymbol B(m-1,m)=\boldsymbol 1\Delta$ and for $m-1>t\geq0$:
\[
\begin{aligned}
A(t,m)&=A(t+1,m)-\boldsymbol B(t+1,m)^\top\boldsymbol b+\frac{1}{2}\boldsymbol B(t+1,m)^\top\Sigma\boldsymbol B(t+1,m),\\
\boldsymbol B(t,m)&=\left(\mathds{1}-\beta^\top\right)^{-1}\left(\mathds{1}-(\beta^\top)^{m-t}\right)\boldsymbol 1\Delta.
\end{aligned}
\]
\end{Theorem}
In the discrete-time multifactor Vasi\v{c}ek Models \eqref{eq: spot rate} and \eqref{eq: ARn}, the term structure of interest rates (yield curve) takes the following form at time $t\Delta$ for maturity dates $m\Delta>t\Delta$:
\begin{equation}\label{eq: ARn yields}
Y(t,m)=-\frac{1}{(m-t)\Delta}\log P(t,m)
=-\frac{A(t,m)}{(m-t)\Delta}+\frac{\boldsymbol B(t,m)^\top\boldsymbol X(t)}{(m-t)\Delta},
\end{equation}
with the spot rate at time $t\Delta$ given by $Y(t,t+1)=\boldsymbol 1^\top\boldsymbol X(t)=r(t)$.

\subsection[Hull--White Extended Discrete-Time Multifactor Vasicek Model]{Hull--White Extended Discrete-Time Multifactor Vasi\v{c}ek Model}
The possible shapes of the Vasi\v{c}ek yield curve \eqref{eq: ARn yields} are restricted by the choice of the parameters $\boldsymbol b\in\R^n$, $\beta\in\R^{n\times n}$ and $\Sigma\in\R^{n\times n}$. These parameters are not sufficiently flexible to exactly calibrate the model to an arbitrary observed initial yield curve. Therefore, we consider the Hull--White extended version (see \cite{HW1994}) of the discrete-time multifactor Vasi\v{c}ek model. We replace the factor process defined in \eqref{eq: ARn} as follows. For fixed $k\in\N_0$, let $\boldsymbol X^{(k)}$ satisfy: 
\begin{equation}\label{eq: ARn+}
\boldsymbol X^{(k)}(t)=\boldsymbol b+\theta(t-k)\boldsymbol e_1+\beta\boldsymbol X^{(k)}(t-1)+\Sigma^{\frac{1}{2}}\boldsymbol\varepsilon^\ast(t),\quad t>k,
\end{equation}
with starting factor $\boldsymbol X^{(k)}(k)\in\R^n$, $\boldsymbol e_1=(1,0,\ldots,0)^\top\in\R^n$ and function $\theta:\N\rightarrow\R$. Model assumption~\eqref{eq: ARn+} corresponds to \eqref{eq: ARn}, where the first component of $\boldsymbol b$ is replaced by the time-dependent coefficient $(b_1+\theta(i))_{i\in\N}$ and all other terms ceteris paribus. Without loss of generality, we choose the first component for this replacement. Note that parameter $b_1$ is redundant in this model specification, but for didactical reasons, it is used below. The time-dependent coefficient $\theta$ is called the \emph{Hull--White extension}, and it is used to calibrate the model to a given yield curve at a given time point $k\Delta$. The upper index $^{(k)}$ denotes that time point and corresponds to the time shift we apply to the Hull--White extension $\theta$ in Model~\eqref{eq: ARn+}. The factor process $\boldsymbol X^{(k)}$ generates the spot rate process and the bank account process as in \eqref{eq: spot rate}.

The model defined by (\ref{eq: spot rate}, \ref{eq: ARn+}) is called the Hull--White extended discrete-time multifactor Vasi\v cek model. Under these model assumptions, we have for $m>t\geq k$:
\[
\boldsymbol X^{(k)}(m)|\mathcal F(t)\stackrel{\mathbb P^\ast}{\sim}\mathcal N\left(\sum_{s=0}^{m-t-1}\beta^s\left(\boldsymbol b+\theta(m-s-k)\boldsymbol e_1\right)+\beta^{m-t}\boldsymbol X^{(k)}(t),\sum_{s=0}^{m-t-1}\beta^s\Sigma(\beta^\top)^s\right).
\]

\begin{Remark*}
For $m>t\geq k$, the conditional distribution of $\boldsymbol X^{(k)}(m)$, given ${\cal F}(t)$, depends only on the factor $\boldsymbol X^{(k)}(t)$ at time $t\Delta$. In this case, factor process \eqref{eq: ARn+} is a time-inhomogeneous Markov process. Note that the upper index $^{(k)}$ in the notation is important since the conditional distribution depends explicitly on the lag $m-k$.
\end{Remark*}

\begin{Theorem}\label{theo: ARn+ prices}
The ZCB prices in the Hull--White extended discrete-time multifactor Vasi\v cek model (\ref{eq: spot rate}, \ref{eq: ARn+}) with respect to filtration $\mathbb F$ and equivalent martingale measure $\p^\ast$ have affine term structure:
\[
P^{(k)}(t,m)=e^{A^{(k)}(t,m)-\boldsymbol B(t,m)^\top\boldsymbol X^{(k)}(t)},\quad m>t\geq k,
\]
with $\boldsymbol B(t,m)$ as in Theorem \ref{theo: ARn prices}, $A^{(k)}(m-1,m)=0$ and for $m-1>t\geq k$:
\[
\begin{aligned}
A^{(k)}(t,m)&=A^{(k)}(t+1,m)-\boldsymbol B(t+1,m)^\top\left(\boldsymbol b+\theta(t+1-k)\boldsymbol e_1\right)\\&\quad+\frac{1}{2}\boldsymbol B(t+1,m)^\top\Sigma\boldsymbol B(t+1,m).
\end{aligned}
\]
\end{Theorem}

In the Hull--White extended discrete-time multifactor Vasi\v{c}ek model (\ref{eq: spot rate}, \ref{eq: ARn+}), the yield curve takes the following form at time $t\Delta$ for maturity dates $ m\Delta>t\Delta\geq k\Delta$:
\begin{equation}\label{eq: ARn+ yields}
Y^{(k)}(t,m)=-\frac{1}{(m-t)\Delta}\log P^{(k)}(t,m)
=-\frac{A^{(k)}(t,m)}{(m-t)\Delta}+\frac{\boldsymbol B(t,m)^\top\boldsymbol X^{(k)}(t)}{(m-t)\Delta},
\end{equation}
with spot rate at time $t\Delta$ given by $Y^{(k)}(t,t+1)=\boldsymbol 1^\top\boldsymbol X^{(k)}(t)$.

\begin{Remark*}
Note that the coefficient $\boldsymbol B(t,m)$ in Theorem \ref{theo: ARn+ prices} is not affected by the Hull--White extension $\theta$ and depends solely on $m-t$, whereas the coefficient $A^{(k)}(t,m)$ depends explicitly on the Hull--White extension $\theta$.	
\end{Remark*}

\subsection{Calibration of the Hull--White Extended Model}
We consider the term structure model defined by the Hull--White extended factor process $\boldsymbol X^{(k)}$ and calibrate the Hull--White extension $\theta\in\R^\N$ to a given yield curve at time point $k\Delta$. We explicitly introduce the time index $k$ in Model~\eqref{eq: ARn+} because the CRC algorithm is a concatenation of multiple Hull--White extended models, which are calibrated at different time points $k\Delta$, see Section~\ref{sec: crc} below. 

Assume that there is a fixed final time to maturity date $M\Delta$ and that we observe at time $k\Delta$ the yield curve $\widehat{\boldsymbol{y}}(k)\in\R^M$ for maturity dates $(k+1)\Delta,\ldots,(k+M)\Delta$. For these maturity dates, the Hull--White extended discrete-time multifactor Vasi\v{c}ek yield curve at time $k\Delta$, given by Theorem \ref{theo: ARn+ prices}, reads as: 
\begin{equation*}
\boldsymbol{y}^{(k)}(k)=
\left(
-\frac{1}{i\Delta}A^{(k)}(k,k+i)+\frac{1}{i\Delta}\boldsymbol B(k,k+i)^\top\boldsymbol X^{(k)}(k)\right)^\top_{i=1,\ldots,M} \in \R^M.
\end{equation*}

For given starting factor $\boldsymbol X^{(k)}(k)\in\mathbb R^n$ and parameters $\boldsymbol b\in\mathbb R^n$,
$\beta\in\mathbb R^{n\times n}$ and $\Sigma\in\mathbb R^{n\times n}$, our aim is to choose the Hull--White extension $\theta\in\mathbb R^\N$ such that we get an exact fit at time $k\Delta$ to the yield curve $\widehat{\boldsymbol{y}}(k)$, that is,
\begin{equation}\label{eq: calibration theta}
\boldsymbol{y}^{(k)}(k)=\widehat{\boldsymbol{y}}(k).
\end{equation}

The following theorem provides an equivalent condition to \eqref{eq: calibration theta}, which allows one to calculate the Hull--White extension $\theta\in\R^\N$ explicitly. 
\begin{Theorem}\label{theo: calibration}
Denote by $\boldsymbol{y}^{(k)}(k)$ the yield curve at time $k\Delta$ obtained from the Hull--White extended discrete-time multifactor Vasi\v{c}ek Model (\ref{eq: spot rate}, \ref{eq: ARn+}) for given starting factor $\boldsymbol X^{(k)}(k)=\boldsymbol x\in\mathbb R^n$, parameters $\boldsymbol b\in\mathbb R^n$, $\beta\in\mathbb R^{n\times n}$ and $\Sigma\in\mathbb R^{n\times n}$ and Hull--White extension $\theta\in\R^\N$. For given $\boldsymbol y\in\R^M$, identity $\boldsymbol{y}^{(k)}(k)=\boldsymbol y$ holds if and only if the Hull--White extension $\boldsymbol\theta$ fulfills: 
\begin{equation}\label{eq: calibration matrix}
\boldsymbol{\theta}={\cal C}(\beta)^{-1} \boldsymbol{z}\left(\boldsymbol b,\beta,\Sigma, \boldsymbol x,\boldsymbol y\right),
\end{equation}
where $\boldsymbol\theta=(\theta_i)_{i=1,\ldots,M-1}^\top\in\R^{M-1}$, ${\cal C}(\beta)=\left({\cal C}_{ij}(\beta)\right)_{i,j=1,\ldots,M-1}\in\R^{(M-1)\times(M-1)}$ and \\$\boldsymbol{z}\left(\boldsymbol b,\beta,\Sigma, \boldsymbol x,\boldsymbol y\right)=\left(z_i\left(\boldsymbol b,\beta,\Sigma, \boldsymbol x,\boldsymbol y\right)\right)_{i=1,\ldots,M-1}^\top\in\R^{M-1}$ are defined by:
\[
\begin{aligned}
\theta_i&=\theta(i),\\
{\cal C}_{ij}(\beta)&=B_1(k+j, k+i+1)~1_{\{ j \le i \}},\\
z_i\left(\boldsymbol b,\beta,\Sigma, \boldsymbol x,\boldsymbol y\right)&=\sum_{s=k+1}^{k+i}\left(\frac{1}{2}\boldsymbol B(s,k+i+1)^\top\Sigma\boldsymbol B(s,k+i+1)-\boldsymbol B(s,k+i+1)^\top\boldsymbol b\right)\\&\qquad\qquad\qquad-\boldsymbol 1^\top\left(\mathds{1}-\beta^{i+1}\right)\left(\mathds{1}-\beta\right)^{-1}\boldsymbol x\Delta+(i+1) y_{i+1}(k)\Delta,
\end{aligned}
\]
with $i,j=1,\ldots,M-1$ and $\boldsymbol B(\cdot,\cdot)=\left(B_1(\cdot,\cdot),\ldots,B_n(\cdot,\cdot)\right)^\top$ given by Theorem \ref{theo: ARn prices}.
\end{Theorem}
Theorem \ref{theo: calibration} shows that the Hull--White extension can be calculated by inverting the \mbox{$(M-1)\times(M-1)$} lower triangular positive definite matrix ${\cal C}(\beta)$. 

\section{Consistent Re-Calibration}\label{sec: crc}
The crucial extension now is the following: we let parameters $\boldsymbol b$, $\beta$ and $\Sigma$ vary over time, and we re-calibrate the Hull--White extension in a consistent way at each time point, that is according to the actual choice of the parameter values using Theorem \ref{theo: calibration}. Below, we show that this naturally leads to a Heath--Jarrow--Morton \cite{HJM} (HJM) approach to term structure modeling. 

\subsection{Consistent Re-Calibration Algorithm} \label{Re-calibration algorithm} 
Assume that $\left(\boldsymbol b(k)\right)_{k\in\mathbb N_0}$, $(\beta(k))_{k\in \N_0}$ and $(\Sigma(k))_{k\in \N_0}$ are $\F$-adapted parameter processes with $\beta(k)$ and $\Sigma(k)$ satisfying Assumption \ref{assumption}, $\p^\ast$-a.s., for all $k\in\N_0$. Based on these parameter processes, we define the $n$-dimensional $\F$-adapted CRC factor process $\boldsymbol{\mathcal X}$, which evolves according to Steps (i)--(iv) of the CRC algorithm described below. Thus, factor process $\boldsymbol{\mathcal X}$ will define a spot rate model similar to \eqref{eq: spot rate}. 

 In the CRC algorithm, Steps \ref{subsubsec crc step 1}--\ref{subsubsec crc step 3} below are executed iteratively.

\subsubsection[Initialization k=0]{Initialization $k=0$} \label{subsubsec crc step 1}

Assume that the initial yield curve observation at Time 0 is given by $\widehat{\boldsymbol{y}}(0)\in\R^{M}$. Let $\theta^{(0)}\in\R^\N$ be an ${\cal F}(0)$-measurable Hull--White extension, such that condition \eqref{eq: calibration theta} is satisfied at Time $0$ for initial factor $\boldsymbol{\mathcal X}(0)\in\R^n$ and parameters $\boldsymbol b(0)$, $\beta(0)$ and $\Sigma(0)$. By Theorem \ref{theo: calibration}, the values $\boldsymbol \theta^{(0)}=(\theta^{(0)}(i))_{i=1,\ldots,M-1}\in\R^{M-1}$ are given by:
\begin{equation*}
\boldsymbol{\theta}^{(0)} = 
{\cal C}\left(\beta(0)\right)^{-1} \boldsymbol{z}\left(\boldsymbol b(0),\beta(0),\Sigma(0),\boldsymbol{\cal X}(0),\widehat{\boldsymbol{y}}(0)\right).
\end{equation*}

This provides Hull--White extended Vasi\v{c}ek yield curve $\boldsymbol y^{(0)}(0)$ identically equal to $\widehat{\boldsymbol{y}}(0)$ for given initial factor $\boldsymbol{\cal X}(0)$ and parameters $\boldsymbol b(0)$, $\beta(0)$, $\Sigma(0)$.

\subsubsection[Increments of the Factor Process from k to k+1]{Increments of the Factor Process from $k \to k+1$}\label{subsubsec crc step 2} 

Assume factor $\boldsymbol{\cal X}(k)$, parameters $\boldsymbol b(k),\beta(k)$ and $\Sigma(k)$ and Hull--White extension $\theta^{(k)}$ are given. Define the Hull--White extended model $\boldsymbol X^{(k)}=(\boldsymbol X^{(k)}(t))_{t\geq k}$ by:
\begin{equation}\label{re-calibration step 1}
\boldsymbol X^{(k)}(t)=\boldsymbol b(k)+\theta^{(k)}(t-k)\boldsymbol e_1+\beta(k)\boldsymbol X^{(k)}(t-1) + \Sigma(k)\boldsymbol\varepsilon^\ast(t),\quad t>k,
\end{equation}
with starting value $\boldsymbol X^{(k)}(k)=\boldsymbol{\cal X}(k)$, ${\cal F}(k)$-measurable parameters $\boldsymbol b(k)$, $\beta(k)$ and $\Sigma(k)$ and Hull--White extension $\theta^{(k)}$. We update the factor process $\boldsymbol{\cal X}$ at time $(k+1)\Delta$ according to the $\boldsymbol X^{(k)}$-dynamics, that is, we set:
\[
\boldsymbol{\cal X}(k+1)=\boldsymbol X^{(k)}(k+1).
\]

This provides ${\cal F}(k+1)$-measurable yield curve at time $(k+1)\Delta$ for maturity dates \mbox{$ m\Delta>(k+1)\Delta$:}
\begin{equation*}
Y^{(k)}(k+1,m)=-\frac{A^{(k)}(k+1,m)}{(m-(k+1))\Delta}+\frac{\boldsymbol B^{(k)}(k+1,m)^\top\boldsymbol{\cal X}(k+1)}{(m-(k+1))\Delta},
\end{equation*}
with $A^{(k)}(m-1,m)=0$ and $\boldsymbol B^{(k)}(m-1,m)=\Delta\boldsymbol 1$, and recursively for $m-1>t\geq k$:
\[
\begin{aligned}
A^{(k)}(t,m)&=A^{(k)}(t+1,m)-\boldsymbol B^{(k)}(t+1,m)^\top\left(\boldsymbol b(k)+\theta^{(k)}(t+1-k)\boldsymbol e_1\right)\\&\quad+\frac{1}{2}\boldsymbol B^{(k)}(t+1,m)^\top\Sigma(k)\boldsymbol B^{(k)}(t+1,m),\\
\boldsymbol B^{(k)}(t,m)&=\left(\mathds{1}-\beta(k)^\top\right)^{-1}\left(\mathds{1}-(\beta(k)^\top)^{m-t}\right)\boldsymbol 1\Delta.
\end{aligned}
\]

This is exactly the no-arbitrage price under $\p^\ast$ if the parameters $\boldsymbol b(k)$, $\beta(k)$ and $\Sigma(k)$ and the Hull--White extension $\theta^{(k)}$ remain constant for all $t>k$.

\subsubsection[Parameter Update and Re-Calibration at k+1]{Parameter Update and Re-Calibration at $k+1$}\label{subsubsec crc step 3} 

 Assume that at time $(k+1)\Delta$, the parameters $(\boldsymbol b(k),\beta(k),\Sigma(k))$ are updated to $(\boldsymbol b(k+1),\beta(k+1),\Sigma(k+1))$. We may think of this parameter update as a consequence of model selection after we observe a new yield curve at time $(k+1)\Delta$. This is discussed in more detail in Section~\ref{sec: parameters} below.
The no-arbitrage yield curve at time $(k+1)\Delta$ from the model with parameters $(\boldsymbol b(k),\beta(k),\Sigma(k))$ and Hull--White extension $\theta^{(k)}$ is given by: 
\[
\boldsymbol{y}^{(k)}(k+1)=\left(Y^{(k)}(k+1,k+2),\ldots,Y^{(k)}(k+1,k+1+M)\right)^\top\in\R^M.
\]

The parameter update $\left(\boldsymbol b(k), \beta(k),\Sigma(k)\right) \mapsto \left(\boldsymbol b(k+1), \beta(k+1),\Sigma(k+1)\right)$ requires re-calibration of the Hull--White extension, otherwise arbitrage is introduced into the model. This re-calibration provides ${\cal F}(k+1)$-measurable Hull--White extension $\theta^{(k+1)}\in\R^\N$ at time $(k+1)\Delta$. The values $\boldsymbol \theta^{(k+1)}=(\theta^{(k+1)}(i))_{i=1,\ldots,M-1}\in\R^{M-1}$ are given by (see Theorem \ref{theo: calibration}):
\begin{equation}\label{eq: re-calibration step 2}
\boldsymbol\theta^{(k+1)}={\cal C}\left(\beta(k+1)\right)^{-1} \boldsymbol{z}\left(\boldsymbol b(k+1),\beta(k+1),\Sigma(k+1),\boldsymbol{\cal X}(k+1),\boldsymbol{y}^{(k)}(k+1)\right),
\end{equation}
and the resulting yield curve $\boldsymbol{y}^{(k+1)}(k+1)$ under the updated parameters is identically equal to $\boldsymbol{y}^{(k)}(k+1)$. Note that this CRC makes the upper index $(k)$ in the yield curve superfluous, because the Hull--White extension is re-calibrated to the new parameters, such that the resulting yield curve remains unchanged. Therefore, we write ${\cal Y}(k,\cdot)$ in the sequel for the CRC yield curve with factor $\boldsymbol{\cal X}(k)$, parameters $\boldsymbol b(k),\beta(k),\Sigma(k)$ and Hull--White extension $\theta^{(k)}$.

(End of algorithm.)

\begin{Remark*}
For the implementation of the above algorithm, we need to consider the following issue. Assume we start the algorithm at Time $0$ with initial yield curve $\widehat{\boldsymbol{y}}(0)\in\R^M$. At times $ k\Delta$, for $k>0$, calibration of $\boldsymbol\theta^{(k)}\in\R^{M-1}$ requires yields with times to maturity beyond $ M\Delta$. Either yields for these times to maturity are observable, and the length of $\boldsymbol{\theta}^{(k)}$ is reduced in every step of the CRC algorithm or an appropriate extrapolation method beyond the latest available maturity date is applied in every~step.	
\end{Remark*}

\subsection{Heath--Jarrow--Morton Representation}
We analyze the yield curve dynamics $({\cal Y}(k,\cdot))_{k\in\N_0}$
obtained by the CRC algorithm of Section~\ref{Re-calibration algorithm}.
Due to re-calibration \eqref{eq: re-calibration step 2},
the yield curve fulfills the following identity for $m>k+1$:
\begin{equation}\label{eq: crc yields}
\begin{aligned}
{\cal Y}(k+1,m)&=-\frac{A^{(k)}(k+1,m)}{(m-(k+1))\Delta}+\frac{\boldsymbol B^{(k)}(k+1,m)^\top\boldsymbol{\cal X}(k+1)}{(m-(k+1))\Delta}\\&=-\frac{A^{(k+1)}(k+1,m)}{(m-(k+1))\Delta}+\frac{\boldsymbol B^{(k+1)}(k+1,m)^\top\boldsymbol{\cal X}(k+1)}{(m-(k+1))\Delta},
\end{aligned}
\end{equation}
where the first line is based on the ${\cal F}(k)$-measurable parameters
$(\boldsymbol b(k),\beta(k),\Sigma(k))$ and Hull--White extension $\theta^{(k)}$, and the second line is based on the ${\cal F}(k+1)$-measurable parameters and Hull--White extension $(\boldsymbol b(k+1),\beta(k+1),\Sigma(k+1),\theta^{(k+1)})$ after CRC Step (iii). Note that in the re-calibration only $(\boldsymbol b(k+1),\beta(k+1),\Sigma(k+1))$ can be chosen exogenously, and the Hull--White extension $\theta^{(k+1)}$ is used for consistency property \eqref{eq: re-calibration step 2}. Our aim is to express ${\cal Y}(k+1,m)$ as a function of $\boldsymbol{\cal X}(k)$ and ${\cal Y}(k,m)$. Using Equations \eqref{re-calibration step 1} and \eqref{eq: crc yields}, we have for $m>k+1$:
\begin{equation} \label{start HJM}
\begin{aligned}
{\cal Y}(k+1,m)&\left(m-(k+1)\right) \Delta=-A^{(k)}(k+1,m)\\&+\boldsymbol B^{(k)}(k+1,m)^\top \left(\boldsymbol b(k)+\theta^{(k)}(1)\boldsymbol e_1+\beta(k)\boldsymbol{\cal X}(k)+\Sigma(k)^{\frac{1}{2}}\boldsymbol\varepsilon^\ast(k+1)
\right).
\end{aligned}
\end{equation}

This provides the following theorem; see Appendix \ref{sec: proofs} for the proof.
\begin{Theorem} \label{theorem HJM view}
Under equivalent martingale measure $\p^\ast$, the yield curve dynamics $({\cal Y}(k,\cdot))_{k\in\N_0}$ obtained by the CRC algorithm of Section~\ref{Re-calibration algorithm} has the following HJM representation for $m>k+1$: 
\[
\begin{aligned}
{\cal Y}(k+1,m)(m-(k+1))\Delta &={\cal Y}(k,m)(m-k)\Delta- {\cal Y}(k,k+1)\Delta\\&\quad+\frac{1}{2}\boldsymbol B^{(k)}(k+1,m)^\top\Sigma(k)\boldsymbol B^{(k)}(k+1,m)\\&\quad+\boldsymbol B^{(k)}(k+1,m)^\top\Sigma(k)^{\frac{1}{2}}\boldsymbol\varepsilon^\ast(k+1),
\end{aligned}
\]
with $\boldsymbol B^{(k)}(k+1,m)=\left(\mathds{1}-\beta^\top(k)\right)^{-1}\left(\mathds{1}-(\beta(k)^\top)^{m-k-1}\right)\boldsymbol 1\Delta$.
\end{Theorem}
\begin{KeyObservation}
Observe that in Theorem \ref{theorem HJM view}, a remarkable simplification happens. Simulating the CRC algorithm \eqref{re-calibration step 1} and \eqref{eq: re-calibration step 2} to future time points $k\Delta>0$ does not require the calculation of the Hull--White extensions $(\theta^{(k)})_{k\in \N_0}$ according to \eqref{eq: re-calibration step 2},
but the knowledge of the parameter process $\left(\boldsymbol b(k),\beta(k), \Sigma(k)\right)_{k\in \N_0}$ is sufficient. The
Hull--White extensions are fully encoded in the yield curve process $({\cal Y}(k,\cdot))_{k\in\mathbb N_0}$, and we can avoid the inversion of (potentially) high dimensional matrices ${\cal C}(\beta(k))_{k\in\mathbb N_0}$.	
\end{KeyObservation}

\begin{FurtherRemarks*}
\begin{itemize}[leftmargin=*,labelsep=4mm]
\item CRC of the multifactor Vasi\v cek spot rate model can be defined directly in the HJM framework assuming a stochastic dynamics for the parameters. However, solely from the HJM representation, one cannot see that the yield curve dynamics is obtained, in our case, by combining well-understood Hull--White extended multifactor Vasi\v cek spot rate models using the CRC algorithm of Section~\ref{sec: crc}; that is, the Hull--White extended multifactor Vasi\v cek model gives an explicit functional form to the HJM representation.
\item The CRC algorithm of Section~\ref{sec: crc} does not rely directly on $(\boldsymbol\varepsilon^\ast(t))_{t\in\N}$ having independent and Gaussian components. The CRC algorithm is feasible as long as explicit formulas for ZCB prices in the Hull--White extended model are available. Therefore, one may replace the Gaussian innovations by other distributional assumptions, such as normal variance mixtures. This replacement is possible provided that conditional exponential moments can be calculated under the new innovation assumption. Under non-Gaussian innovations, it will no longer be the case that the HJM representation does not depend on the Hull--White extension $\theta^{(k)}\in\R^\N$. 
\item Interpretation of the parameter processes will be given in Section~\ref{sec: parameters}, below.
\end{itemize}
\end{FurtherRemarks*}

\section{Real World Dynamics and Market Price of Risk} \label{sec: real world dynamics}

All previous derivations were done under an equivalent martingale measure $\p^\ast$ for the bank account numeraire. In order to statistically estimate parameters from market data, we need to specify a Girsanov transformation to the real-world measure, which is denoted by $\mathbb P$. We present a specific change of measure, which provides tractable spot rate dynamics under $\mathbb P$. Assume that $(\boldsymbol\lambda(k))_{k\in\mathbb N_0}$ and $(\Lambda(k))_{k\in\mathbb N_0}$ are $\mathbb R^n$- and $\mathbb R^{n\times n}$-valued $\mathbb F$-adapted processes, respectively. Let $(\boldsymbol{\cal X}(k))_{k\in\mathbb N_0}$ be the factor process obtained by the CRC algorithm of Section~\ref{Re-calibration algorithm}. Then, we assume that the $n$-dimensional $\F$-adapted process $(\boldsymbol{\lambda}(k)+\Lambda(k)\boldsymbol{\cal X}(k))_{k\in \mathbb N_0}$ describes the market price of risk dynamics. We define the following $\p^\ast$-density process:
$(\xi(k))_{k\in \mathbb N_0}$
\[
\xi(k)=\exp\left\{-\frac{1}{2}\sum_{s=0}^{k-1}\left\|\boldsymbol\lambda(s)+\Lambda(s)\boldsymbol {\cal X}(s)\right\|_2^2+\sum_{s=0}^{k-1}\left(\boldsymbol\lambda(s)+\Lambda(s)\boldsymbol {\cal X}(s)\right)^\top \boldsymbol\varepsilon^{*}(s+1)\right\},\quad k \in \mathbb N_0.
\]

The real-world probability measure $\p$ is then defined by the Radon--Nikodym derivative:
\begin{equation}\label{transformation}
\left.\frac{d\p}{d\p^\ast}\right|_{\mathcal F(k)}= \xi(k),\quad k \in \mathbb N_0.
\end{equation}

An immediate consequence is that for $k \in \mathbb N_0$:
\begin{equation*}
\boldsymbol\varepsilon(k+1)=\boldsymbol\lambda(k)+
\Lambda(k)\boldsymbol{\cal X}(k)+\boldsymbol\varepsilon^\ast(k+1),
\end{equation*}
has a standard Gaussian distribution under $\p$, conditionally on ${\cal F}(k)$.
This implies that under the real-world measure $\p$, the factor process
$(\boldsymbol{\cal X}(k))_{k\in \mathbb N_0}$ is described by:
\begin{equation}\label{eq: real world dynamics}
\boldsymbol{\cal X}(k+1)=\boldsymbol a(k)+\alpha(k)\boldsymbol{\cal X}(k)+\Sigma(k)^{\frac{1}{2}}\boldsymbol\varepsilon(k+1),
\end{equation}
where we define:
\begin{equation}\label{eq: real world transform}
\boldsymbol a(k)=\boldsymbol b(k)+\theta^{(k)}(1)\boldsymbol e_1-\Sigma(k)^{\frac{1}{2}}\boldsymbol\lambda(k)
\quad \text{and }\quad \alpha(k)=\beta(k)-\Sigma(k)^{\frac{1}{2}}\Lambda(k).
\end{equation}

As in Assumption~\ref{assumption}, we require $\Lambda(k)$ to be such that the spectrum of $\alpha(k)$ is a subset of $(-1,1)^n$. Formula \eqref{eq: real world dynamics} describes the dynamics of the factor process $(\boldsymbol{\cal X}(k))_{k\in \mathbb N_0}$ obtained by the CRC algorithm of Section~\ref{Re-calibration algorithm} under real-world measure $\p$. The following corollary describes the yield curve dynamics obtained by the CRC algorithm under $\p$, in analogy to Theorem \ref{theorem HJM view}.

\begin{Corollary}\label{HJM under the real world measure}
Under real-world measure $\p$ satisfying \eqref{transformation}, the yield curve dynamics $({\cal Y}(k,\cdot))_{k\in\N_0}$ obtained by the CRC algorithm of Section~\ref{Re-calibration algorithm} has the following HJM representation for $m>k+1$: 
\[
\begin{aligned}
{\cal Y}(k+1,m)
\left(m-(k+1)\right)\Delta &={\cal Y}(k,m)(m-k)\Delta- {\cal Y}(k,k+1)\Delta\\&\quad+\frac{1}{2}\boldsymbol B^{(k)}(k+1,m)^\top\Sigma(k)\boldsymbol B^{(k)}(k+1,m)\\&\quad-\boldsymbol B^{(k)}(k+1,m)^\top\Sigma(k)^{\frac{1}{2}}\boldsymbol\lambda(k)\\&\quad-\boldsymbol B^{(k)}(k+1,m)^\top\Sigma(k)^{\frac{1}{2}}\Lambda(k)\boldsymbol {\cal X}(k)\\&\quad+\boldsymbol B^{(k)}(k+1,m)^\top\Sigma(k)^{\frac{1}{2}}\boldsymbol\varepsilon(k+1),
\end{aligned}
\]
with $\boldsymbol B^{(k)}(k+1,m)=\left(\mathds{1}-\beta(k)^\top\right)^{-1}\left(\mathds{1}-\left(\beta(k)^\top\right)^{m-k-1}\right)\boldsymbol 1\Delta$.
\end{Corollary}
Compared to Theorem~\ref{theorem HJM view}, there are additional drift terms $-\boldsymbol B^{(k)}(k+1,m)^\top\Sigma(k)^{\frac{1}{2}}\boldsymbol\lambda(k)$ and $-\boldsymbol B^{(k)}(k+1,m)^\top\Sigma(k)^{\frac{1}{2}}\Lambda(k)\boldsymbol X(k)$, which
are characterized by the market price of risk parameters $\boldsymbol{\lambda}(k)\in\R^n$ and $\Lambda(k)\in\R^{n\times n}$.

\section{Choice of Parameter Process}\label{sec: parameters}
The yield curve dynamics obtained by the CRC algorithm of Section~\ref{Re-calibration algorithm} require exogenous specification of the parameter process of the multifactor Vasi\v cek Models~\eqref{eq: spot rate} and \eqref{eq: ARn} and the market price of risk process, i.e., we need to model the process:
\begin{equation}\label{eq: parameter process}
\left(\boldsymbol b(t),\beta(t), \Sigma(t), \boldsymbol{\lambda}(t), \Lambda(t)\right)_{t\in \N_0}.
\end{equation} 

By Equation \eqref{re-calibration step 1}, the one-step ahead development of the CRC factor process $\boldsymbol{\mathcal{X}}$ under $\p$ reads as: 
\begin{equation}\label{eq: factor evolution}
\boldsymbol {\cal X}(t+1)=\boldsymbol b(t)+\theta^{(t)}(1)\boldsymbol e_1-\Sigma(t)^{\frac{1}{2}}\boldsymbol\lambda(t)+\left(\beta(t)-\Sigma(t)^{\frac{1}{2}}\Lambda(t)\right)\boldsymbol {\cal X}(t) + \Sigma(t)^{\frac{1}{2}}\boldsymbol\varepsilon(t+1),
\end{equation}
with ${\cal F}(t)$-measurable parameters $\boldsymbol b(t)$, $\beta(t)$ and $\Sigma(t)$ and Hull--White extension $\theta^{(t)}$. Thus, on the one hand, the factor process $(\boldsymbol {\cal X}(t))_{t\in \N_0}$ evolves according to \eqref{eq: factor evolution}, and on the other hand, parameters $(\boldsymbol b(t),\beta(t), \Sigma(t), \boldsymbol{\lambda}(t), \Lambda(t))_{t\in \N_0}$ evolve according to the financial market conditions. Note that the process $(\theta^{(t)})_{t\in\N_0}$ of Hull--White extensions is fully determined through CRC by \eqref{eq: re-calibration step 2}. In order to distinguish the evolutions of $(\boldsymbol {\cal X}(t))_{t\in \N_0}$ and $(\boldsymbol b(t),\beta(t), \Sigma(t), \boldsymbol{\lambda}(t), \Lambda(t))_{t\in \N_0}$, respectively, we assume that process \eqref{eq: parameter process} changes at a slower pace than the factor process, and therefore, parameters can be assumed to be constant over a short time window. This assumption motivates the following approach to specifying a model for process \eqref{eq: parameter process}. For each time point $t\Delta$, we fit multifactor Vasi\v cek Models~\eqref{eq: spot rate} and \eqref{eq: ARn} with fixed parameters $\left(\boldsymbol b, \beta, \Sigma, \boldsymbol\lambda,\Lambda\right)$ on observations from a time window $\{t-K+1,\ldots, t\}$ of length $K$. For estimation, we assume that we have yield curve observations $(\widehat{\boldsymbol y}(k))_{k=t-K+1,\ldots, t}=((\widehat{y}_1(k),\ldots,\widehat{y}_M(k)))_{{k=t-K+1,\ldots, t}}$ for times to maturity $ \tau_1\Delta<\ldots< \tau_M\Delta$. Since yield curves are not necessarily observed on a regular time to the maturity grid, we introduce the indices $\tau_1,\ldots,\tau_M\in\N$ to refer to the available times to maturity. Varying the time of estimation $t\Delta$, we obtain time series for the parameters from historical data. Finally, we fit a stochastic model to these time series. In the following, we discuss the interpretation of the parameters and present two different estimation procedures. The two procedures are combined to obtain a full specification of the model parameters. 

\subsection{Interpretation of Parameters}
\subsubsection{Level and Speed of Mean Reversion}
By Equation \eqref{eq: AR conditional distribution}, we have under $\p^\ast$ for $m>t$:
\[
\begin{aligned}
\mathbb E^{\ast}\left[\boldsymbol X(m)\middle|\mathcal F(t)\right]&=\left(\mathds{1}-\beta\right)^{-1}\left(\mathds{1}-\beta^{m-t}\right)\boldsymbol b+\beta^{m-t}\boldsymbol X(t),\\
\mathbb E^{\ast}\left[r(m)\middle|\mathcal F(t)\right]&=\boldsymbol 1^\top\left(\mathds{1}-\beta\right)^{-1}\left(\mathds{1}-\beta^{m-t}\right)\boldsymbol b+\boldsymbol 1^\top\beta^{m-t}\boldsymbol X(t).
\end{aligned}
\]

Thus, $\beta$ determines the speed at which the factor process $(\boldsymbol X(t))_{t\in \N_0}$ and the spot rate process $(r(t))_{t\in \N_0}$ return to their long-term means:
\[
\lim_{m\to\infty}\mathbb E^{\ast}\left[\boldsymbol X(m)|\mathcal F(t)\right]=\left(\mathds{1}-\beta\right)^{-1}\boldsymbol b\quad\text{and}\quad\lim_{m\to\infty}\mathbb E^{\ast}\left[r(m)|\mathcal F(t)\right]=\boldsymbol 1^\top\left(\mathds{1}-\beta\right)^{-1}\boldsymbol b.
\]

A sensible choice of $(\beta(t))_{t\in\N_0}$ adapts the speed of mean reversion to the prevailing financial market conditions at each time point $t\Delta$.

\subsubsection{Instantaneous Variance} 
By Equation \eqref{eq: AR conditional distribution}, we have under $\p^{\ast}$ for $t>0$:
\[
\mathrm{Cov}^\ast\left[\boldsymbol X(t)\middle|\mathcal F(t-1)\right]=\Sigma,\quad\text{and}\quad\mathrm{Var}^\ast\left[r(t)\middle|\mathcal F(t-1)\right]=\boldsymbol 1^\top\Sigma\boldsymbol 1.
\]

Thus, matrix $\Sigma$ plays the role of the instantaneous covariance matrix of $\boldsymbol X$, and it describes the instantaneous spot rate volatility. 

\subsection{State Space Modeling Approach}\label{sec: kalman MLE}

On each time window, we want to use yield curve observations to estimate the parameters of time-homogeneous Vasi\v cek Models~\eqref{eq: spot rate} and \eqref{eq: ARn}. In general, this model is not able to reproduce the yield curve observations exactly. One reason might be that the data are given in the form of parametrized yield curves, and the parametrization might not be compatible with the Vasi\v cek model. For example, this is the case for the widely-used Svensson family \cite{Svensson}. Another reason might be that yield curve observations do not exactly represent risk-free zero-coupon bonds. 

The discrepancy between the Vasi\v cek model and the yield curve observations can be accounted for by adding a noise term to the Vasi\v cek yield curves. This defines a state space model with the factor process as the hidden state variable. In this state space model, the parameters of the factor dynamics can be estimated using Kalman filter techniques in conjunction with maximum likelihood estimation (\cite{Wuethrich} Section 3.6.3). This is explained in detail in Sections~\ref{subsubsec: kalman MLE transition}--\ref{subsubsec: kalman MLE likelihood} below.

\subsubsection{Transition System}\label{subsubsec: kalman MLE transition} The evolution of the unobservable process $\boldsymbol X$ under $\p$ is assumed to be given on time window $\{t-K+1,\ldots,t\}$ by:
\[
\boldsymbol X(k)=\boldsymbol a+\alpha \boldsymbol X(k-1)+\Sigma^{\frac{1}{2}}\boldsymbol \varepsilon(k),\quad k\in\{t-K+1,\ldots,t\},
\]
with initial factor $\boldsymbol X(t-K)\in\R^n$ and parameters $\boldsymbol a=\boldsymbol b-\Sigma^{\frac{1}{2}}\boldsymbol\lambda$ and $\alpha=\beta-\Sigma^{\frac{1}{2}}\Lambda$. The initial factor $\boldsymbol X(t-K)$ is updated according to the output of the Kalman filter for the previous time window $\{t-K,\ldots,t-1\}$. The initial factor is set to zero for the first time window available.

\begin{Remark*}
Parameters $(\boldsymbol b,\beta,\Sigma,\boldsymbol\lambda,\Lambda)$ are assumed to be constant over the time window \{$t-K+1$,$\ldots$,$t$\}. Thus, we drop the index $k$ compared to Equations \eqref{eq: real world dynamics} and \eqref{eq: real world transform}. For estimation, we assume that the factor process evolves according to the time-homogeneous multifactor Vasi\v cek Models~\eqref{eq: spot rate} and \eqref{eq: ARn} in that time window. The Hull--White extension is calibrated to the yield curve at time $t\Delta$ given the estimated parameter values of the time-homogeneous model. 
\end{Remark*}

\subsubsection{Measurement System} 
We assume that the observations in the state space model are given by:
\begin{equation}\label{eq: noisy measurement}
\widehat{\boldsymbol Y}(k)=\boldsymbol d+D\boldsymbol X(k)+S^{\frac{1}{2}}\boldsymbol\eta(k),\quad k\in\{t-K,\ldots,t\},
\end{equation}
where:
\[
\begin{aligned}
\widehat{\boldsymbol Y}(k)&=\left(\widehat Y(k,k+\tau_1),\ldots,\widehat Y(k,k+\tau_M)\right)^\top\in\R^M,\\
\boldsymbol d&=\left(-( \tau_1\Delta)^{-1}A(k,k+\tau_1),\ldots,-( \tau_M\Delta)^{-1}A(k,k+\tau_M)\right)^\top\in\R^M,\\
D_{ij}&=( \tau_i\Delta)^{-1}B_j(k,k+\tau_i),\quad1\leq i\leq M,\quad1\leq j\leq n,
\end{aligned}
\]
with $A(\cdot,\cdot)$ and $\boldsymbol B(\cdot,\cdot)=(B_1(\cdot,\cdot),\dots,B_n(\cdot,\cdot))^\top$ given by Theorem \ref{theo: ARn prices} and $M$-dimensional $\mathcal{F}(k)$-measurable noise term $S^{\frac{1}{2}}\boldsymbol\eta(k)$ for non-singular $S^{\frac{1}{2}}\in\mathbb R^{M\times M}$. We assume that $\boldsymbol\eta(k)$ is independent of $\mathcal{F}(k-1)$ and $\boldsymbol\varepsilon(k)$ under $\p$ and that $\boldsymbol \eta(k)\stackrel{\mathbb P}{\sim}\mathcal N(\boldsymbol 0, \mathds{1})$. The error term $S^{\frac12}\boldsymbol{\eta}$ describes the discrepancy between the yield curve observations and the model. For $S=0$, we would obtain a yield curve in \eqref{eq: noisy measurement} that corresponds exactly to the multifactor Vasi\v cek one. 

Given the parameter and market price of risk value $(\boldsymbol b, \beta, \Sigma, \boldsymbol\lambda, \Lambda)$, we estimate the factor using the following iterative procedure. For each fixed value of $k\in\{t-K,\ldots,t\}$ and fixed time $t$, we consider the $\sigma$-field \linebreak$\mathcal F^{\widehat{\boldsymbol Y}}(k)=\sigma\left(\widehat{\boldsymbol Y}(s)\:\middle|\:t-K\leq s\leq k\right)$ and describe the estimation procedure in this state space~model.

\subsubsection{Anchoring} Fix initial factor $\boldsymbol X(t-K)=\boldsymbol x(t-K|t-K-1)$, and initialize:
\[
\begin{aligned}
\boldsymbol x(t-K+1|t-K)&=\mathbb E\left[\boldsymbol X(t-K+1)\middle|\mathcal F^{\widehat{\boldsymbol Y}}(t-K)\right]=\boldsymbol a+\alpha\boldsymbol x(t-K|t-K-1),\\
\Sigma(t-K+1|t-K)&=\mathrm{Cov}\left(\boldsymbol X(t-K+1)\middle|\mathcal F^{\widehat{\boldsymbol Y}}(t-K)\right)=\Sigma.
\end{aligned}
\]

\subsubsection{Forecasting the Measurement System} At time $k\in\{t-K+1,\ldots,t\}$, we have:
\[
\begin{aligned}
\boldsymbol y(k|k-1)&=\mathbb E\left[\widehat{\boldsymbol Y}(k)\middle|\mathcal F^{\widehat{\boldsymbol Y}}(k-1)\right]=\boldsymbol d+D\boldsymbol x(k|k-1),\\
F(k)&=\mathrm{Cov}\left(\widehat{\boldsymbol Y}(k)\middle|\mathcal F^{\widehat{\boldsymbol Y}}(k-1)\right)=D\Sigma(k|k-1)D^\top+S,\\
\boldsymbol \zeta(k)&=\widehat{\boldsymbol y}(k)-\boldsymbol y(k|k-1).
\end{aligned}
\]

\subsubsection{Bayesian Inference in the Transition System} The prediction error $\boldsymbol \zeta(k)$ is used to update the unobservable factors.
\[
\begin{aligned}
\boldsymbol x(k|k)&=\mathbb E\left[\boldsymbol X(k)\middle|\mathcal F^{\widehat{\boldsymbol Y}}(k)\right]=\boldsymbol x(k|k-1)+K(k)\boldsymbol\zeta(k),\\
\Sigma(k|k)&=\mathrm{Cov}\left(\boldsymbol X(k)\middle|\mathcal F^{\widehat{\boldsymbol Y}}(k)\right)=\left(\mathds{1}-K(k)D\right)\Sigma(k|k-1),
\end{aligned}
\] 
where $K(k)$ denotes the Kalman gain matrix given by:
\[
K(k)=\mathrm{Cov}\left(\boldsymbol X(k)\middle|\mathcal F^{\widehat{\boldsymbol Y}}(k-1)\right)D^\top\mathrm{Cov}\left(\widehat{\boldsymbol Y}(k)\middle|\mathcal F^{\widehat{\boldsymbol Y}}(k-1)\right)^{-1}=\Sigma(k|k-1)D^\top F(k)^{-1}.
\]

\subsubsection{Forecasting the Transition System} For the unobservable factor process, we have the following forecast: 
\[
\begin{aligned}
\boldsymbol x(k+1|k)&=\mathbb E\left[\boldsymbol X(k+1)\middle|\mathcal F^{\widehat{\boldsymbol Y}}(k)\right]=\boldsymbol a + \alpha\boldsymbol x(k|k),\\
\Sigma(k+1|k)&=\mathrm{Cov}\left(\boldsymbol X(k+1)\middle|\mathcal F^{\widehat{\boldsymbol Y}}(k)\right)=\alpha\Sigma(k|k)\alpha^\top+\Sigma.
\end{aligned}
\]

\subsubsection{Likelihood Function}\label{subsubsec: kalman MLE likelihood}
The Kalman filter procedure above allows one to infer factors $\boldsymbol X$ given the parameter and market price of risk values. Of course, in this section, we are interested in estimating these values in the first~place. For this purpose, the procedure above can be used in conjunction with maximum likelihood estimation. For the underlying parameters $\boldsymbol\Theta=\left(\boldsymbol b,\beta,\Sigma,\boldsymbol a,\alpha\right)$, we have the following likelihood function given the observations $(\widehat{\boldsymbol y}(k))_{k=t-K+1,\ldots, t}$:
\begin{equation} \label{eq: likelihood}
\mathcal L_t(\boldsymbol\Theta)=\prod_{k=t-K+1}^{t}\frac{\exp\left(-\frac{1}{2}\boldsymbol\zeta(k)^\top F(k)^{-1}\boldsymbol\zeta(k)\right)}{\left(2\pi\right)^{\frac{M}{2}}\det F(k)^{\frac{1}{2}}}.
\end{equation}

The maximum likelihood estimator (MLE) $\widehat{\boldsymbol\Theta}^{\mathrm{MLE}}=(\widehat{\boldsymbol b}^{\mathrm{MLE}},\widehat{\beta}^{\mathrm{MLE}},\widehat{\Sigma}^{\mathrm{MLE}},\widehat{\boldsymbol a}^{\mathrm{MLE}},\widehat{\alpha}^{\mathrm{MLE}})$ is found by maximizing the likelihood function $\mathcal L_t(\boldsymbol\Theta)$ over $\boldsymbol\Theta$, given the data. As in the EM (expectation maximization) algorithm, maximization of the likelihood function is alternated with Kalman filtering until convergence of the estimated parameters $\widehat{\boldsymbol\Theta}^{\mathrm{MLE}}$ is achieved.

\subsection{Estimation Motivated by Continuous Time Modeling}
\label{calibration real world 2}
\subsubsection{Rescaling the Time Grid}
Assume factor process $(\boldsymbol X(t))_{t\in\N_0}$ is given under $\p$ by $\boldsymbol X(0)\in \R^n$ and for $t>0$:
\[
\boldsymbol X(t)=\boldsymbol a+\alpha\boldsymbol X(t-1)+\Sigma^{\frac{1}{2}}\boldsymbol\varepsilon(t),
\]
where $\boldsymbol a=\boldsymbol b-\Sigma^{\frac{1}{2}}\boldsymbol\lambda$ and $\alpha=\beta-\Sigma^{\frac{1}{2}}\Lambda$. Furthermore, assume that $\alpha$ is a diagonalizable matrix with $\alpha=TDT^{-1}$ for $T\in\R^{n\times n}$ and diagonal matrix $D\in(-1,1)^{n\times n}$. Then, the transformed process $\boldsymbol Z=(T^{-1}\boldsymbol X(t))_{t\in\N_0}$ evolves according to: 
\[
\boldsymbol Z(t)=\boldsymbol{c}+D\boldsymbol{Z}(t-1)+\Psi^{\frac{1}{2}}\boldsymbol\varepsilon(t),\quad t>0,
\]
where $\boldsymbol c=T^{-1}\boldsymbol a$ and $\Psi=T^{-1}\Sigma(T^{-1})^\top$. For $d \in \mathbb N_+$, the $d$-step ahead conditional distribution of $\boldsymbol Z$ under $\p$ is given by:
\[
\boldsymbol Z(t+d)|\mathcal F(t)\stackrel{\mathbb P}{\sim}\mathcal N\left(\boldsymbol\mu+\gamma\boldsymbol{Z}(t),\Gamma\right),\quad t\geq0,
\]
where $\boldsymbol\mu=\left(\mathds{1}-D\right)^{-1}\left(\mathds{1}-D^d\right)\boldsymbol c$, $\gamma=D^d$ and $\Gamma=\sum_{s=0}^{d-1}D^s\Psi D^s$. Suppose we have estimated $\boldsymbol\mu\in\R^n$, the diagonal matrix $\gamma\in(-1,1)^n$ and $\Gamma\in\R^{n\times n}$ on the time grid with size $d\Delta$, for instance, using MLE, as explained in Section~\ref{sec: kalman MLE}. We are interested in recovering the parameters $\boldsymbol c$, $D$ and $\Psi$ of the dynamics on the refined time grid with size $\Delta$ from $\boldsymbol\mu$, $\gamma$ and $\Gamma$.

The diagonal matrix $D$ and vector $\boldsymbol c$ are reconstructed from the diagonal matrix $\gamma$ as follows:
\[
\begin{aligned}
D&=\gamma^{\frac{1}{d}}=\mathds{1}+\frac{1}{d}\log(\gamma)+o\left(\frac{1}{d}\right), \quad \text{as $d\to\infty$},\\
\boldsymbol c&=\left(\mathds 1-\gamma\right)^{-1}\left(\mathds 1-\gamma^{\frac{1}{d}}\right)\boldsymbol{\mu}=\frac{1}{d}\left(\mathds 1-\gamma\right)^{-1}\log\left(\gamma^{-1}\right)\boldsymbol{\mu}+o\left(\frac{1}{d}\right), \quad \text{as $d\to\infty$},
\end{aligned}
\]
where logarithmic and power functions applied to diagonal matrices are defined on their diagonal elements. Note that for $i,j=1,\ldots,n$, we have:
\[
\Gamma_{ij}=\sum_{s=0}^{d-1}\gamma_{ii}^\frac{s}{d}\Psi_{ij}\gamma_{jj}^\frac{s}{d}=\Psi_{ij}\sum_{s=0}^{d-1}\left(\gamma_{ii}^{\frac{1}{d}}\gamma_{jj}^{\frac{1}{d}}\right)^s=\Psi_{ij}\frac{1-\gamma_{ii}\gamma_{jj}}{1-\left(\gamma_{ii}\gamma_{jj}\right)^{\frac{1}{d}}}.
\]

Therefore, we recover $\Psi$ from $\gamma$ and $\Gamma$ as follows.
\[
\Psi=\frac{1}{d}\upsilon+o\left(\frac{1}{d}\right),\quad\text{as $d\to\infty$},
\]
where $\upsilon=(-\Gamma_{ij}\log(\gamma_{ii}\gamma_{jj})(1-\gamma_{ii}\gamma_{jj})^{-1})_{i,j=1,\ldots,n}\in\R^{n\times n}$. Consider for $t>0$ the increments \linebreak$\mathcal D_t\boldsymbol Z=\boldsymbol Z(t)-\boldsymbol Z(t-1)$. From the formulas for $\boldsymbol c$, $D$ and $\Psi$, we observe that the $\mathcal F_{t-1}$-conditional mean of $\mathcal D_t \boldsymbol Z$:
\[
\boldsymbol c+\left(D-\mathds 1\right)\boldsymbol Z(t-1)=-\frac{1}{d}\left(\mathds1-\gamma\right)^{-1}\log(\gamma)\boldsymbol\mu+\frac{1}{d}\log(\gamma)\boldsymbol Z(t-1)+o\left(\frac{1}{d}\right),
\]
and the $\mathcal F_{t-1}$-conditional volatility of $\mathcal D_t \boldsymbol Z$:
\[
\Psi^{\frac{1}{2}}=\sqrt{\frac{1}{d}}\upsilon^{\frac{1}{2}}+o\left(\sqrt{\frac{1}{d}}\right),
\]
live on different scales as $d \to \infty$; in fact, volatility dominates for large $d$. Under $\p$ for $t>0$, we have:
\[
\begin{aligned}
&\mathbb E\left[\mathcal D_t\boldsymbol Z\left(\mathcal D_t\boldsymbol Z\right)^\top\middle|\mathcal F_{t-1}\right]=\mathrm{Cov}\left[\mathcal D_t\boldsymbol Z,\mathcal D_t\boldsymbol Z\middle|\mathcal F_{t-1}\right]+\mathbb E\left[\mathcal D_t\boldsymbol Z\middle|\mathcal F_{t-1}\right]\mathbb E\left[\mathcal D_t\boldsymbol Z\middle|\mathcal F_{t-1}\right]^\top\\\quad&=\mathrm{Cov}\left[\boldsymbol Z(t),\boldsymbol Z(t)\middle|\mathcal F_{t-1}\right]+\left(\mathbb E\left[\boldsymbol Z(t)\middle|\mathcal F_{t-1}\right]-Z(t-1)\right)\left(\mathbb E\left[\boldsymbol Z(t)\middle|\mathcal F_{t-1}\right]-Z(t-1)\right)^\top\\&=\Psi+\left(\boldsymbol c+\left(D-\mathds 1\right)\boldsymbol Z(t-1)\right)\left(\boldsymbol c+\left(D-\mathds 1\right)\boldsymbol Z(t-1)\right)^\top.
\end{aligned}
\]

Therefore, setting $\mathcal D_t\boldsymbol X=\boldsymbol X(t)-\boldsymbol X(t-1)$, we obtain as $d\to\infty$:
\begin{equation}\label{eq: small delta}
\begin{aligned}
&\mathbb E\left[\mathcal D_t\boldsymbol X\left(\mathcal D_t\boldsymbol X\right)^\top\middle|\mathcal F_{t-1}\right]=T\mathbb E\left[\mathcal D_t\boldsymbol Z\left(\mathcal D_t\boldsymbol Z\right)^\top\middle|\mathcal F_{t-1}\right]T^\top\\&\quad=T\Psi T^\top+T\left(\boldsymbol c+\left(D-\mathds 1\right)\boldsymbol Z(t-1)\right)\left(\boldsymbol c+\left(D-\mathds 1\right)\boldsymbol Z(t-1)\right)^\top T^\top\\&\quad=\frac{1}{d}T\upsilon T^\top+o\left(\frac{1}{d}\right)=T\Psi T^\top+o\left(\frac{1}{d}\right)=\Sigma+o\left(\frac{1}{d}\right),
\end{aligned}
\end{equation}

\subsubsection{Longitudinal Realized Covariations of Yields}
We consider the yield curve increments within the discrete-time multifactor Vasi\v{c}ek Models~\eqref{eq: spot rate} and \eqref{eq: ARn}. The increments of the yield process $(Y(t,t+\tau))_{t\in\mathbb N_0}$ for fixed time to maturity $\tau\Delta>0$ are given by:
\[
\begin{aligned}
\mathcal D_{t,\tau}Y&=Y\left(t,t+\tau\right)-Y\left(t-1,t-1+\tau\right)\\&=\frac{1}{\tau\Delta}\boldsymbol B(t,t+\tau)^\top\left(\boldsymbol X(t)-\boldsymbol X(t-1)\right)=\frac{1}{\tau\Delta}\boldsymbol B(t,t+\tau)^\top\mathcal D_t\boldsymbol X,
\end{aligned}
\]
where $\mathcal D_t\boldsymbol X|\mathcal F(t-1)\stackrel{\mathbb P}{\sim}\mathcal N\left(\boldsymbol a+\left(\alpha-\mathds{1}\right)\boldsymbol X(t-1),\Sigma\right)$. For times to maturity $\tau_1\Delta,\tau_2\Delta>0$, we get under~$\p$: 
\[
\mathbb E\left[\mathcal D_{t,\tau_1}Y\mathcal D_{t,\tau_2}Y\middle|{\cal F}_{t-1}\right]=\frac{1}{\tau_1\tau_2\Delta^2}\boldsymbol B(t,t+\tau_1)^\top\mathbb E\left[\mathcal D_t\boldsymbol X\left(\mathcal D_t\boldsymbol X\right)^\top\middle|{\cal F}_{t-1}\right]\boldsymbol B(t,t+\tau_2).
\]

By Equation \eqref{eq: small delta} for small grid size $\Delta$, we estimate the last expression by:
\begin{equation}\label{eq: calibration}
\mathbb E\left[\mathcal D_{t,\tau_1}Y\mathcal D_{t,\tau_2}Y\middle|{\cal F}_{t-1}\right]\approx\frac{1}{\tau_1\tau_2}\boldsymbol 1^\top\left(\mathds{1}-\beta^{\tau_1}\right)\left(\mathds{1}-\beta\right)^{-1}\Sigma\left(\mathds{1}-\beta^\top\right)^{-1}\left(\mathds{1}-\left(\beta^\top\right)^{\tau_2}\right)\boldsymbol 1.
\end{equation}

Formula \eqref{eq: calibration} is interesting for the following reasons:
\begin{itemize}[leftmargin=*,labelsep=6mm]
\item It does not depend on the unobservable factors $\boldsymbol X$.
\item It allows for direct cross-sectional estimation of $\beta$ and $\Sigma$. That is, $\beta$ and $\Sigma$ can directly be estimated from market observations without knowing the market-price of risk.
\item It is helpful to determine the number of factors needed to fit the model to market yield curve increments. This can be analyzed by principal component analysis.
\item It can also be interpreted as a small-noise approximation for noisy measurement systems of the form \eqref{eq: noisy measurement}.
\end{itemize}

Let $\widehat{y}_{1}(k)$ and $\widehat{y}_{2}(k)$ be market observations for times to maturity $\tau_{1}\Delta$ and $\tau_{2}\Delta$ and at times $k\in\{t-K+1,\ldots,t\}$, also specified in Section~\ref{sec: kalman MLE}. Then, the expectation on the left hand side of \eqref{eq: calibration} can be estimated by the realized covariation:
\begin{equation}\label{eq: rcov}
\widehat{\rm RCov}(t,\tau_1,\tau_2)=\frac{1}{K} \sum_{k=t-K+1}^t\big(\widehat{y}_{1}(k)-\widehat{y}_{1}(k-1)\big)\big(\widehat{y}_{2}(k)-\widehat{y}_{2}(k-1)\big).
\end{equation}

The quality of this estimator hinges on two crucial assumptions. First, higher order terms in \eqref{eq: small delta} are negligible in comparison to $\Sigma$. Second, the noise term $S^{\frac12}\boldsymbol\eta$ in \eqref{eq: noisy measurement} leads to a negligible distortion in the sense that observations $\widehat{\boldsymbol Y}$ are reliable indicators for the underlying Vasi\v cek yield curves.

\subsubsection[Cross-Sectional Estimation of beta and Sigma]{Cross-Sectional Estimation of $\beta$ and $\Sigma$}
Realized covariation estimator \eqref{eq: rcov} can be used in conjunction with asymptotic relation \eqref{eq: calibration} to estimate parameters $\beta$ and $\Sigma$ at time $ t\Delta$ in the following way. For given symmetric weights $w_{ij}=w_{ji}\geq0$, we solve the least squares problem:
\begin{equation}\label{eq: co-var estimate}
\begin{aligned}
\left(\widehat{\beta}^{\rm RCov},\widehat{\Sigma}^{\rm RCov}\right)&={\arg\min}_{\beta,\Sigma}\Bigg\{\sum_{i,j=1}^{M}w_{ij}\bigg[\widehat{\rm RCov}(t,\tau_i,\tau_j)\\&\quad-\frac{1}{\tau_i\tau_j}\boldsymbol 1^\top\left(\mathds{1}-\beta^{\tau_i}\right)\left(\mathds{1}-\beta\right)^{-1}\Sigma\left(\mathds{1}-\beta^\top\right)^{-1}\left(I-\left(\beta^\top\right)^{\tau_j}\right)\boldsymbol 1\bigg]^2\Bigg\},
\end{aligned}
\end{equation}
where we optimize over $\beta$ and $\Sigma$ satisfying Assumption~\ref{assumption}.

\subsection{Inference on Market Price of Risk}

Finally, we aim at determining parameters $\boldsymbol\lambda$ and $\Lambda$ of the change of measure specified in Section~\ref{sec: real world dynamics}. For this purpose, we combine MLE estimation (Section~\ref{sec: kalman MLE}) with estimation from realized covariations of yields (Section~\ref{calibration real world 2}). First, we estimate $\beta$ and $\Sigma$ by $\widehat\beta^{\mathrm{RCov}}$ and $\widehat\Sigma^{\mathrm{RCov}}$ as in Section~\ref{calibration real world 2}. Second, we estimate $\boldsymbol a$, $\boldsymbol b$ and $\alpha$ by maximizing the log-likelihood:
\[
\log\mathcal{L}_t\left(\boldsymbol b,\beta,\Sigma,\boldsymbol a,\alpha\right)=\sum_{k=t-K+1}^{t}\log\left(\det F(k)\right)-\sum_{k=t-K+1}^{t}\boldsymbol\zeta(k)^\top F(k)^{-1}\boldsymbol\zeta(k)+\mathrm{const},
\]
for fixed $\beta$ and $\Sigma$ over $\boldsymbol{b}\in\R^n$, $\boldsymbol{a}\in\R^n$ and $\alpha\in\R^{n\times n}$ with spectrum in $(-1,1)^n$, i.e., 
\begin{equation}\label{eq: mle}
\left(\widehat{\boldsymbol{b}}^{\mathrm{MLE}},\widehat{\boldsymbol{a}}^{\mathrm{MLE}},\widehat{\alpha}^{\rm MLE}\right)={\arg\max}_{\boldsymbol b,\boldsymbol a,\alpha}\:\:\log\mathcal{L}_t\left(\boldsymbol b,\widehat\beta^{\mathrm{RCov}},\widehat\Sigma^{\mathrm{RCov}},\boldsymbol a,\alpha\right).
\end{equation}

The constraint on the matrix $\alpha$ ensures that the factor process is stationary under the real-world measure $\p$. From Equation \eqref{eq: real world transform}, we have $\boldsymbol\lambda=\Sigma^{-\frac{1}{2}}\left(\boldsymbol b-\boldsymbol a\right)$ and $\Lambda=\Sigma^{-\frac{1}{2}}\left(\beta-\alpha\right)$. This motivates the inference of $\boldsymbol\lambda$ by:
\begin{equation}\label{eq: lambda}
\widehat{\boldsymbol\lambda}=\left(\widehat{\Sigma}^{\mathrm{RCov}}\right)^{-\frac{1}{2}}\left(\widehat{\boldsymbol b}^{\mathrm{MLE}}-\widehat{\boldsymbol a}^{\mathrm{MLE}}\right),
\end{equation}
and the inference of $\Lambda$ by:
\begin{equation}\label{eq: lambda mat}
\widehat{\Lambda}(k)=\left(\widehat{\Sigma}^{\mathrm{RCov}}\right)^{-\frac{1}{2}}\left(\widehat{\beta}^{\mathrm{RCov}}-\widehat{\alpha}^{\mathrm{MLE}}\right).
\end{equation}

We stress the importance of estimating as many parameters as possible from the realized covariations of yields prior to using maximum likelihood estimation. The MLE procedure of Section~\ref{sec: kalman MLE} is computationally intensive and generally does not work well to estimate volatility~parameters.

\section{Numerical Example for Swiss Interest Rates}\label{sec: numerical example}
\subsection{Description and Selection of Data}

We choose $\Delta=1/252$, which corresponds to a daily time grid (assuming that a financial year has 252 business days). For the Swiss currency (CHF), we consider as yield observations the Swiss Average Rate (SAR), the London InterBank Offered
Rate (LIBOR) and the Swiss Confederation Bond (SWCNB). See Figures \ref{fig: data start} and \ref{fig: libor comment}.

\begin{itemize}[leftmargin=*,labelsep=6mm]
\item {\it Short times to maturity.} The SAR is an ongoing volume-weighted average rate calculated by the Swiss National Bank (SNB) based on repo transactions between financial institutions. It is used for short times to maturity of at most three months. For SAR, we have the Over-Night SARON that corresponds to a time to maturity of $\Delta$ (one business day) and the SAR Tomorrow-Next (SARTN) for time to maturity $2 \Delta$ (two business days). The latter is not completely correct, because SARON is a collateral over-night rate and tomorrow-next is a call money rate for receiving money tomorrow, which has to be paid back the next business day. Moreover, we have the SAR for times to maturity of one week (SAR1W), two weeks (SAR2W), one month (SAR1M) and three months (SAR3M); see also \cite{Jordan}.
\item {\it Short to medium times to maturity.} The LIBOR reflects times to maturity, which correspond to one~month (LIBOR1M), three months (LIBOR3M), six months (LIBOR6M) and 12 months (LIBOR12M) in the London interbank market.
\item {\it Medium to long times to maturity.} The SWCNB is based on Swiss government bonds, and it is used for times to maturity, which correspond to two years (SWCNB2Y), three years (SWCNB3Y), four~years (SWCNB4Y), five years (SWCNB5Y), seven years (SWCNB7Y), 10 years (SWCNB10Y), 20 years (SWCNB20Y) and 30 years (SWCNB30Y).
\end{itemize}

These data are available from 8 December 1999, and we set 15 September 2014 to be the last observation date. Of course, SAR, LIBOR and SWCNB do not exactly model risk-free zero-coupon bonds, and these different classes of instruments are not completely consistent, because prices are determined slightly differently for each class. In particular, this can be seen during the 2008--2009 financial crisis. However, these data are in many cases the best approximation to CHF risk-free zero-coupon yields that is available. For the longest times to maturity of SWCNB, one may also raise issues about the liquidity of these instruments, because insurance companies typically run a buy-and-hold strategy for long-term bonds. 

In Figures \ref{fig: rcov start}--\ref{fig: rcov end}, we compute the realized volatility $\widehat{\rm RCov}(t,\tau,\tau)^{\frac{1}{2}}$ of yield curves $(\widehat{y}_\tau(k))_{k=t-K+1,\ldots, t}$ for different times to maturity $\tau\Delta$ and window length $K$; see Equation \eqref{eq: rcov}. In~Figures \ref{fig: libor comment} and \ref{fig: rcov end}, we observe that SAR fits SWCNB better than LIBOR after the financial crisis of 2008. For this reason, we decide to drop LIBOR and build daily yield curves from SAR and SWCNB, only. The mismatch between LIBOR, SAR and SWCNB is attributable to differences in liquidity and the credit risk of the underlying instruments. 

\subsection{Model Selection}\label{sec: model selection}

In this numerical example, we restrict ourselves to multifactor Vasi\v cek models with $\beta$ and $\alpha$ of diagonal form:
\[
\beta=\mathrm{diag}\left(\beta_{11},\ldots,\beta_{nn}\right),\quad\text{and}\quad\alpha=\mathrm{diag}\left(\alpha_{11},\ldots,\alpha_{nn}\right),
\]
where $-1<\beta_{11},\ldots,\beta_{nn},\alpha_{11},\ldots,\alpha_{nn}<1$. In the following, we explain exactly how to perform the delicate task of parameter estimation in the multifactor Vasi\v cek Models~\eqref{eq: spot rate} and \eqref{eq: ARn} using the procedure explained in Section~\ref{sec: parameters}. 

\subsubsection{Discussion of Identification Assumptions} We select short times to maturity (SAR) to estimate parameters $\boldsymbol b$, $\beta$, $\Sigma$, $\boldsymbol a$ and $\alpha$. This is reasonable because these parameters describe the dynamics of the factor process and, thus, of the spot rate. As we are working on a small (daily) time grid, asymptotic Formulas \eqref{eq: small delta} and \eqref{eq: calibration} are expected to give good approximations. Additionally, it is reasonable to assume that the noise covariance matrix $S$ in data-generating Model~\eqref{eq: noisy measurement} is negligible compared to \eqref{eq: calibration}. Therefore, we can estimate the left hand side of \eqref{eq: calibration} by the realized covariation of observed yields; see estimator \eqref{eq: rcov}. Then, we determine the Hull--White extension $\theta$ in order to match the prevailing yield curve interpolated from SAR and~SWCNB. 

\begin{figure}[H]
\centering
\begin{minipage}[t]{0.42\textwidth}
\includegraphics[width=\textwidth]{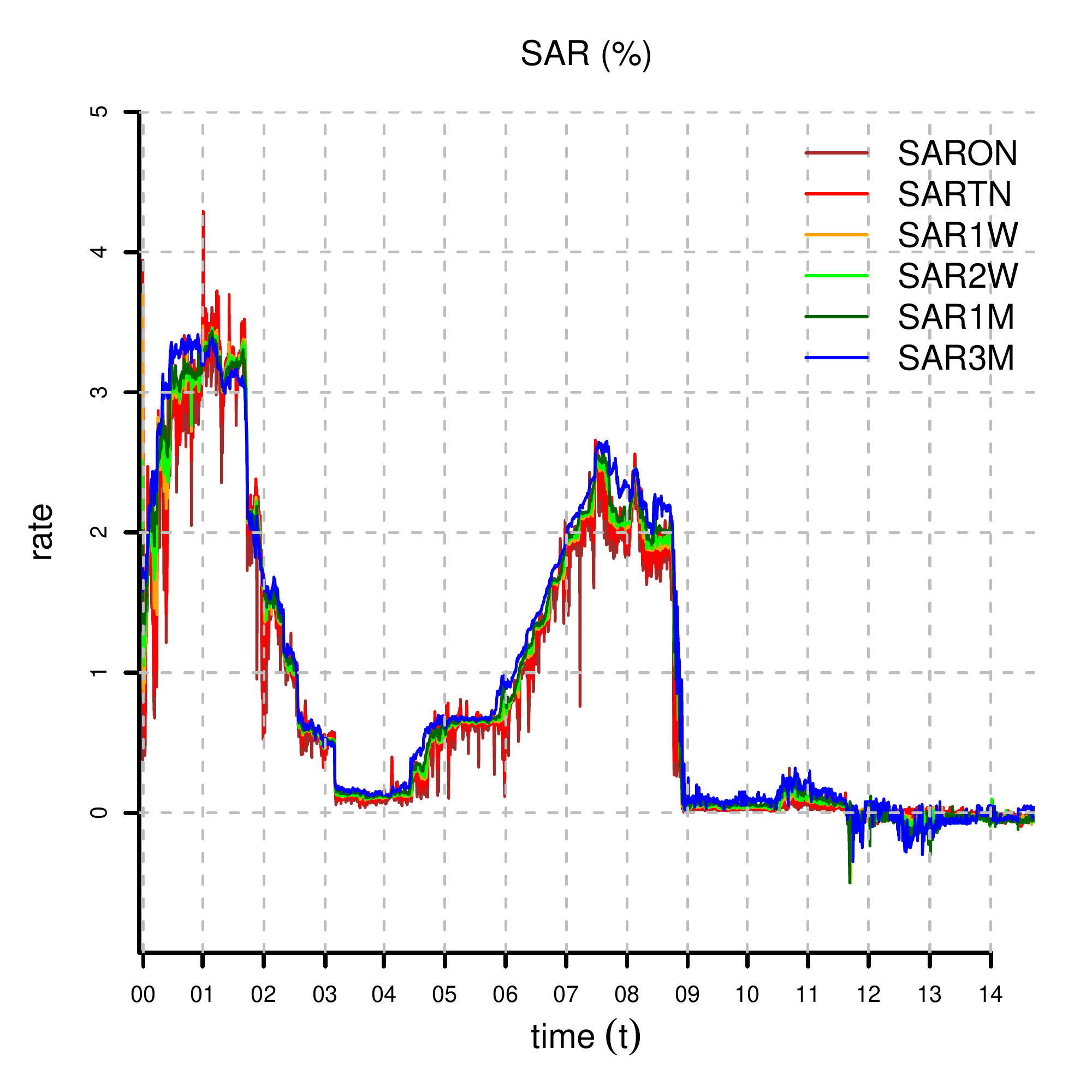}
\end{minipage}
\begin{minipage}[t]{0.42\textwidth}
\includegraphics[width=\textwidth]{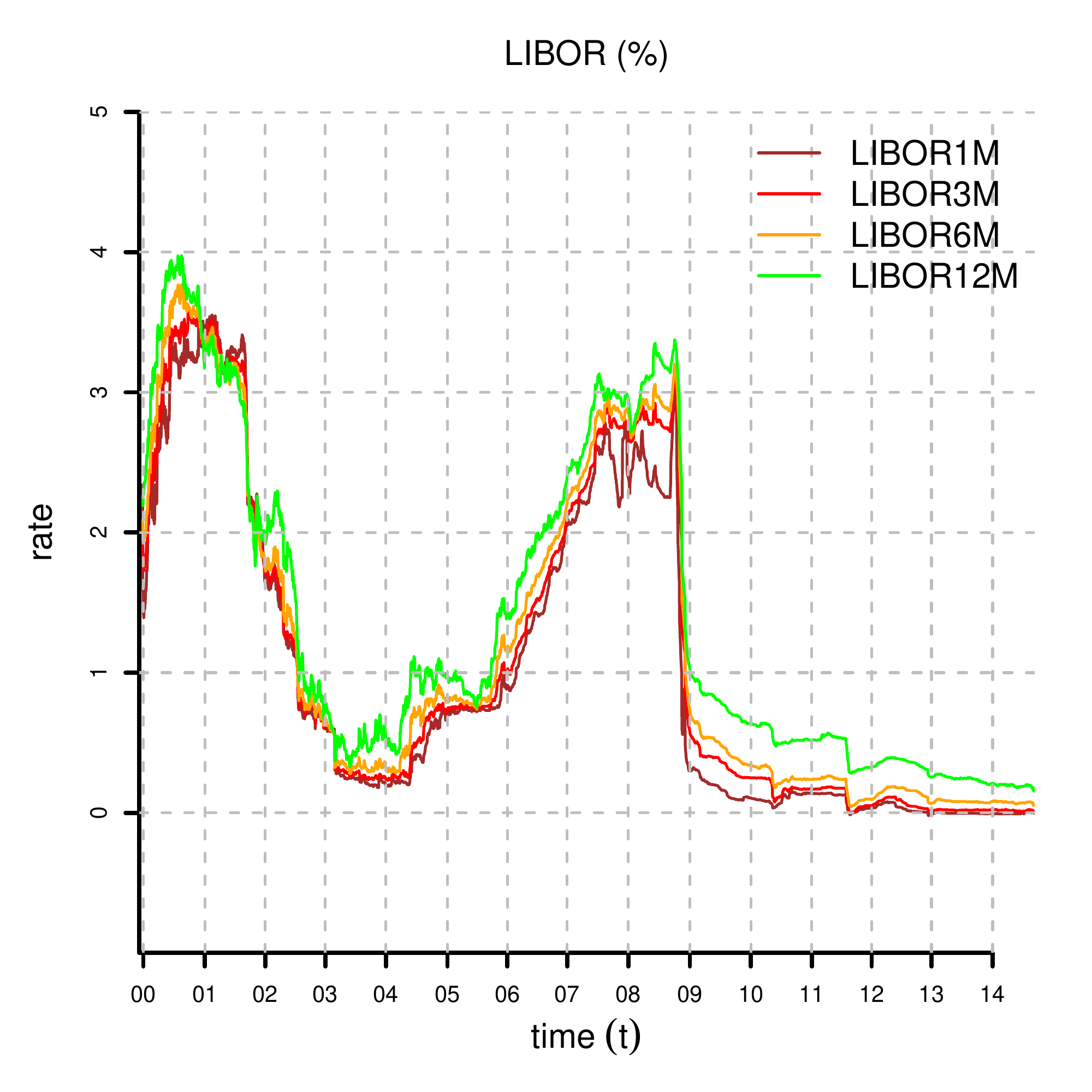}
\end{minipage}
\vspace{-15pt}
\caption{Yield rates (lhs): Swiss Average Rate (SAR) and (rhs) London InterBank Offered Rate (LIBOR) from 8 December 1999, until 15 September 2014.}\label{fig: data start}
\end{figure}
\vspace{-20pt}

\begin{figure}[H]
\centering
\begin{minipage}[t]{0.42\textwidth}
\includegraphics[width=\textwidth]{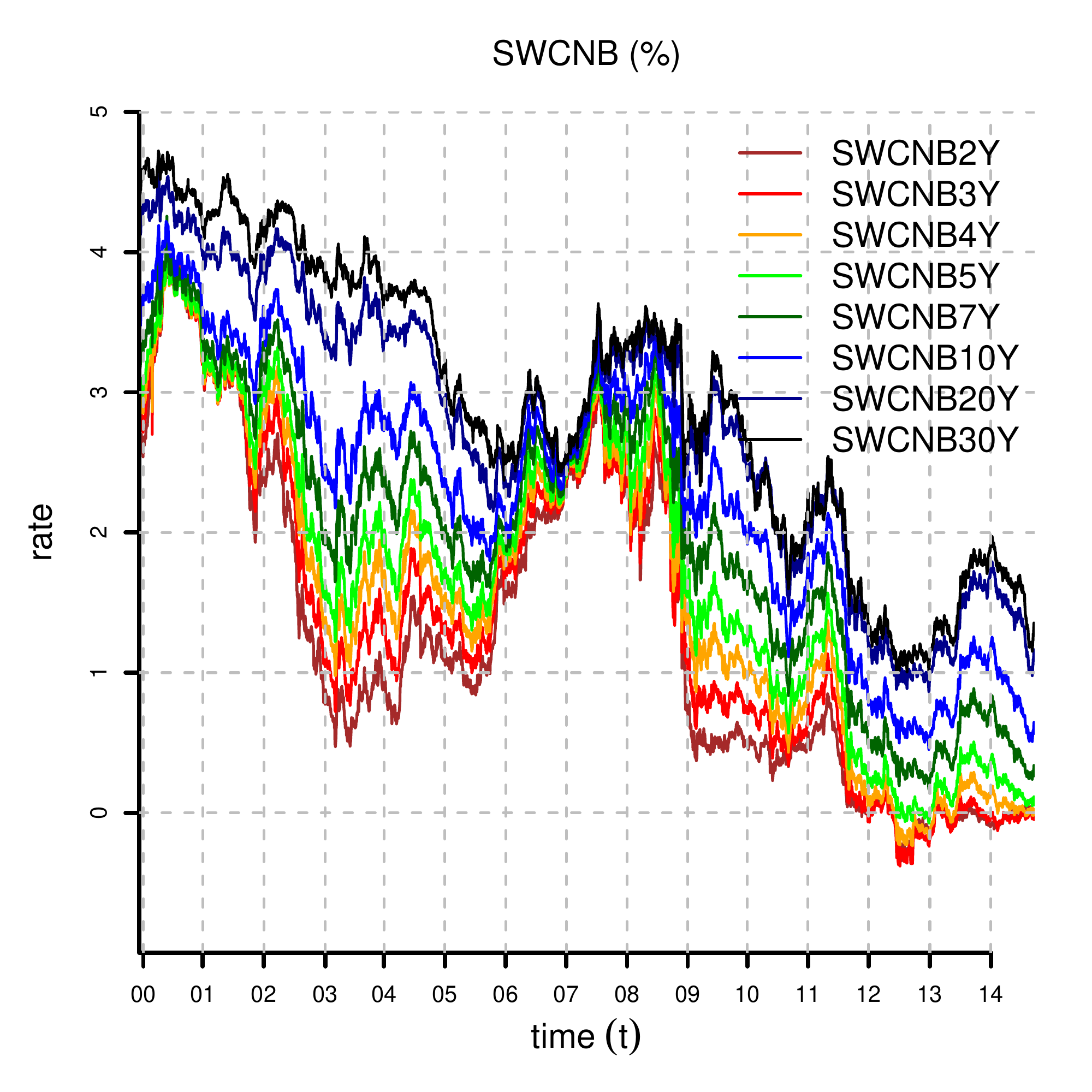}
\end{minipage}
\begin{minipage}[t]{0.42\textwidth}
\includegraphics[width=\textwidth]{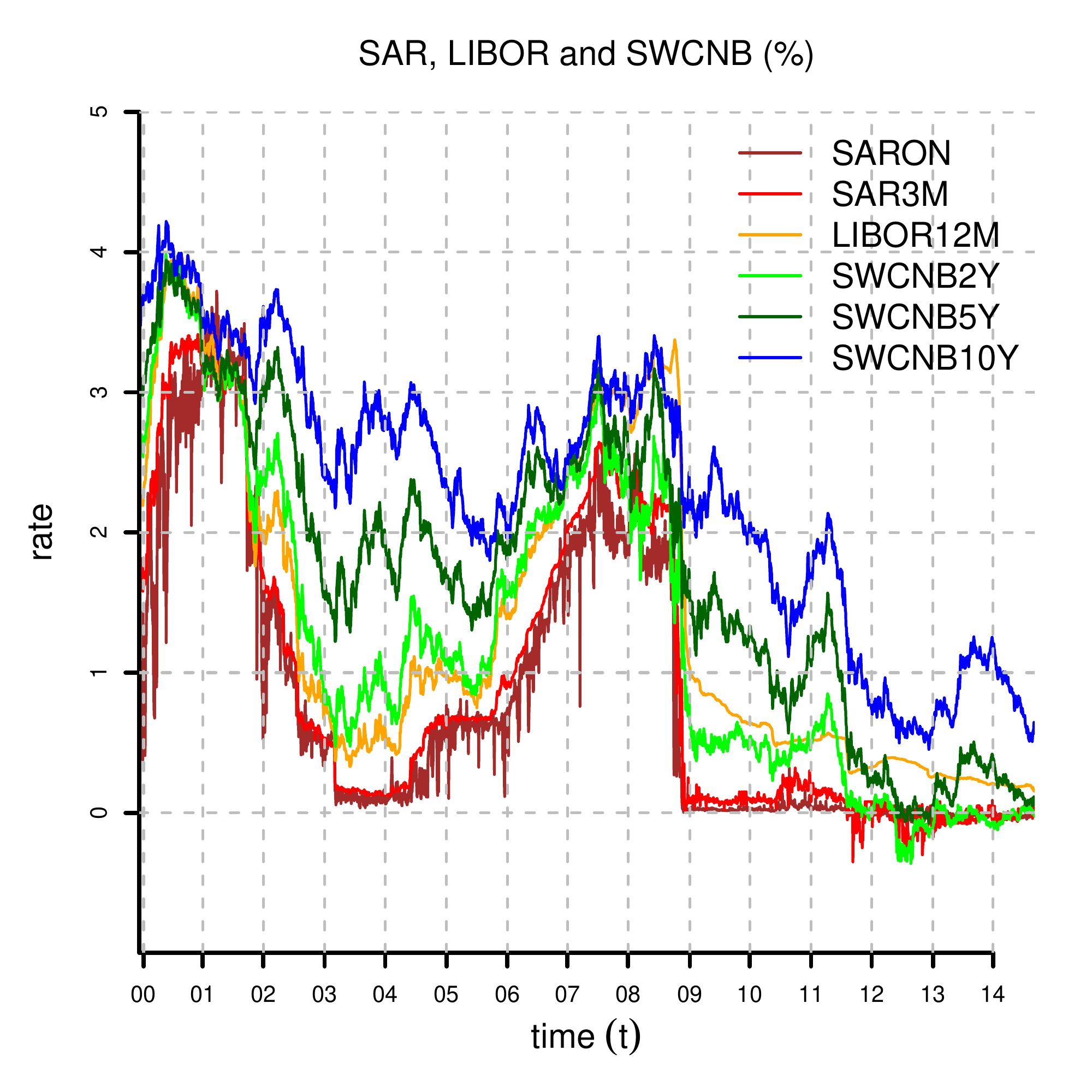}
\end{minipage}
\vspace{-15pt}
\caption{Yield rates: (lhs) Swiss Confederation Bond (SWCNB) and (rhs) a selection of SAR, LIBOR and Swiss Confederation Bond (SWCNB) from 8 December 1999, until 15 September 2014. Note that LIBOR looks rather differently from SAR and SWCNB after the financial crisis of 2008.}\label{fig: libor comment}
\end{figure}
\vspace{-20pt}

\begin{figure}[H]
\centering
\begin{minipage}[t]{0.42\textwidth}
\includegraphics[width=\textwidth]{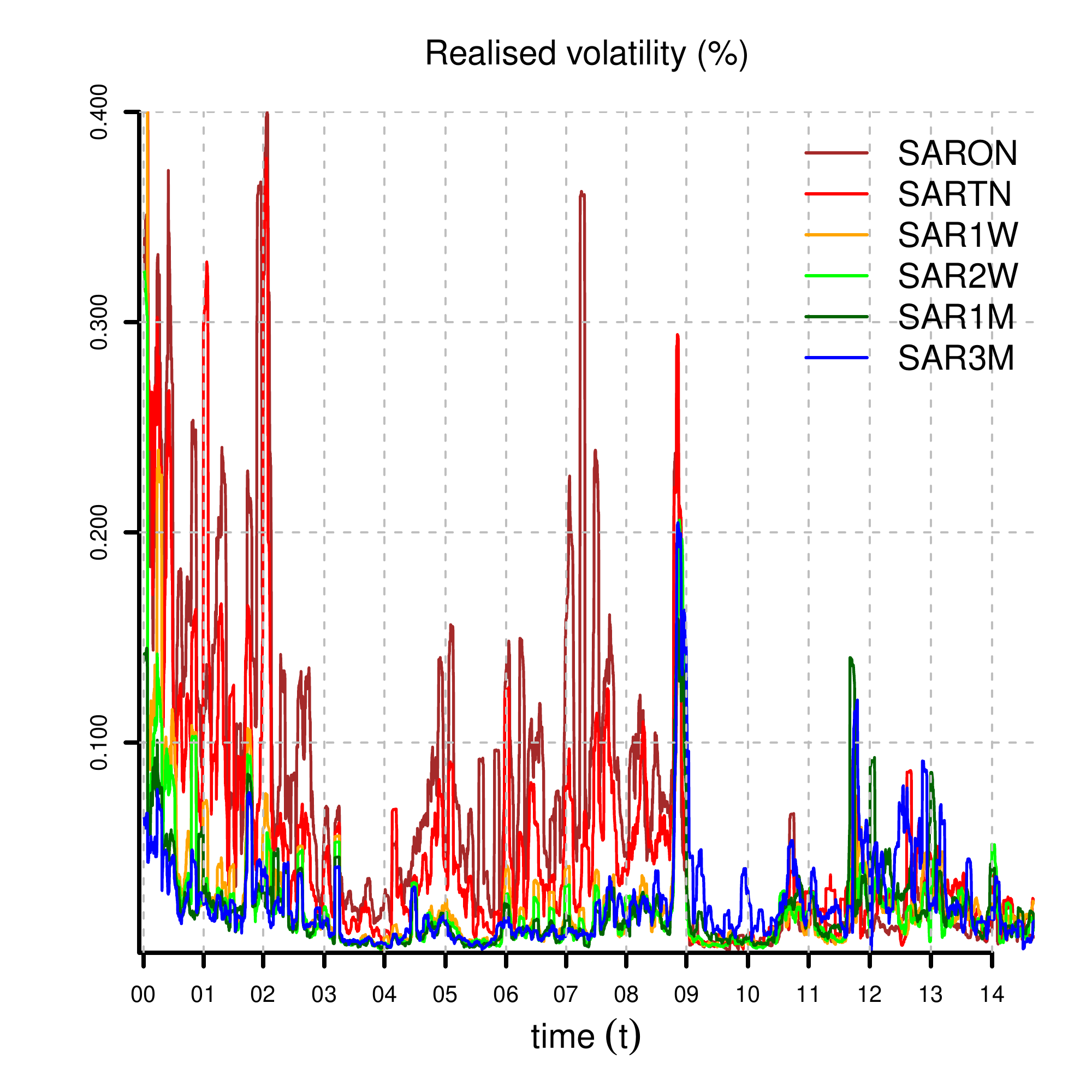}
\end{minipage}
\begin{minipage}[t]{0.42\textwidth}
\includegraphics[width=\textwidth]{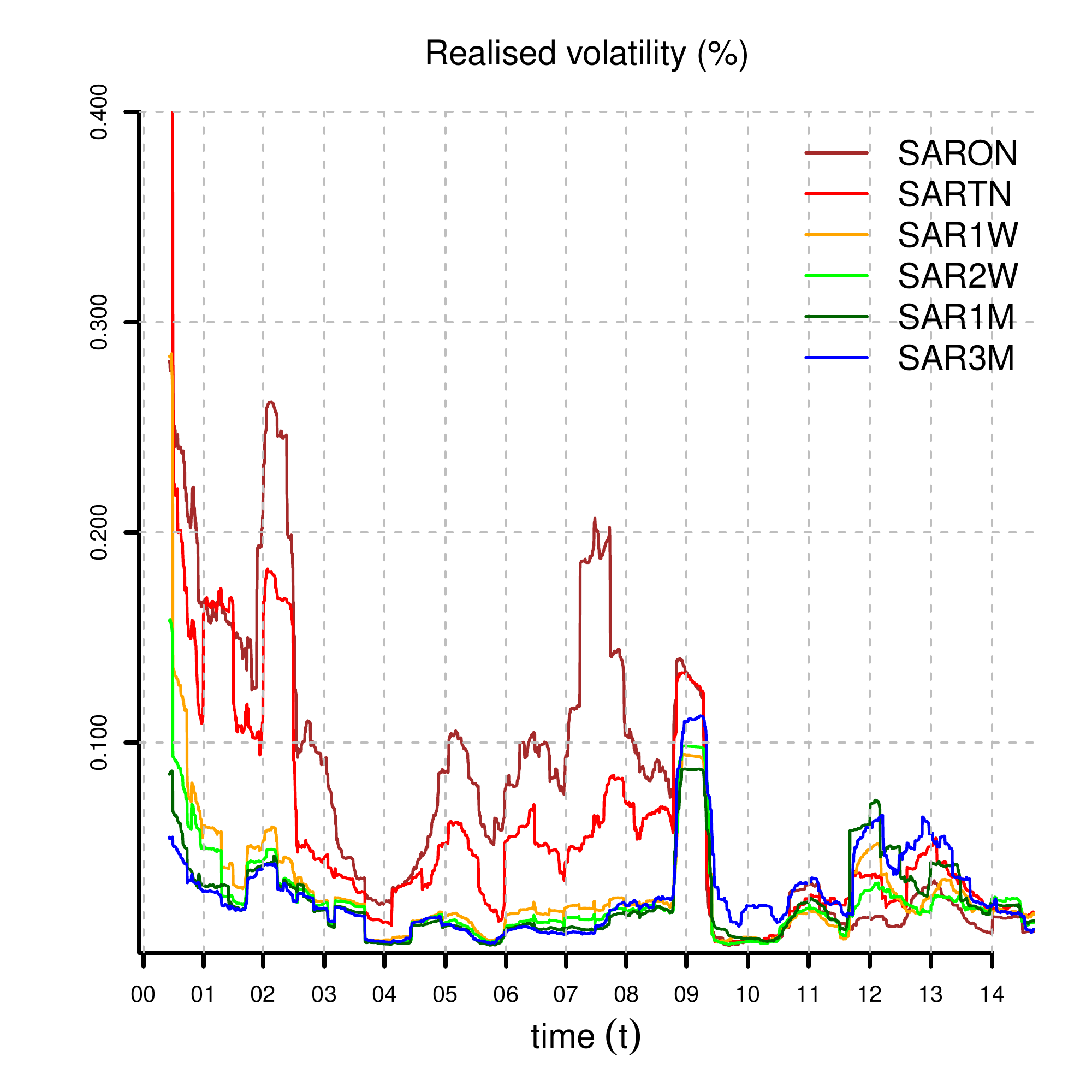}
\end{minipage}
\vspace{-15pt}
\caption{SAR realized volatility $\widehat{\rm RCov}(t,\tau,\tau)^{\frac{1}{2}}$ for $\tau=1,2,5,10,21,63$, window length $K=21$ (lhs) and $K=126$ (rhs).}\label{fig: rcov start}
\end{figure}
\vspace{-20pt}

\begin{figure}[H]
\centering
\begin{minipage}[t]{0.42\textwidth}
\includegraphics[width=\textwidth]{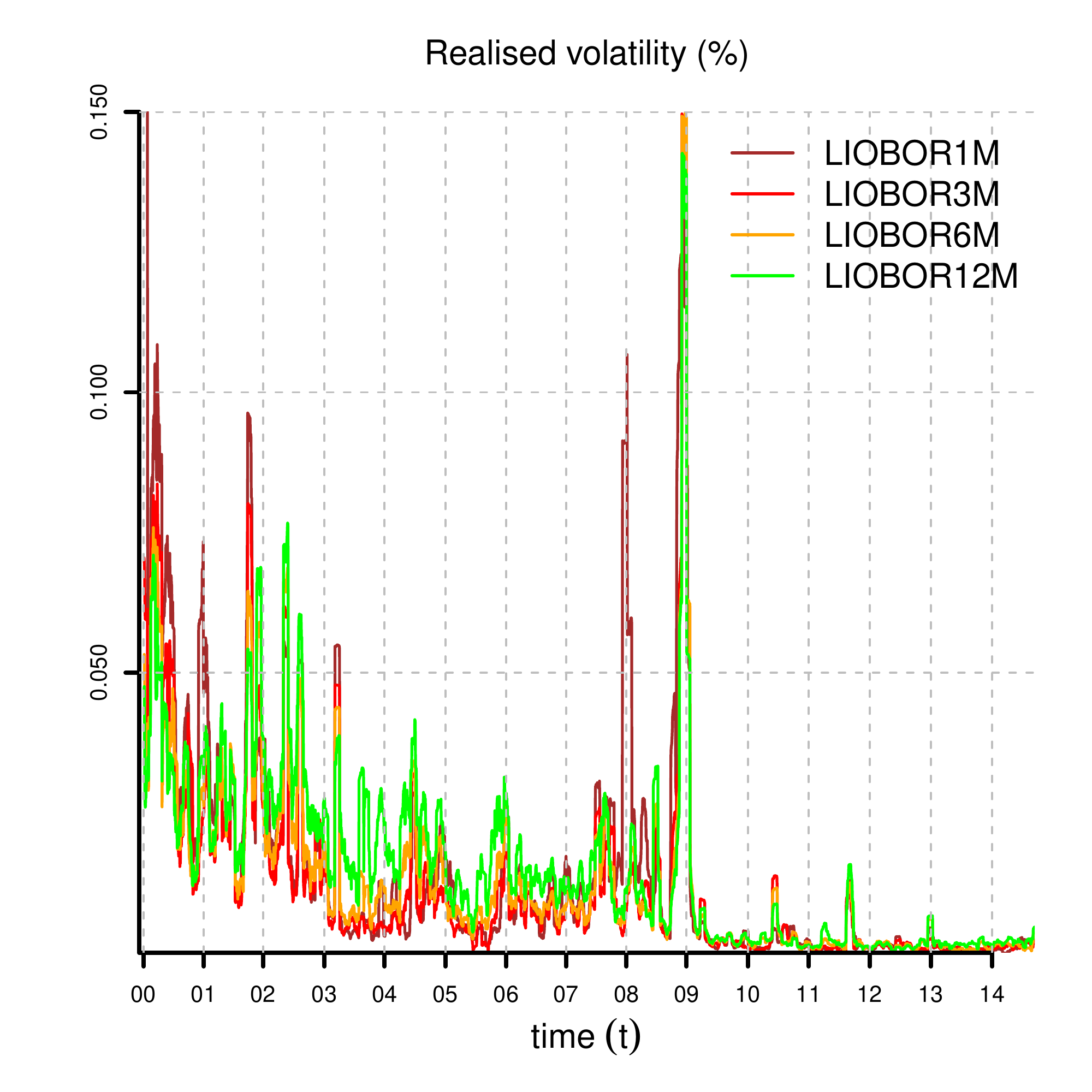}
\end{minipage}
\begin{minipage}[t]{0.42\textwidth}
\includegraphics[width=\textwidth]{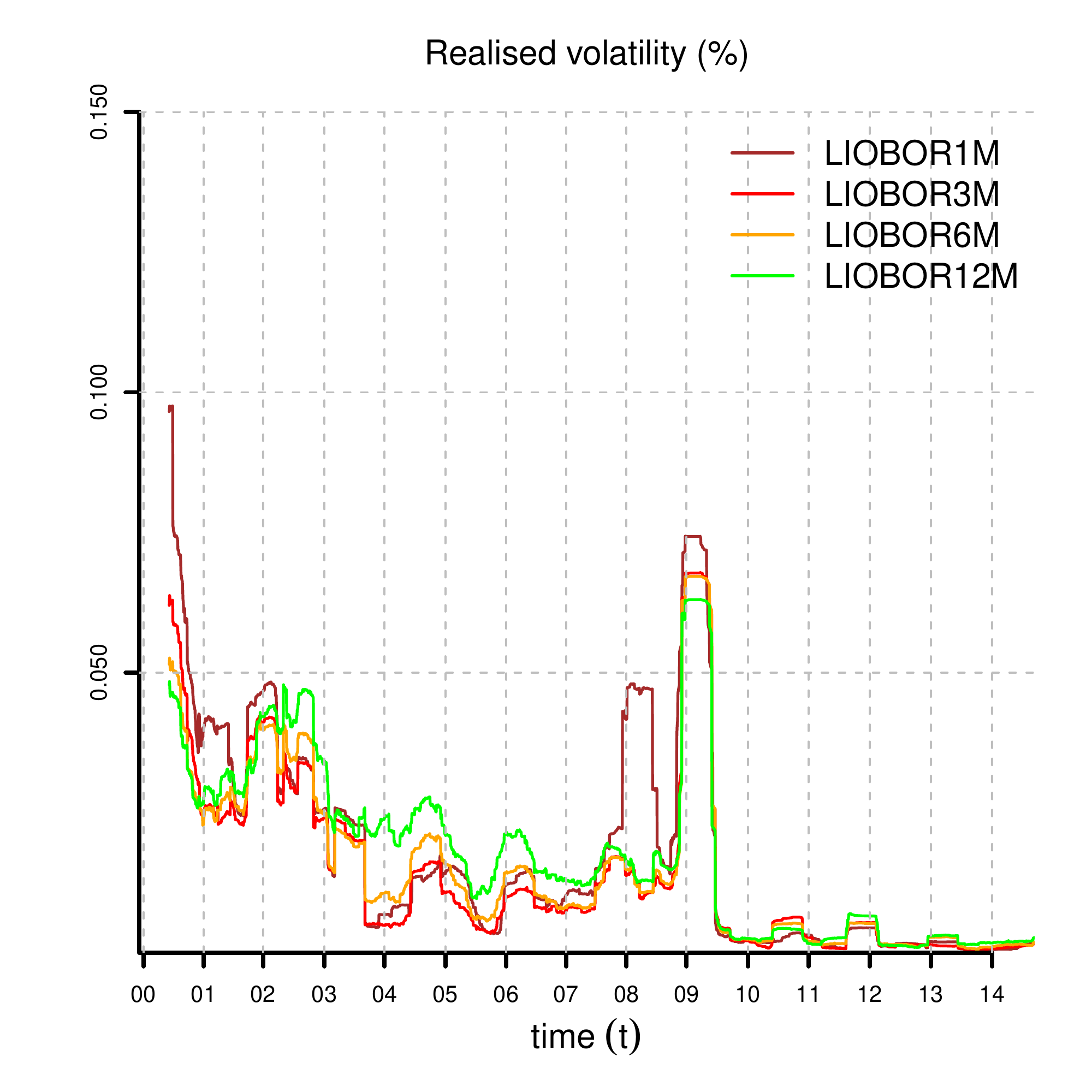}
\end{minipage}
\vspace{-15pt}
\caption{LIBOR realized volatility $\widehat{\rm RCov}(t,\tau,\tau)^{\frac{1}{2}}$ for $\tau=21,63,126,252$, window length $K=21$ (lhs) and $K=126$ (rhs).} 
\end{figure}
\vspace{-20pt}

\begin{figure}[H]
\centering
\begin{minipage}[t]{0.42\textwidth}
\includegraphics[width=\textwidth]{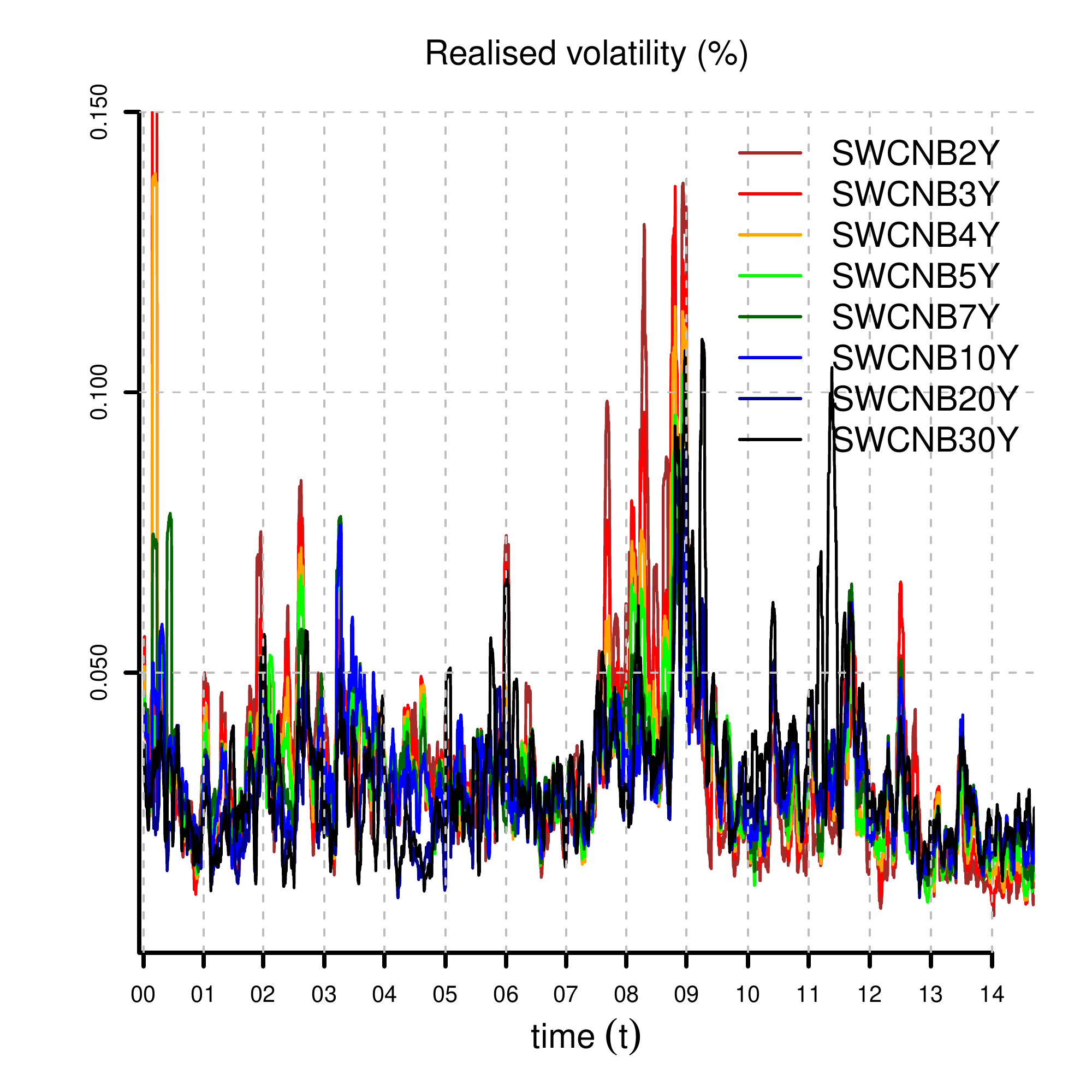}
\end{minipage}
\begin{minipage}[t]{0.42\textwidth}
\includegraphics[width=\textwidth]{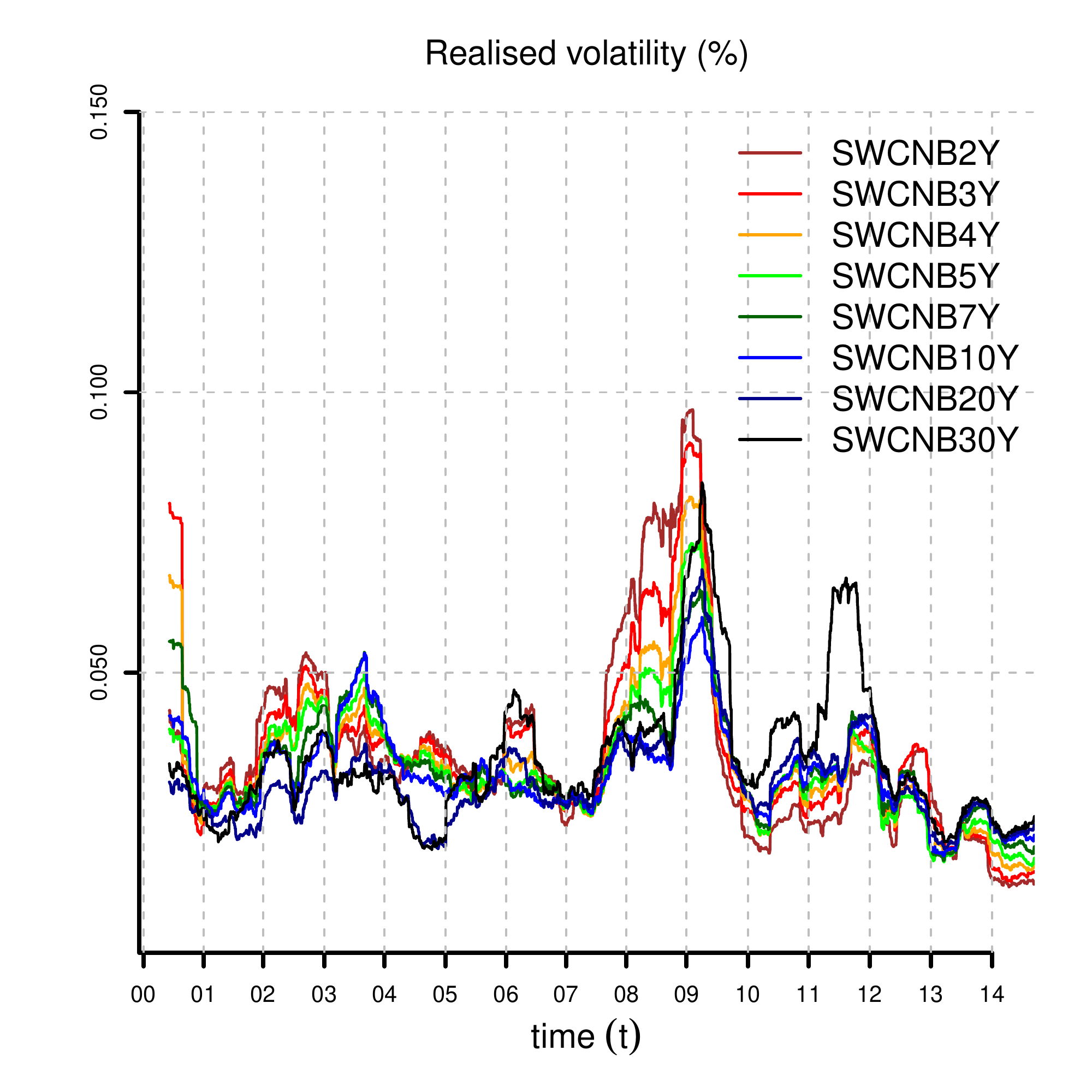}
\end{minipage}
\vspace{-15pt}
\caption{SWCNB realized volatility $\widehat{\rm RCov}(t,\tau,\tau)^{\frac{1}{2}}$ for $\tau/252=2,3,4,5,7,10,20,30$, window length $K=21$ (lhs) and $K=126$ (rhs).} 
\end{figure}
\vspace{-20pt}

\begin{figure}[H]
\centering
\begin{minipage}[t]{0.42\textwidth}
\includegraphics[width=\textwidth]{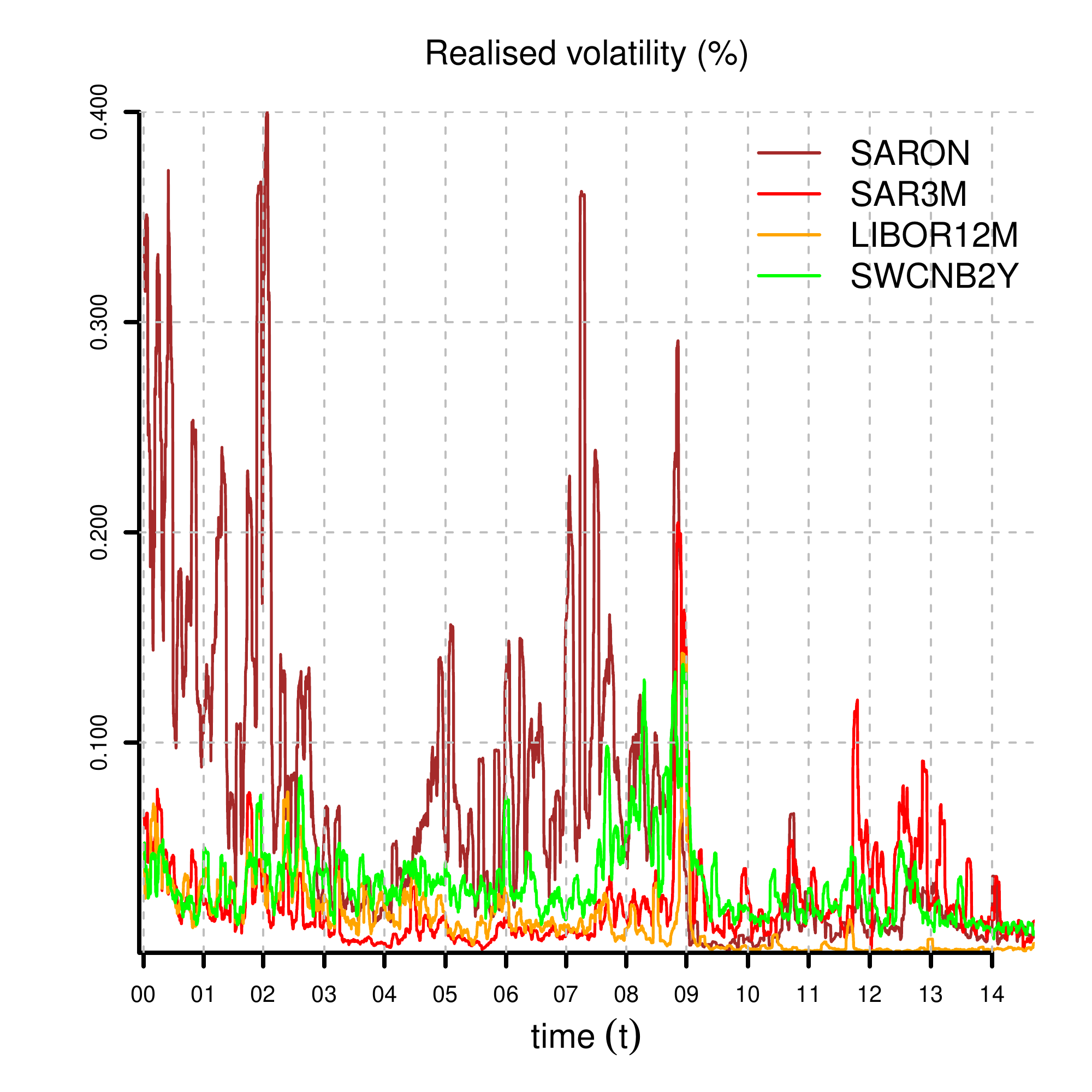}
\end{minipage}
\begin{minipage}[t]{0.42\textwidth}
\includegraphics[width=\textwidth]{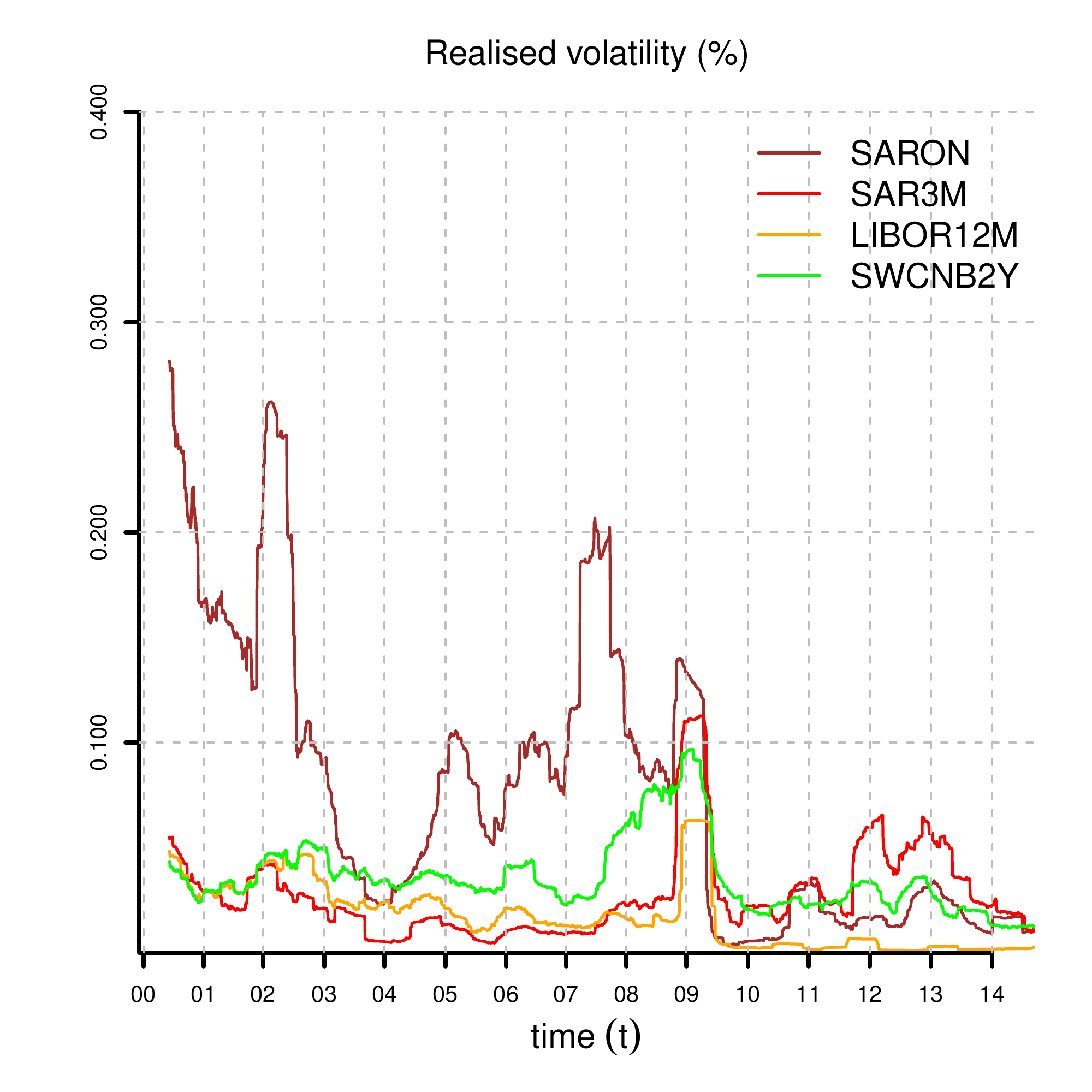}
\end{minipage}
\vspace{-15pt}
\caption{A selection of SAR, LIBOR and SWCNB realized volatility $\widehat{\rm RCov}(t,\tau,\tau)^{\frac{1}{2}}$ for $\tau=$ 1, 63, 252, 504, window length $K=21$ (lhs) and $K=126$ (rhs). Note that LIBOR looks rather differently from SAR and SWCNB after the financial crisis of 2008.}\label{fig: rcov end}
\end{figure}

\subsubsection{Determination of the Number of Factors} We need to determine the appropriate number of factors $n$. The more factors we use, the better we can fit the model to the data. However, the dimensionality of the estimation problem increases quadratically in the number of factors, and the model may become over-parametrized. Therefore, we look for a trade-off between the accuracy of the model and the number of parameters used. In Figure \ref{fig: number of factors sample dates}, we determine $\beta_{11},\ldots,\beta_{nn}$ and $\Sigma$ by solving optimization~\eqref{eq: co-var estimate} numerically for three observation dates and $n=2,3$. A three-factor model is able to capture rather accurately the dependence on the time to maturity $\tau$. In Figure \ref{fig: number of factors all dates start}, we compare the realized volatility of the numerical solution of \eqref{eq: co-var estimate} to the market realized volatility for all observation dates. We observe that in several periods, the two-factor model is not able to fit the SAR realized volatilities accurately for all times to maturities. The three-factor model achieves an accurate fit for most observation dates. The model exhibits small mismatches in 2001, 2008--2009 and 2011--2012. These are periods characterized by a sharp reduction in interest rates in response to financial crises. In September 2011, following strong appreciation of the Swiss Franc with respect to the Euro, the SNB pledged to no longer tolerate Euro-Franc exchange rates below the minimum rate of $1.20$, effectively enforcing a currency floor for more than three years. As a consequence of the European sovereign debt crisis and the intervention of the SNB starting from 2011, we have a long period of very low (even negative) interest rates. 

\begin{figure}[H]
\centering
\begin{minipage}[t]{0.42\textwidth}
\includegraphics[width=\textwidth]{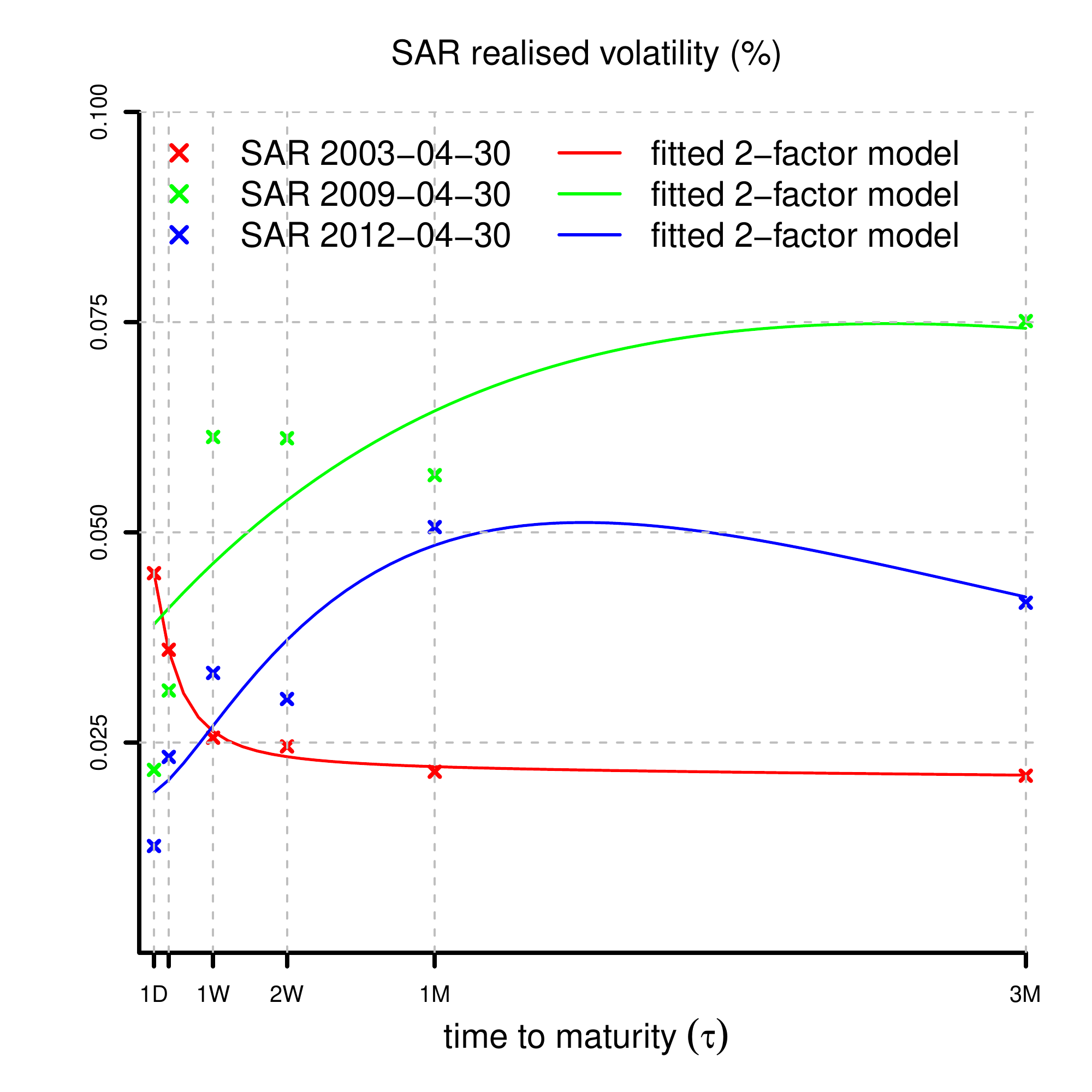}
\end{minipage}
\begin{minipage}[t]{0.42\textwidth}
\includegraphics[width=\textwidth]{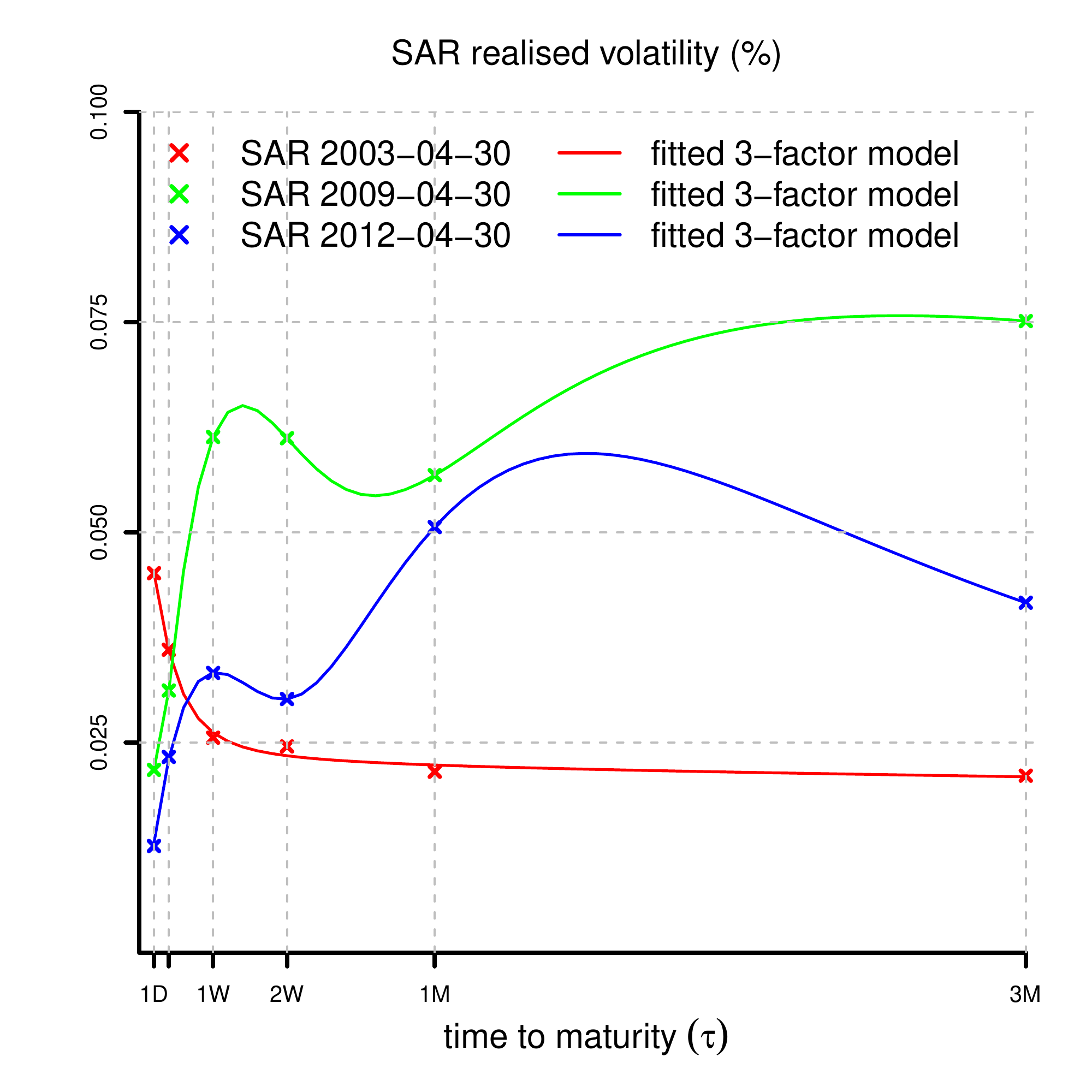}
\end{minipage}
\vspace{-15pt}
\caption{SAR realized volatility $\widehat{\mathrm{RCov}}(t,\tau,\tau)^{\frac{1}{2}}$ for $K=126$, $\tau=1,2,5,10,21,63$ and three observation dates compared to the realized volatility of the two- (lhs) and three-factor (rhs) Vasi\v cek model fitted by optimization~\eqref{eq: co-var estimate} for $M=6$, $\tau_1=1$, $\tau_2=2$, $\tau_3=5$, $\tau_4=10$, $\tau_5=21$, $\tau_6=63$ and $w_{ij}=1_{\{i=j\}}$. The three-factor model achieves an accurate fit.}\label{fig: number of factors sample dates}
\end{figure}

\subsubsection[Determination of Vasicek Parameters]{Determination of Vasi\v cek Parameters} Considering the results of Figure \ref{fig: number of factors all dates start}, we restrict ourselves from now on to three-factor Vasi\v cek models with parameters $\boldsymbol a,\boldsymbol b\in\R^3$ and:
\[
\begin{aligned}
\beta=\mathrm{diag}\left(\beta_{11},\beta_{22},\beta_{33}\right),\quad\alpha=\mathrm{diag}\left(\alpha_{22},\alpha_{22},\alpha_{33}\right),\quad\Sigma^{\frac{1}{2}}=\begin{pmatrix} \Sigma^{\frac{1}{2}}_{11} & 0 & 0 \\ \Sigma^{\frac{1}{2}}_{21} & \Sigma^{\frac{1}{2}}_{22} & 0 \\ \Sigma^{\frac{1}{2}}_{31} & \Sigma^{\frac{1}{2}}_{32} & \Sigma^{\frac{1}{2}}_{33} \end{pmatrix},
\end{aligned}
\]
where $-1\leq\beta_{11},\beta_{22},\beta_{33},\alpha_{11},\alpha_{22},\alpha_{33}\leq1$, $\Sigma^{\frac{1}{2}}_{11},\Sigma^{\frac{1}{2}}_{22},\Sigma^{\frac{1}{2}}_{33}>0$ and $\Sigma^{\frac{1}{2}}_{21},\Sigma^{\frac{1}{2}}_{31},\Sigma^{\frac{1}{2}}_{32}\in\R$.

\begin{figure}[p]
\vspace{-15pt}
\centering
\begin{subfigure}[t]{0.42\textwidth}
\includegraphics[width=\textwidth]{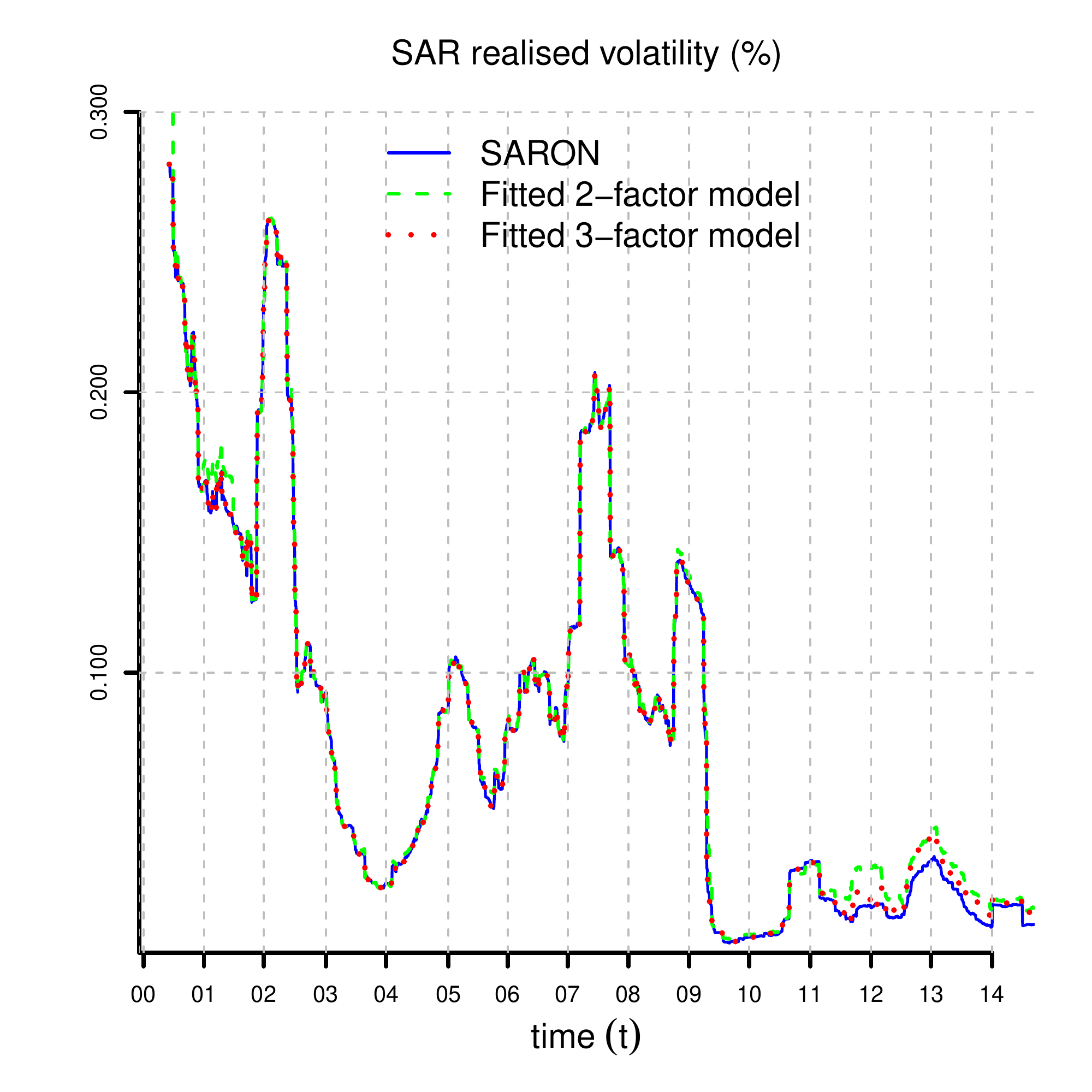}
\vspace{-25pt}
\caption{$\tau=1$.}
\end{subfigure}
\begin{subfigure}[t]{0.42\textwidth}
\includegraphics[width=\textwidth]{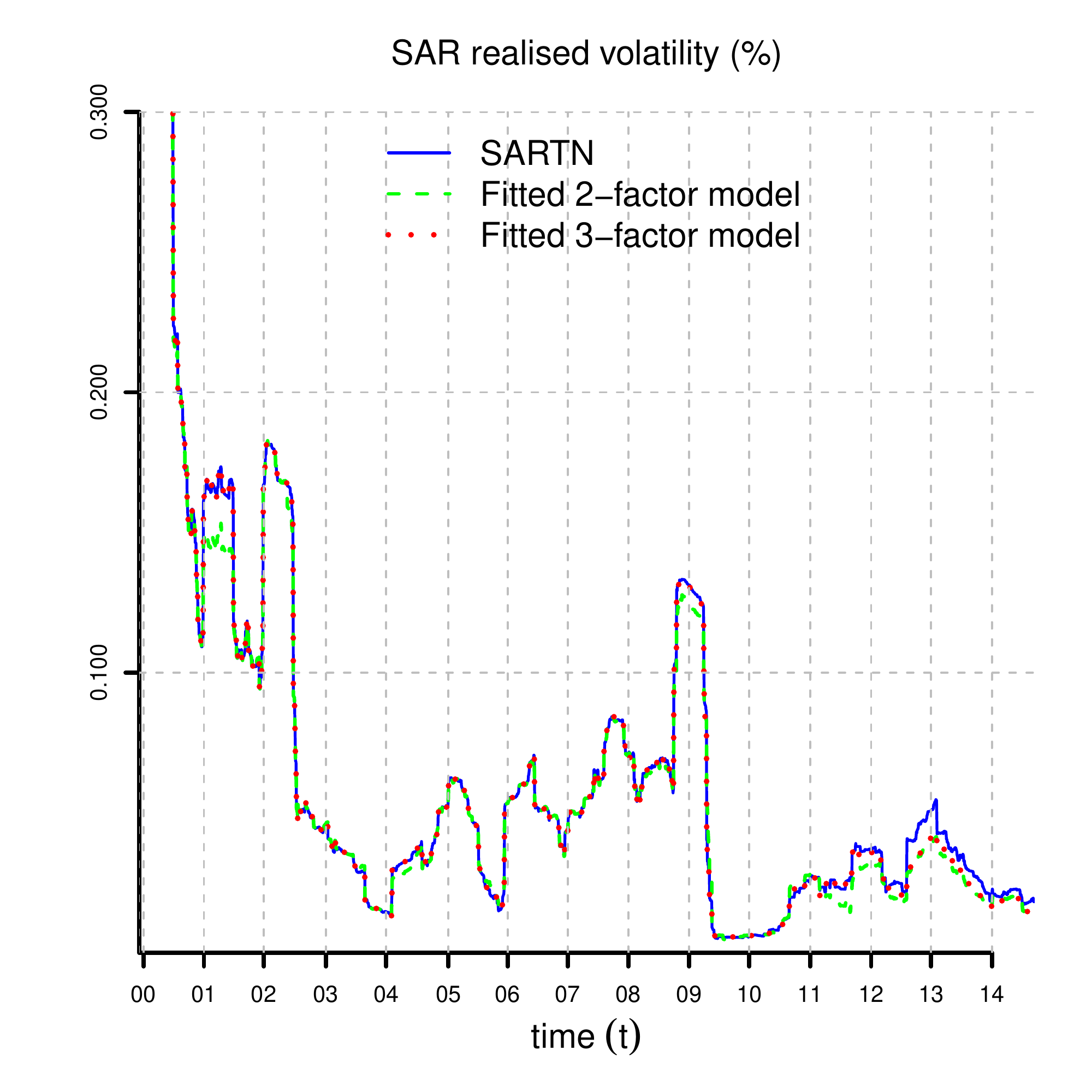}
\vspace{-25pt}
\caption{$\tau=2$.}
\end{subfigure}
\begin{subfigure}[t]{0.42\textwidth}
\includegraphics[width=\textwidth]{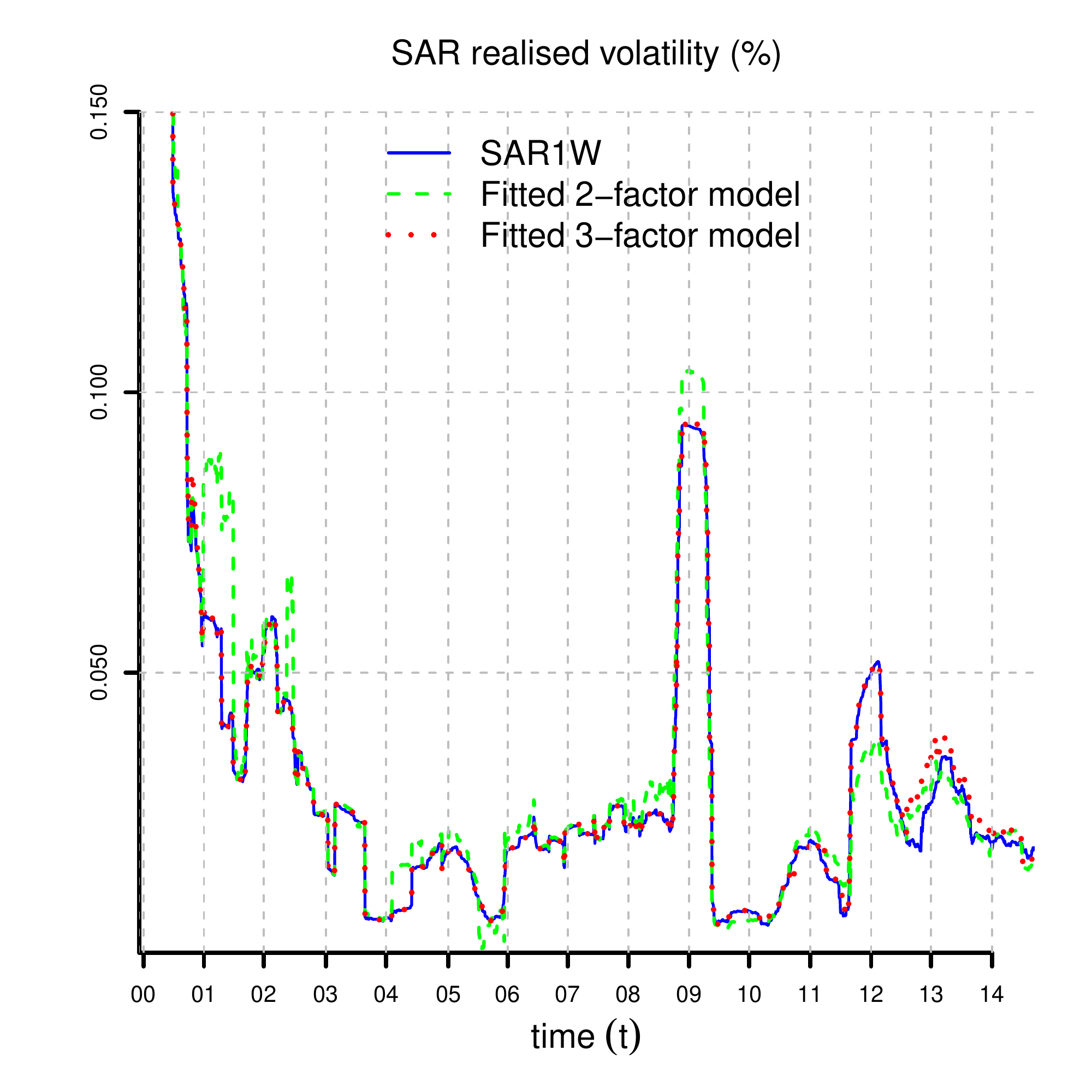}
\vspace{-25pt}
\caption{$\tau=5$.}
\end{subfigure}
\begin{subfigure}[t]{0.42\textwidth}
\includegraphics[width=\textwidth]{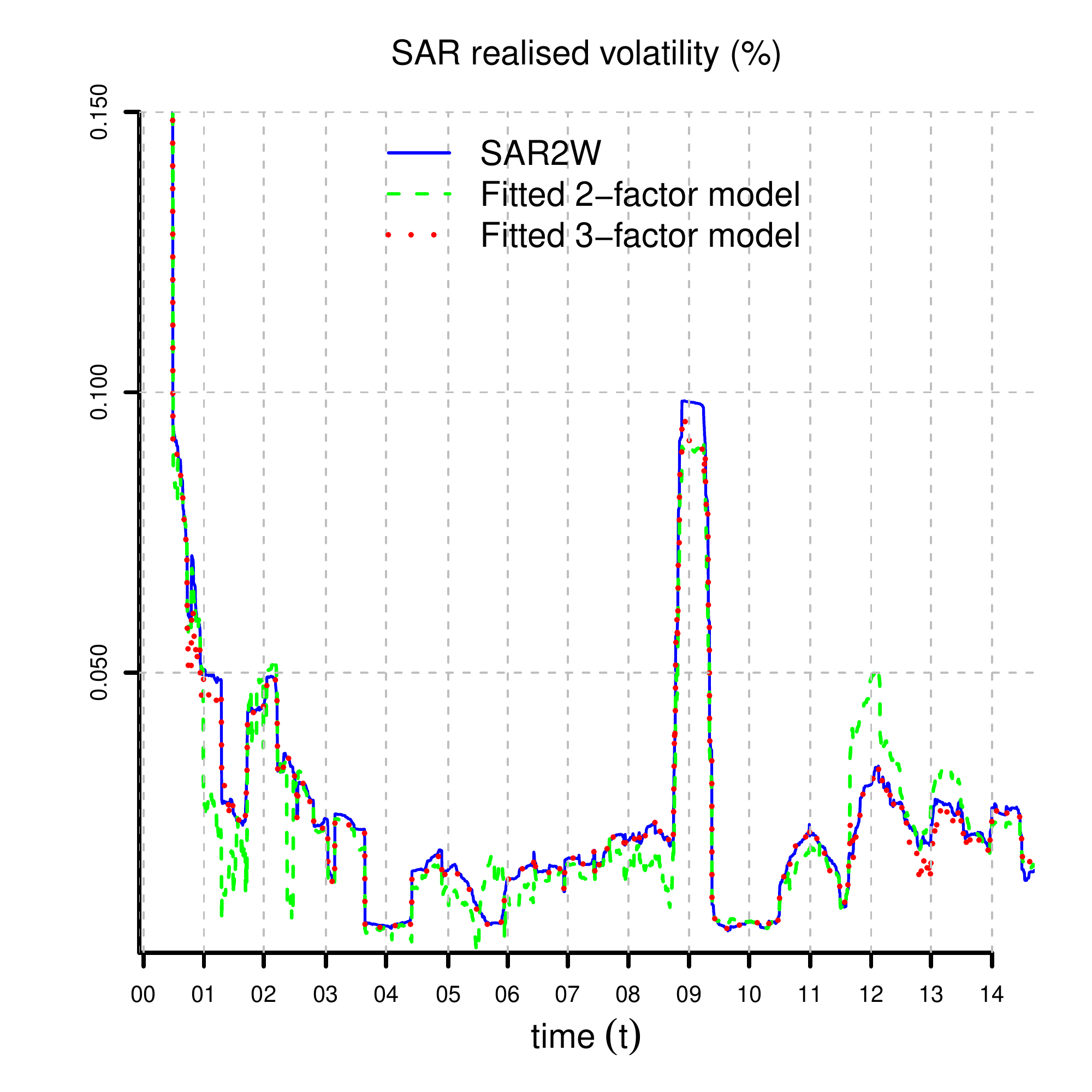}
\vspace{-25pt}
\caption{$\tau=10$.}
\end{subfigure}
\begin{subfigure}[t]{0.42\textwidth}
\includegraphics[width=\textwidth]{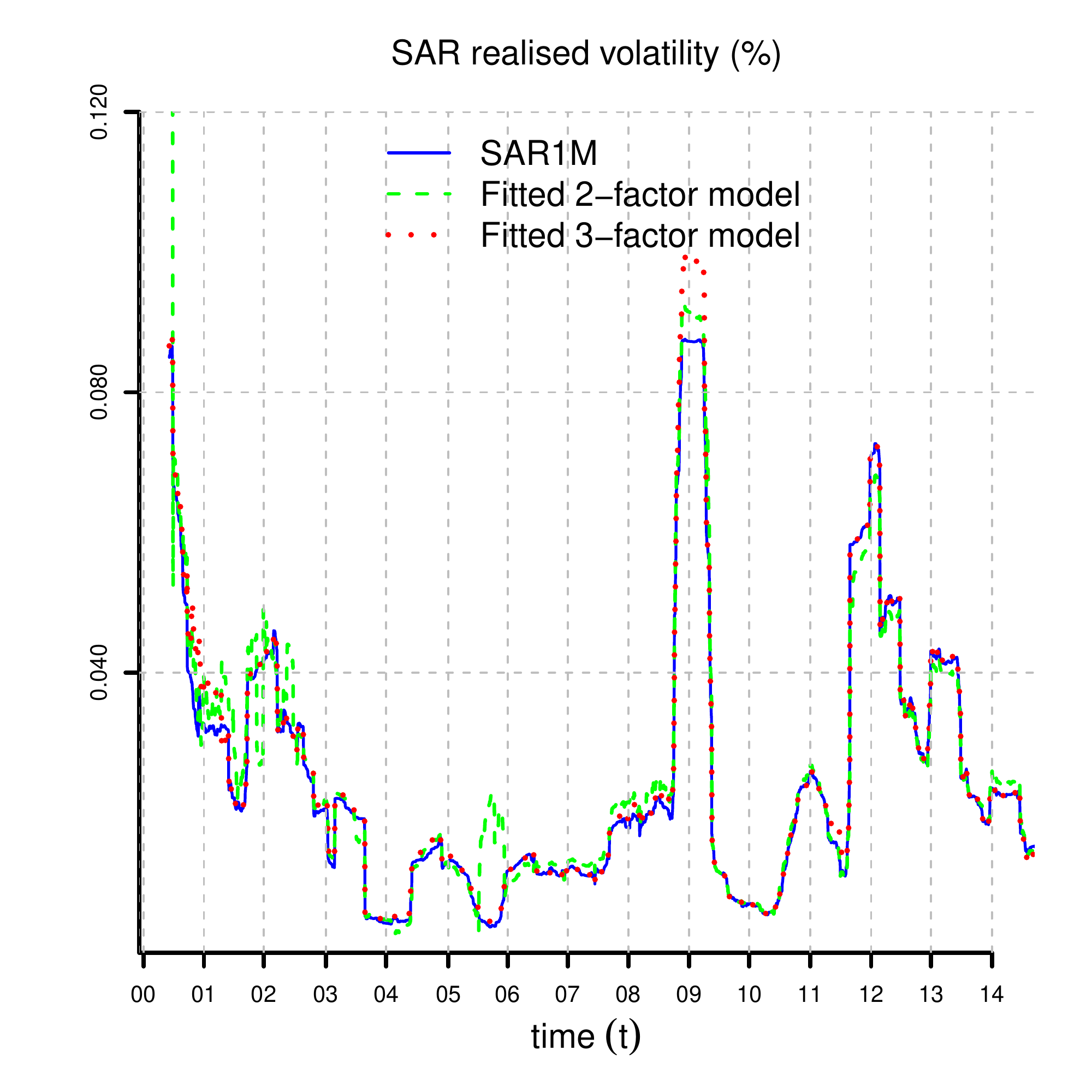}
\vspace{-25pt}
\caption{$\tau=21$.}
\end{subfigure}
\begin{subfigure}[t]{0.42\textwidth}
\includegraphics[width=\textwidth]{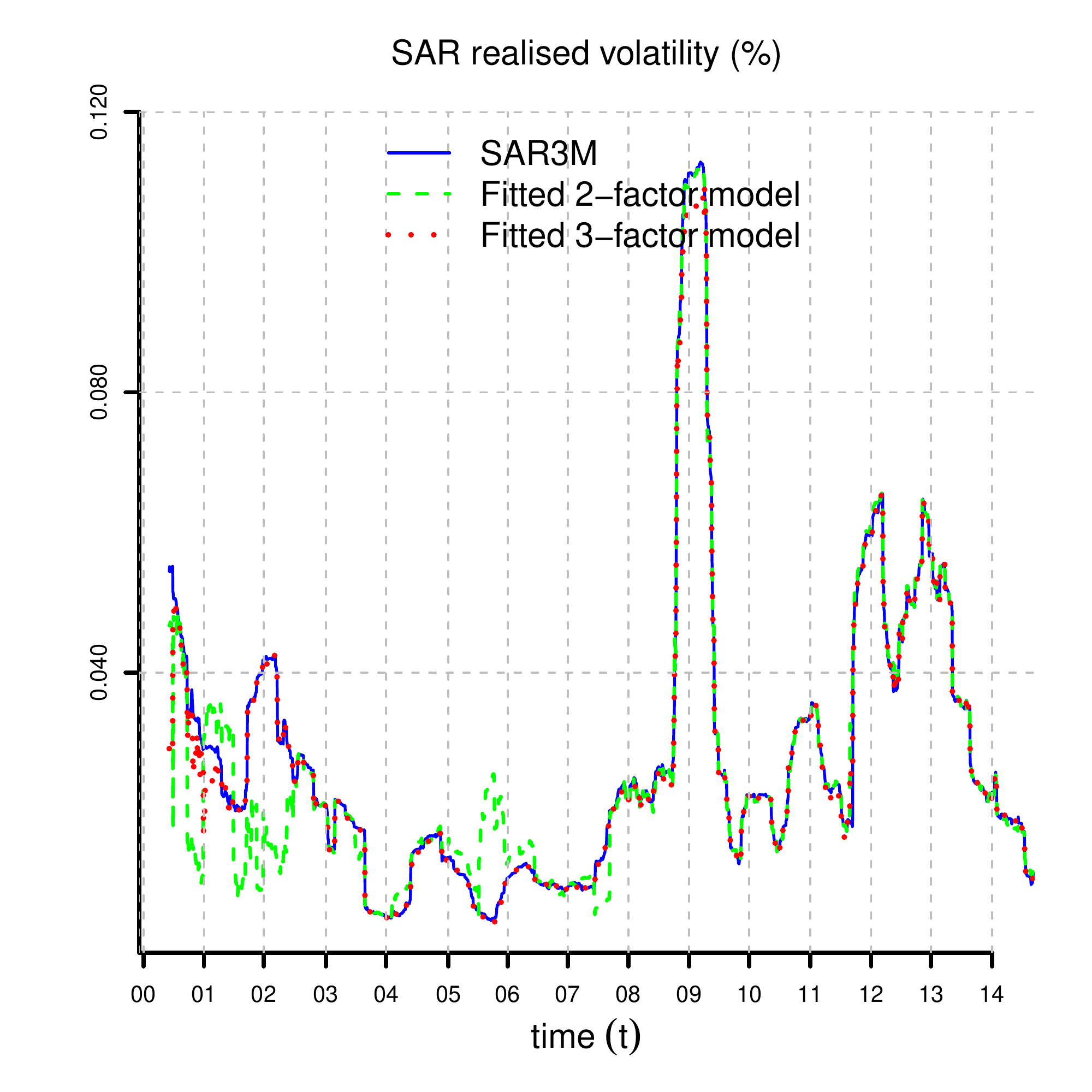}
\vspace{-25pt}
\caption{$\tau=63$.}
\end{subfigure}
\caption{SAR realized volatility $\widehat{\mathrm{RCov}}(t,\tau,\tau)^{\frac{1}{2}}$ for $K=126$ and a selection of times to maturity $\tau$ compared to the realized volatility of the two- and three-factor Vasi\v cek models fitted by optimization~\eqref{eq: co-var estimate} for $M=6$, $\tau_1=1$, $\tau_2=2$, $\tau_3=5$, $\tau_4=10$, $\tau_5=21$, $\tau_6=63$ and $w_{ij}=1_{\{i=j\}}$.}\label{fig: number of factors all dates start}
\end{figure}

\begin{figure}[p]
\vspace{-25pt}
\centering
\begin{subfigure}{\textwidth}
\begin{minipage}[t]{0.42\textwidth}
\includegraphics[width=\textwidth]{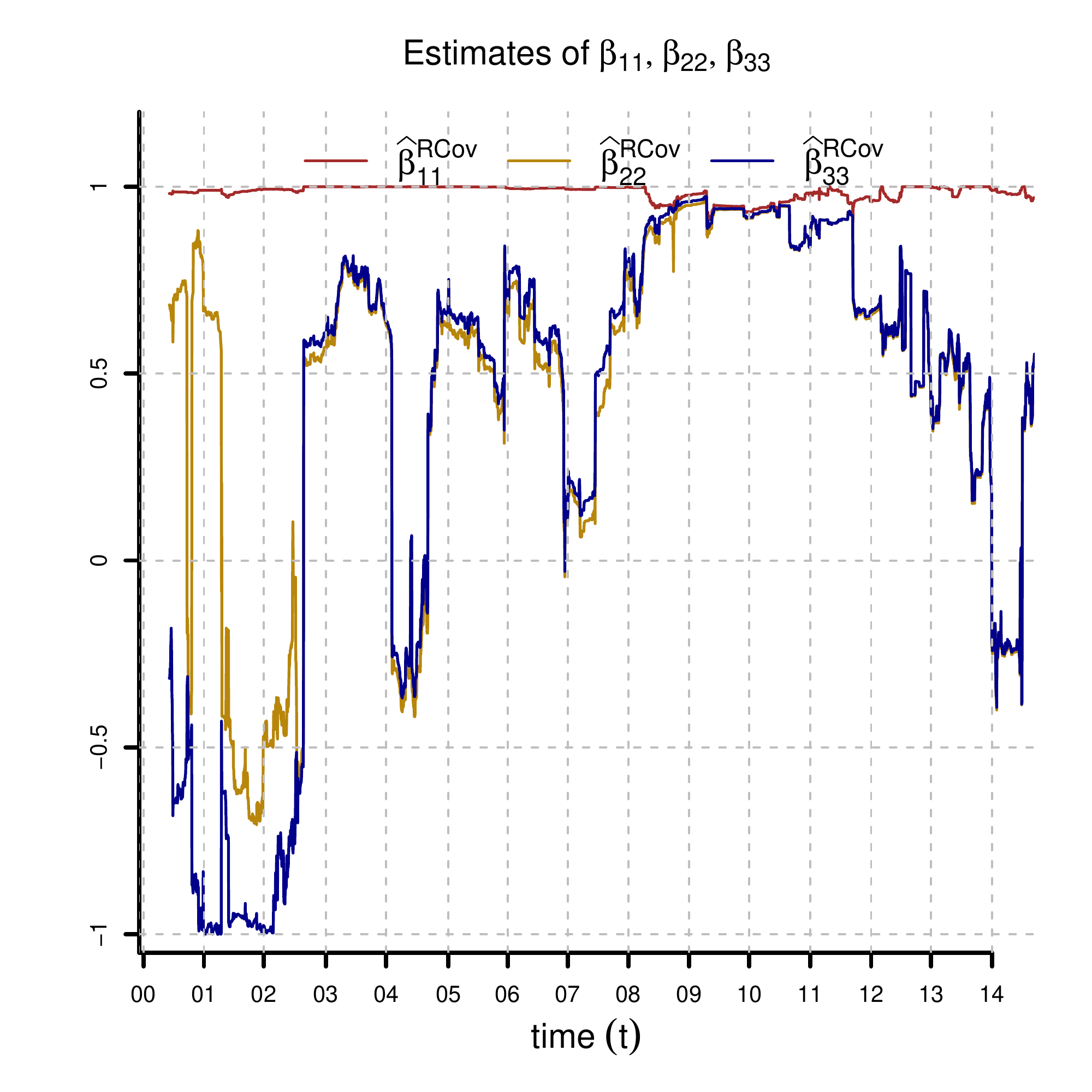}
\end{minipage}
\begin{minipage}[t]{0.42\textwidth}
\includegraphics[width=\textwidth]{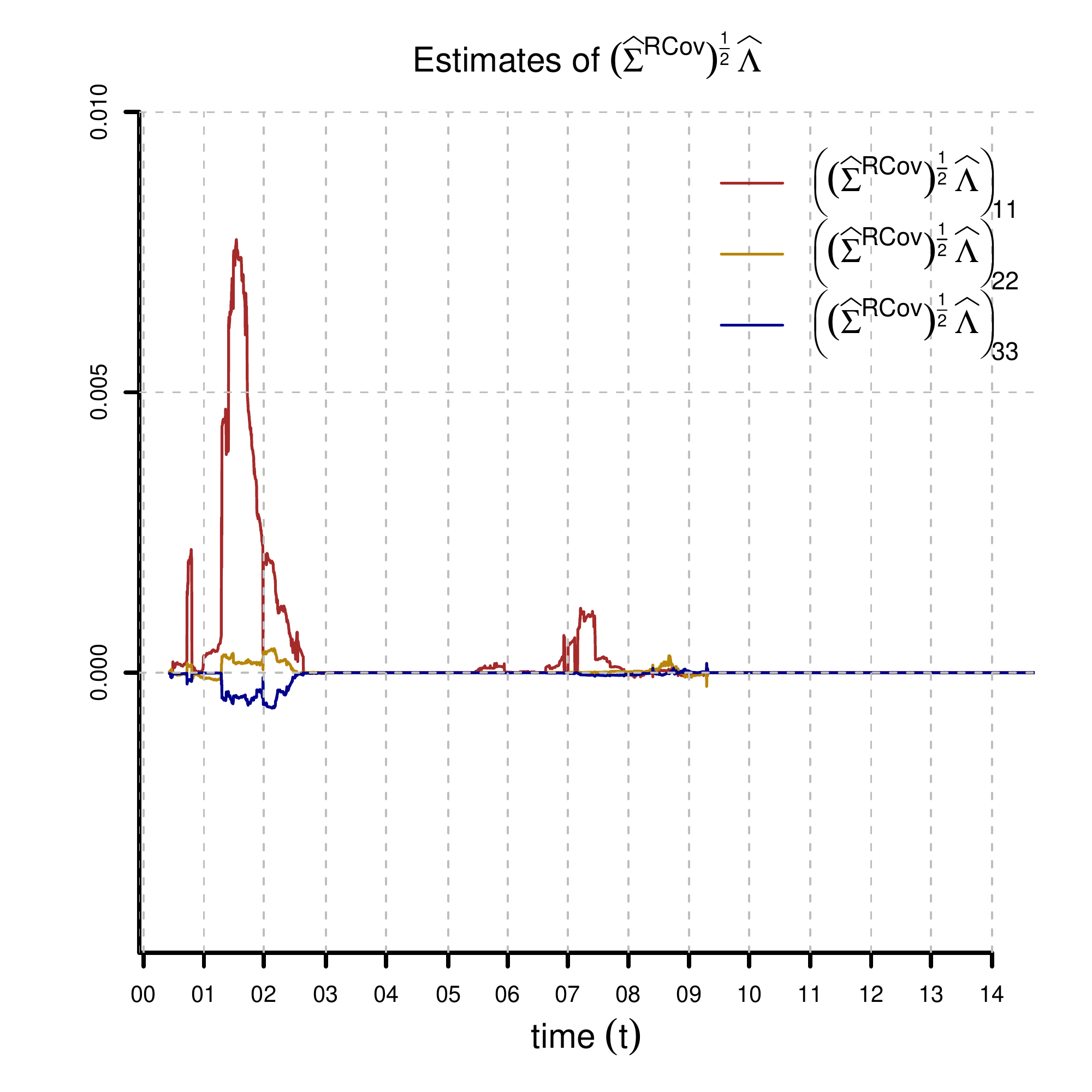}
\end{minipage}
\vspace{-10pt}
\caption{Note the values of $\beta_{ij}$ close to one (i.e., slow mean reversion), which are reasonable on the daily time grid. The differences in $\beta_{ij}$ under the risk-neutral and real-world measures are negligible.}\label{fig: parameters start a}
\end{subfigure}
\begin{subfigure}{\textwidth}
\begin{minipage}[t]{0.42\textwidth}
\includegraphics[width=\textwidth]{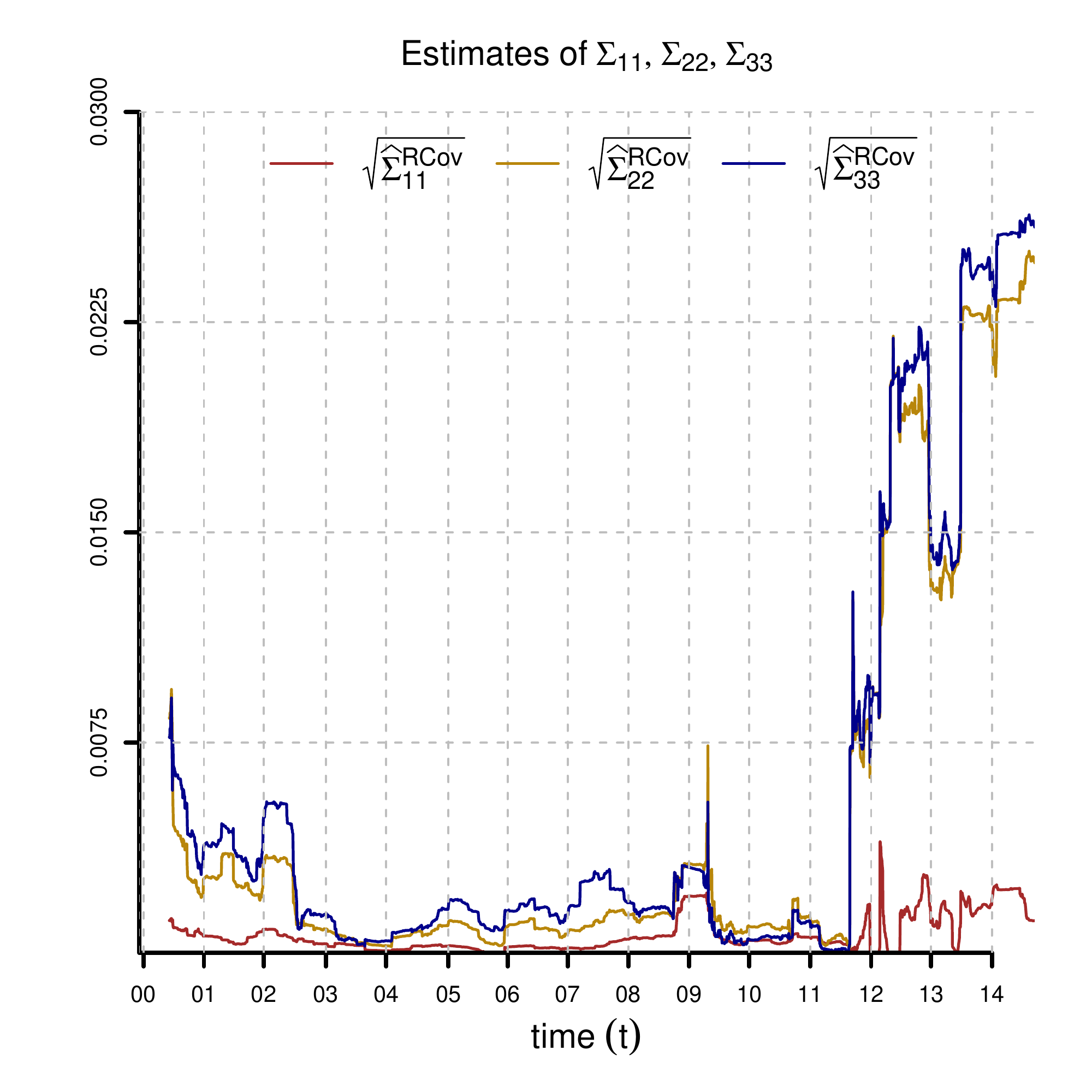}
\end{minipage}
\begin{minipage}[t]{0.42\textwidth}
\includegraphics[width=\textwidth]{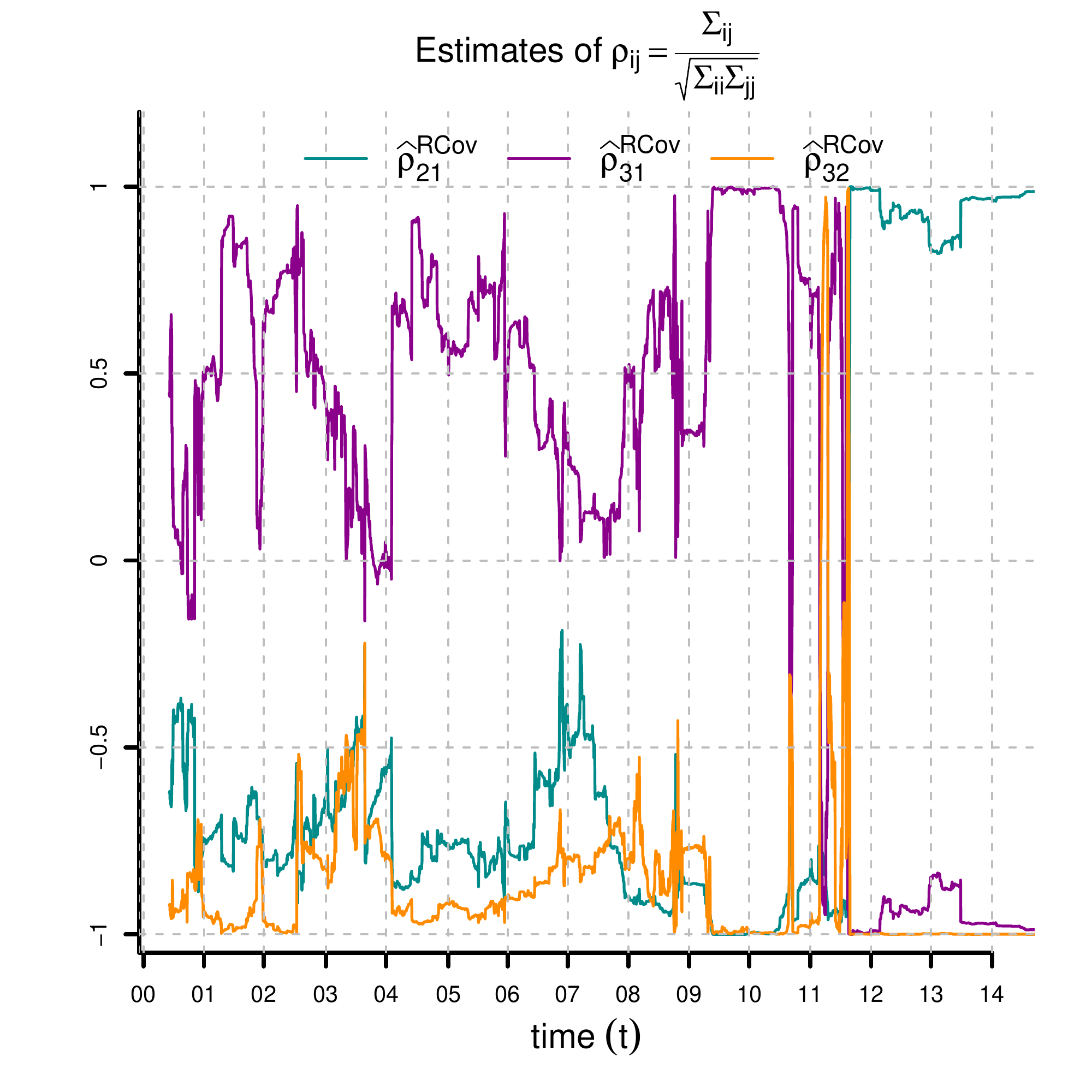}
\end{minipage}
\vspace{-10pt}
\caption{Note the large spikes in the volatilities and strong correlations among the factors during the European sovereign debt crisis and after the SNB intervention in 2011.}\label{fig: sigma}
\end{subfigure}
\begin{subfigure}{\textwidth}
\begin{minipage}[t]{0.42\textwidth}
\includegraphics[width=\textwidth]{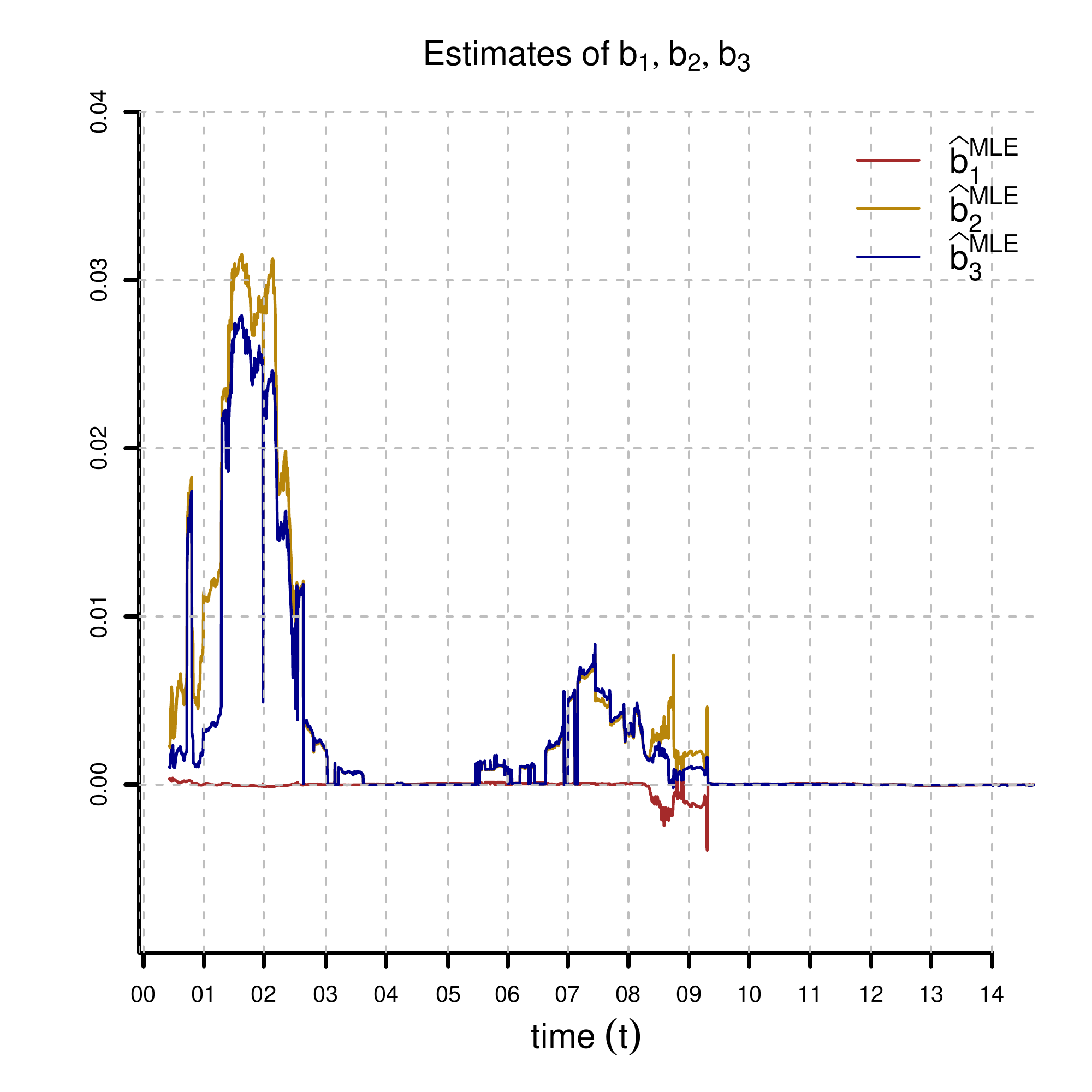}
\end{minipage}
\begin{minipage}[t]{0.42\textwidth}
\includegraphics[width=\textwidth]{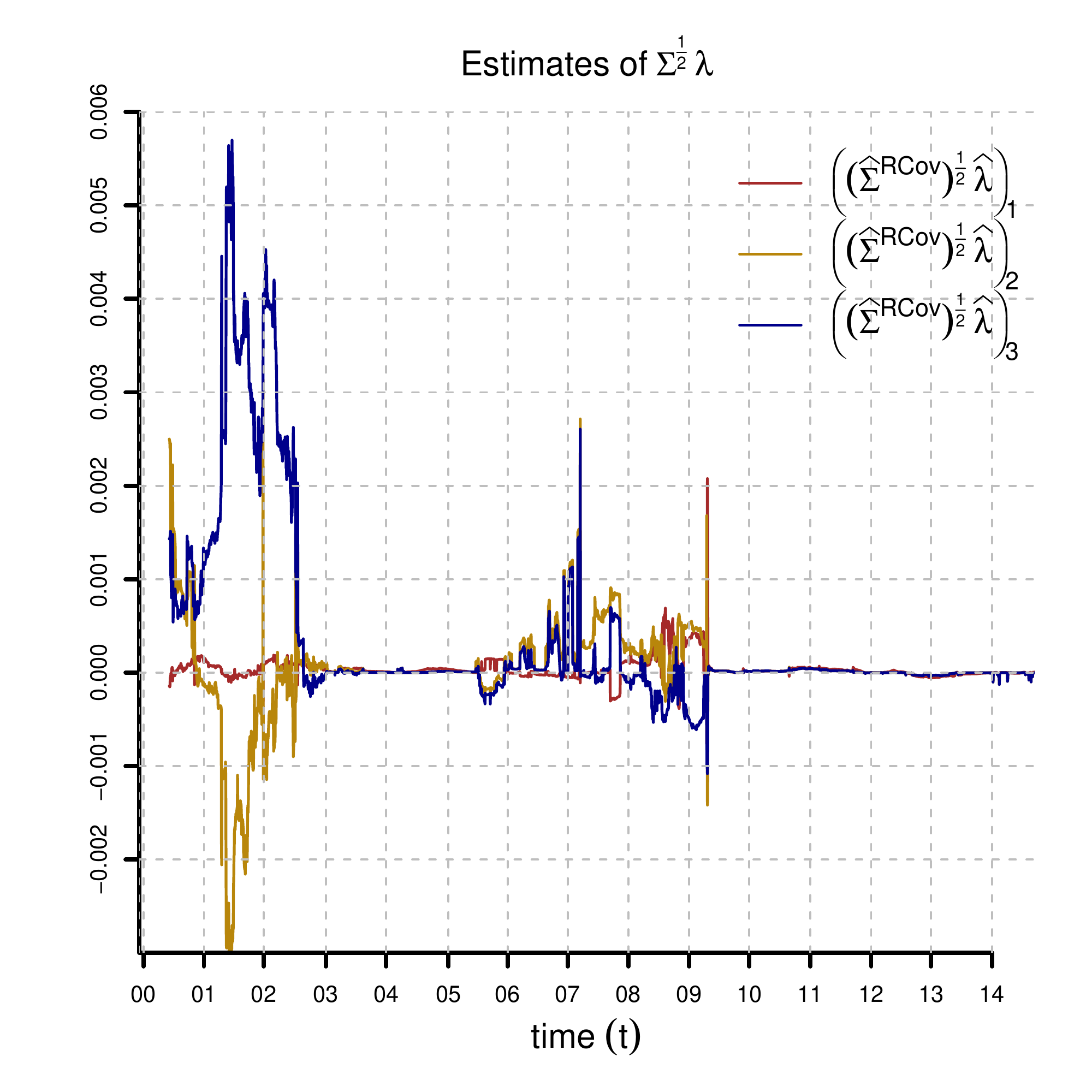}
\end{minipage}
\vspace{-10pt}
\caption{Note the considerable difference between $\boldsymbol b$ and $\boldsymbol a$ in 2000--2002 and 2006--2009.}\end{subfigure}
\vspace{-5pt}
\caption{Parameter estimation by optimizations \eqref{eq: co-var estimate} and \eqref{eq: mle} in the three-factor model for $K=126$, $M=6$, $\tau_1=1$, $\tau_2=2$, $\tau_3=5$, $\tau_4=10$, $\tau_5=21$, $\tau_6=63$, $w_{ij}=1_{\{i=j\}}$ and $S=10^{-5}\cdot \mathds{1}$.}\label{fig: parameters start}
\end{figure}

In Figure \ref{fig: parameters start}, we plot the numerical solutions of optimizations \eqref{eq: co-var estimate} and \eqref{eq: mle} for all observation dates. The parameters are reasonable for most of the observation dates. We observe that the estimates of $\beta_{11}$ are close to one for all observation dates. Our values for the speed of mean reversion are reasonable on a daily time grid. Note that $\beta$ scales as $\beta^d$ on a $d$-days time grid; see Section~\ref{calibration real world 2}. The speeds of mean reversion of $X_2$ and $X_3$ are higher than that of $X_1$ for most of the observation dates. We also see that the volatility of $X_1$ is lower than that of $X_2$ and $X_3$. In 2011, we observe large spikes in the factor volatilities. Starting from 2011, we have a period with strong correlations among the factors. From these results, we conclude that the three-factor Vasi\v cek model is reasonable for Swiss interest rates. Particularly challenging for the estimation is the period 2011--2014 of low interest rates following the European sovereign debt crisis and the SNB intervention. In Figure \ref{fig: parameters start a} (rhs), we observe that the difference in the speeds of mean-reversion under the risk-neutral and real-world measures is negligible. The difference between $\boldsymbol b$ and $\boldsymbol a$ is considerable in certain time periods. From the estimation results, we conclude that a constant market price of risk assumption is reasonable and set from now on $\Lambda=0$. In Figure \ref{fig: loglikelihood}, we compute the objective function of optimization~\eqref{eq: mle} for $(\boldsymbol b,\beta,\Sigma,\boldsymbol a,\alpha)=(\boldsymbol 0, \widehat{\beta}^{\mathrm{RCov}},\widehat{\Sigma}^{\mathrm{RCov}},\boldsymbol 0, \widehat{\beta}^{\mathrm{RCov}})$ and compare it to the numerical solution $(\widehat{\boldsymbol b}^{\mathrm{MLE}}, \widehat{\beta}^{\mathrm{RCov}},\widehat{\Sigma}^{\mathrm{RCov}},\widehat{\boldsymbol a}^{\mathrm{MLE}}, \widehat{\beta}^{\mathrm{RCov}})$. We observe that in 2003--2005 and 2010--2014, the parameter configuration $(\boldsymbol 0, \widehat{\beta}^{\mathrm{RCov}},\widehat{\Sigma}^{\mathrm{RCov}},\boldsymbol 0, \widehat{\beta}^{\mathrm{RCov}})$ is nearly optimal. In these periods, we have very low interest rates, and therefore, estimates of $\boldsymbol b$ and $\boldsymbol a$ close to zero are reasonable. Given the estimated parameters, we calibrate the Hull--White extension by equation $\eqref{eq: re-calibration step 2}$ to the full yield curve interpolated from SAR and SWCNB; see Figure \ref{fig: hwe}. We point out that our fitting method is not a purely statistical procedure; rather, it is a combination of estimation and calibration in accordance with the paradigm of robust calibration, as explained in \cite{Harms}.

\begin{figure}[h]
\centering
\begin{minipage}[t]{0.42\textwidth}
\includegraphics[width=\textwidth]{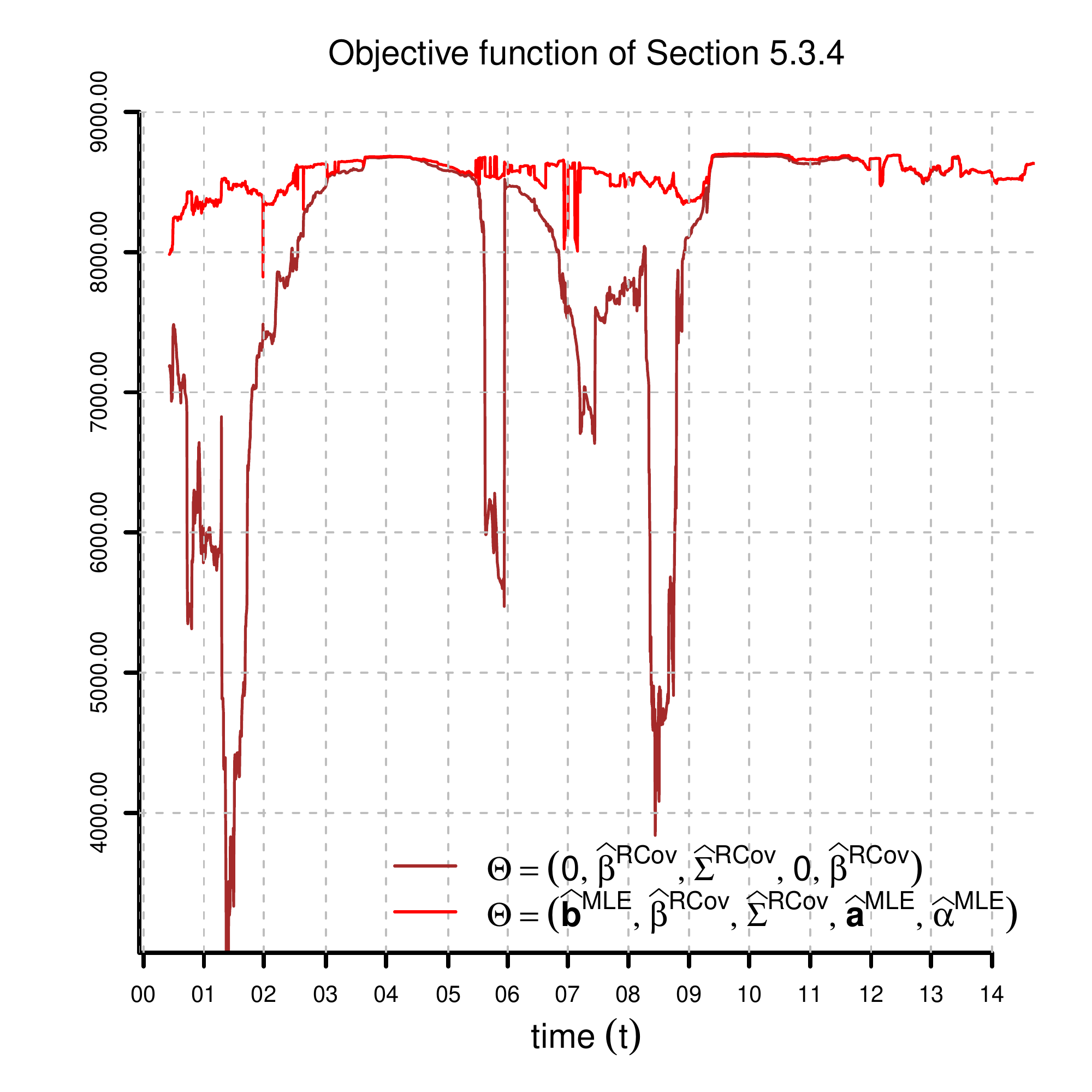}
\end{minipage}
\begin{minipage}[t]{0.42\textwidth}
\includegraphics[width=\textwidth]{./figures/mle_calibration_sar_lambda_3fac.pdf}
\end{minipage}
\vspace{-10pt}
\caption{Objective function $\log\mathcal{L}_t$ (lhs) and values of $(\Sigma^{\frac{1}{2}}\boldsymbol\lambda)_{1}=b_{1}-a_{1}$, $(\Sigma^{\frac{1}{2}}\boldsymbol\lambda)_{2}=b_{2}-a_{2}$ and \mbox{$(\Sigma^{\frac{1}{2}}\boldsymbol\lambda)_{3}=b_{3}-a_{3}$ (rhs)} given by optimization~\eqref{eq: mle} in the three-factor model for $K=126$, $M=6$, $\tau_1=1$, $\tau_2=2$, $\tau_3=5$, $\tau_4=10$, $\tau_5=21$, $\tau_6=63$, $w_{ij}=1_{\{i=j\}}$ and $S=10^{-5}\cdot \mathds{1}$. We compare the value of the objective function for $(\boldsymbol b,\beta,\Sigma,\boldsymbol a,\alpha)=(\boldsymbol 0,\beta^{\mathrm{RCov}},\Sigma^{\mathrm{RCov}},\boldsymbol 0,\alpha^{\mathrm{RCov}})$ and the numerical solution of the optimization. The configuration $(\boldsymbol 0,\beta^{\mathrm{RCov}},\Sigma^{\mathrm{RCov}},\boldsymbol 0,\alpha^{\mathrm{RCov}})$ is almost optimal in low interest rate~times.}\label{fig: loglikelihood}
\end{figure}
\begin{figure}[h]
\centering
\begin{minipage}[t]{0.42\textwidth}
\includegraphics[width=\textwidth]{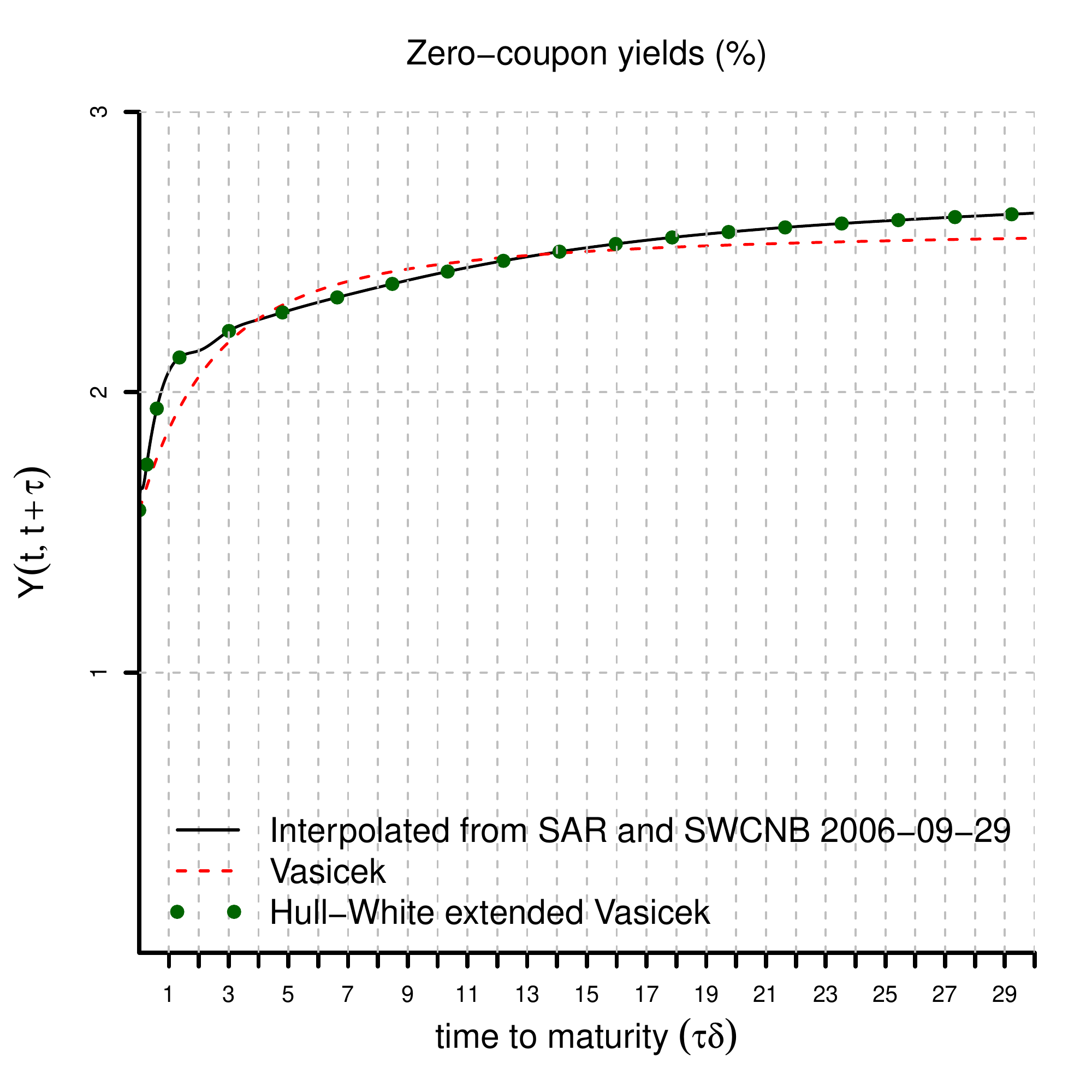}
\end{minipage}
\begin{minipage}[t]{0.42\textwidth}
\includegraphics[width=\textwidth]{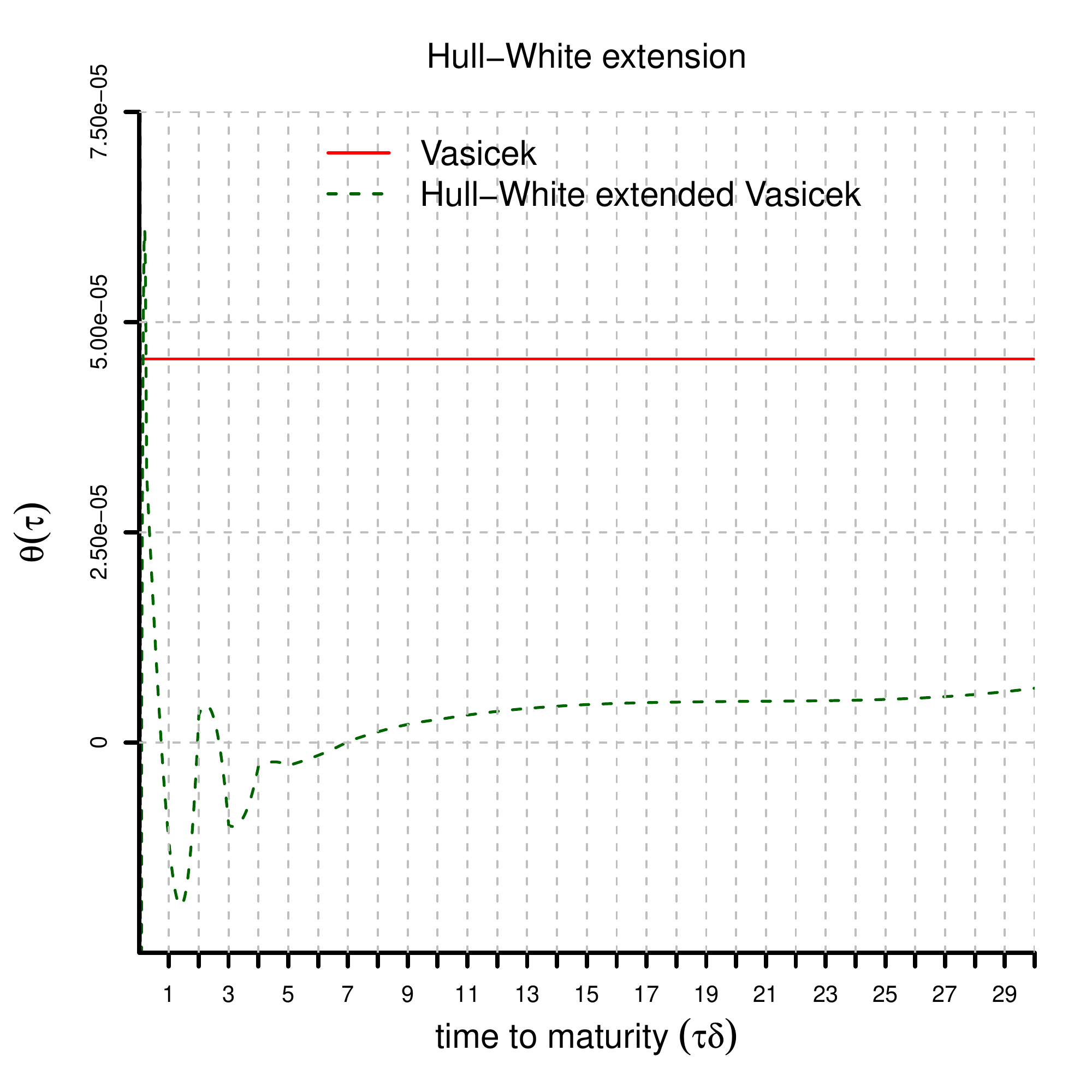}
\end{minipage}
\vspace{-10pt}
\caption{Three-factor Hull--White extended Vasi\v cek yield curve (lhs) and Hull--White extension $\theta$ (rhs) as of 29 September 2006. The parameters are estimated as in Figure \ref{fig: parameters start}. The initial factors are obtained from the Kalman filter for the estimated parameters. The calibration of the Hull--White extension requires yields on a time to maturity grid of size $\Delta$. These are interpolated from SAR and SWCNB using cubic splines.}\label{fig: hwe}
\end{figure}

\subsubsection[Selection of a Model for the Vasicek Parameters]{Selection of a Model for the Vasi\v cek Parameters} In the following, we use the CRC approach to construct a modification of the Vasi\v cek model with stochastic volatility. We model the process $(\Sigma(t))_{t\in\mathbb N_0}$ by a Heston-like \cite{Heston} approach. We assume deterministic correlations among the factors and stochastic volatility given by: 
\[
\begin{pmatrix} \Sigma_{11}(t) \\\Sigma_{22}(t) \\ \Sigma_{33}(t) \end{pmatrix}=\boldsymbol\varphi+\phi\begin{pmatrix} \Sigma_{11}(t-1) \\\Sigma_{22}(t-1) \\ \Sigma_{33}(t-1)\end{pmatrix}+\begin{pmatrix} \sqrt{\Sigma_{11}(t-1)} & 0 & 0 \\ 0 & \sqrt{\Sigma_{22}(t-1)} & 0\\ 0 & 0 & \sqrt{\Sigma_{33}(t-1)}\end{pmatrix}\Phi^{\frac{1}{2}}\widetilde{\boldsymbol\varepsilon}(t),
\]
where $\boldsymbol\varphi\in\mathbb R^3_+$, $\phi=\mathrm{diag}\left(\phi_{11},\phi_{22},\phi_{33}\right)\in\mathbb R^{3\times 3}$, $\Phi^{\frac{1}{2}}\in\mathbb R^{3\times 3}$ non-singular, and for each $t\in\N$, $\widetilde{\boldsymbol\varepsilon}(t)$ has a standard Gaussian distribution under $\p$, conditionally given $\mathcal F(t-1)$. Moreover, we assume that $\left(\boldsymbol\varepsilon(t),\widetilde{\boldsymbol\varepsilon}(t)\right)$ is multivariate Gaussian under $\p$, conditionally given $\mathcal F(t-1)$. Note that $\boldsymbol\varepsilon(t)$ and $\widetilde{\boldsymbol\varepsilon}(t)$ are allowed to be correlated. The matrix valued process $(\Sigma(t))_{t\in\mathbb N_0}$ is constructed combining this stochastic volatility model with fixed correlation coefficients. This model is able to capture the stylized fact that volatility appears to be more noisy in high volatility periods; see Figure \ref{fig: sigma}. 

We use the volatility time series of Figure \ref{fig: sigma} to specify $\boldsymbol\varphi$, $\phi$ and $\Phi$. We rewrite the equation for the evolution of the volatility as:
\[
\frac{\Sigma_{ii}(t)}{\sqrt{\Sigma_{ii}(t-1)}}=\frac{\varphi_i}{\sqrt{\Sigma_{ii}(t-1)}}+\phi_{ii}\sqrt{\Sigma_{ii}(t-1)}+(\Phi^{\frac{1}{2}}\widetilde{\boldsymbol\varepsilon}(t))_{i},\quad i=1,2,3,
\]
and use least square regression to estimate $\boldsymbol\varphi$, $\phi$ and $\Phi$. From the regression residuals, we estimate the correlations between $\boldsymbol\varepsilon(t)$ and $\widetilde{\boldsymbol\varepsilon}(t)$. Figures \ref{fig: parameter crc start}--\ref{fig: parameter crc end} show the estimates of $\boldsymbol\varphi$, $\phi$ and $\Phi$. 

\subsection{Simulation and Back-Testing}

Section~\ref{sec: model selection} provides a full specification of the three-factor Vasi\v cek CRC model under the risk-neutral and real-world probability measures. Various model quantities of interest in applications can then be calculated by simulation. 

\begin{figure}[p]
\centering
\begin{minipage}[t]{0.42\textwidth}
\includegraphics[width=\textwidth]{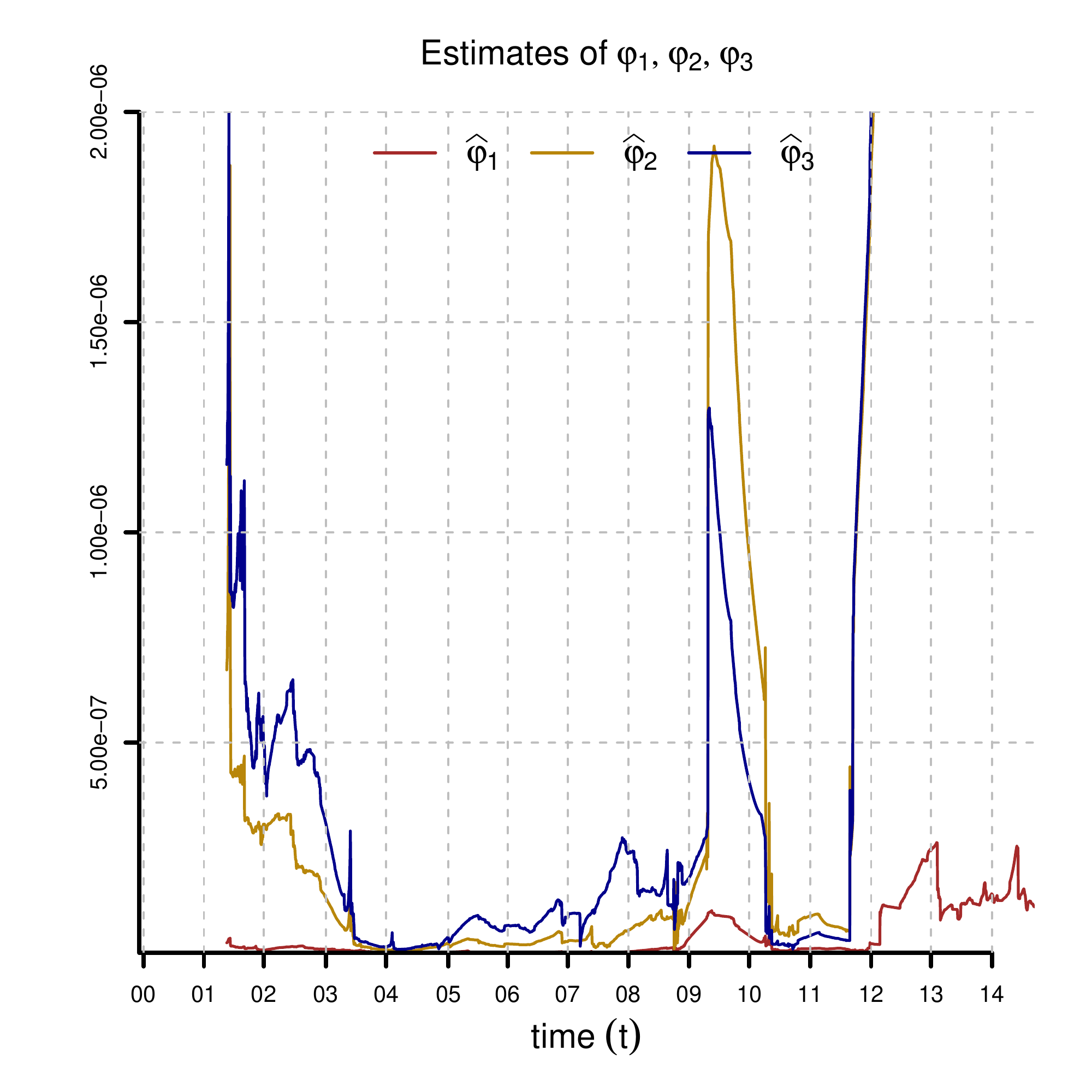}
\end{minipage}
\begin{minipage}[t]{0.42\textwidth}
\includegraphics[width=\textwidth]{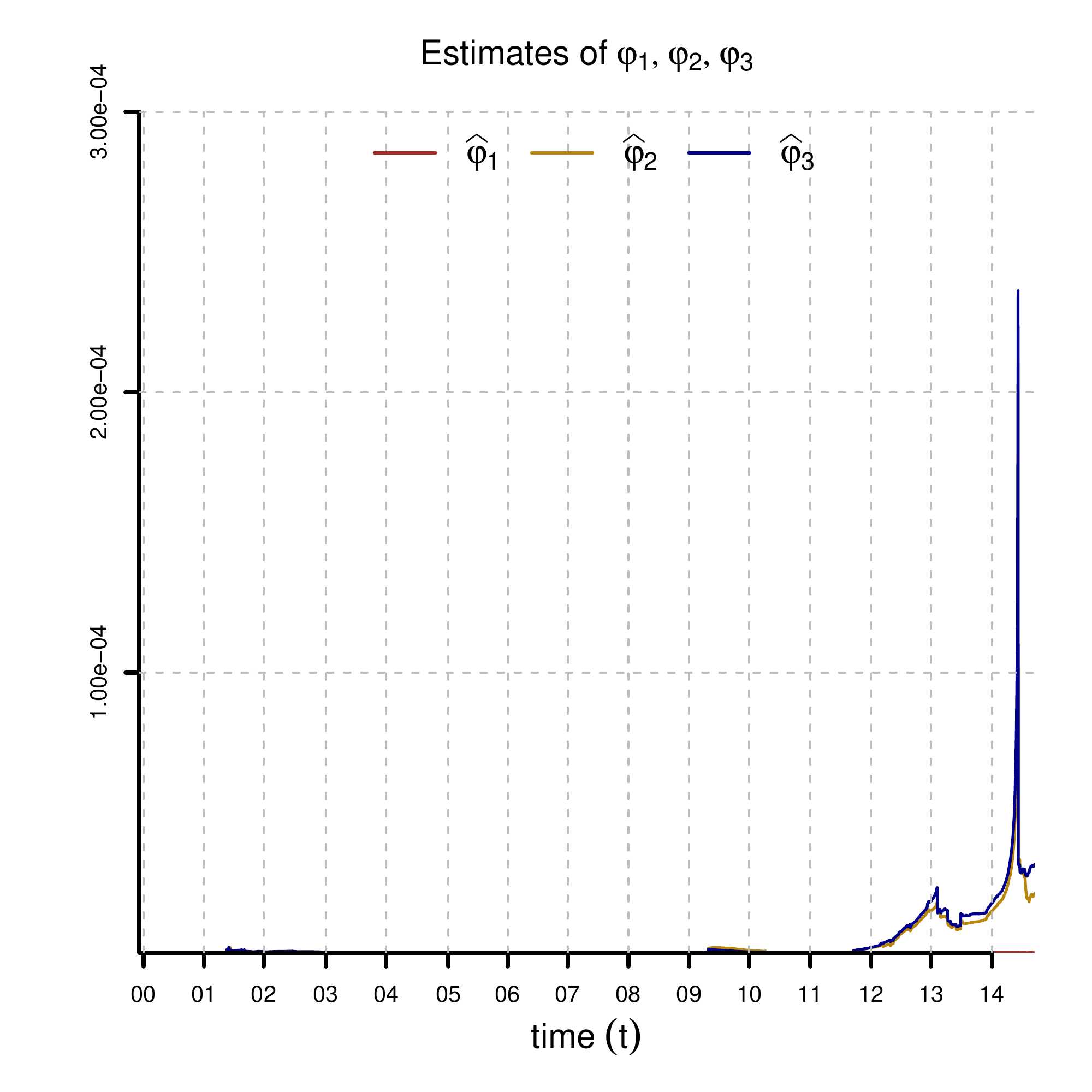}
\end{minipage}
\vspace{-15pt}
\caption{Estimation of $\varphi_{1}$, $\varphi_{2}$ and $\varphi_{3}$ by least square regression (two different scales). We use a time window of 252 observations for the regression.}\label{fig: parameter crc start}\vspace{-5pt}
\end{figure}
\begin{figure}[p]
\centering
\begin{minipage}[t]{0.42\textwidth}
\includegraphics[width=\textwidth]{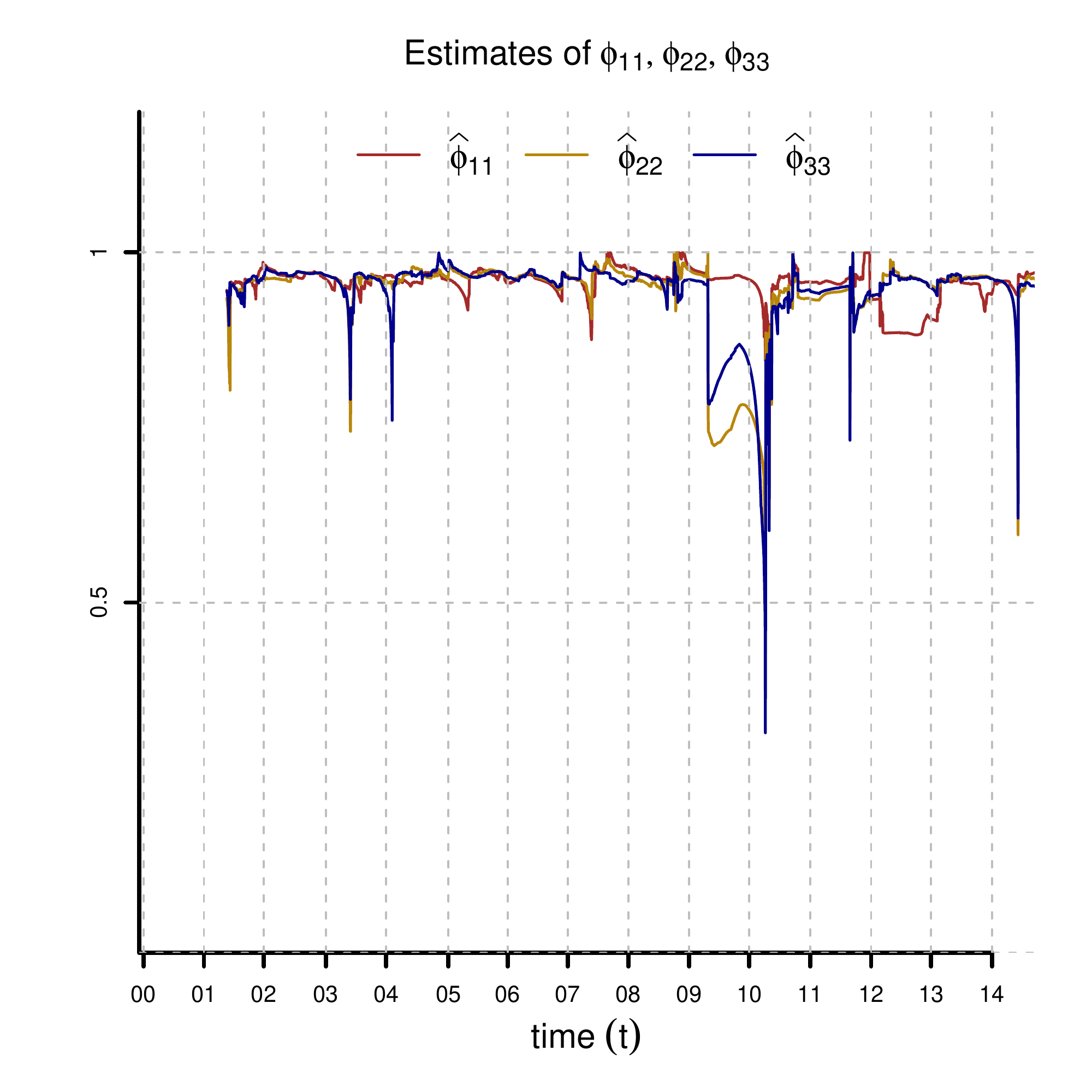}
\end{minipage}
\begin{minipage}[t]{0.42\textwidth}
\includegraphics[width=\textwidth]{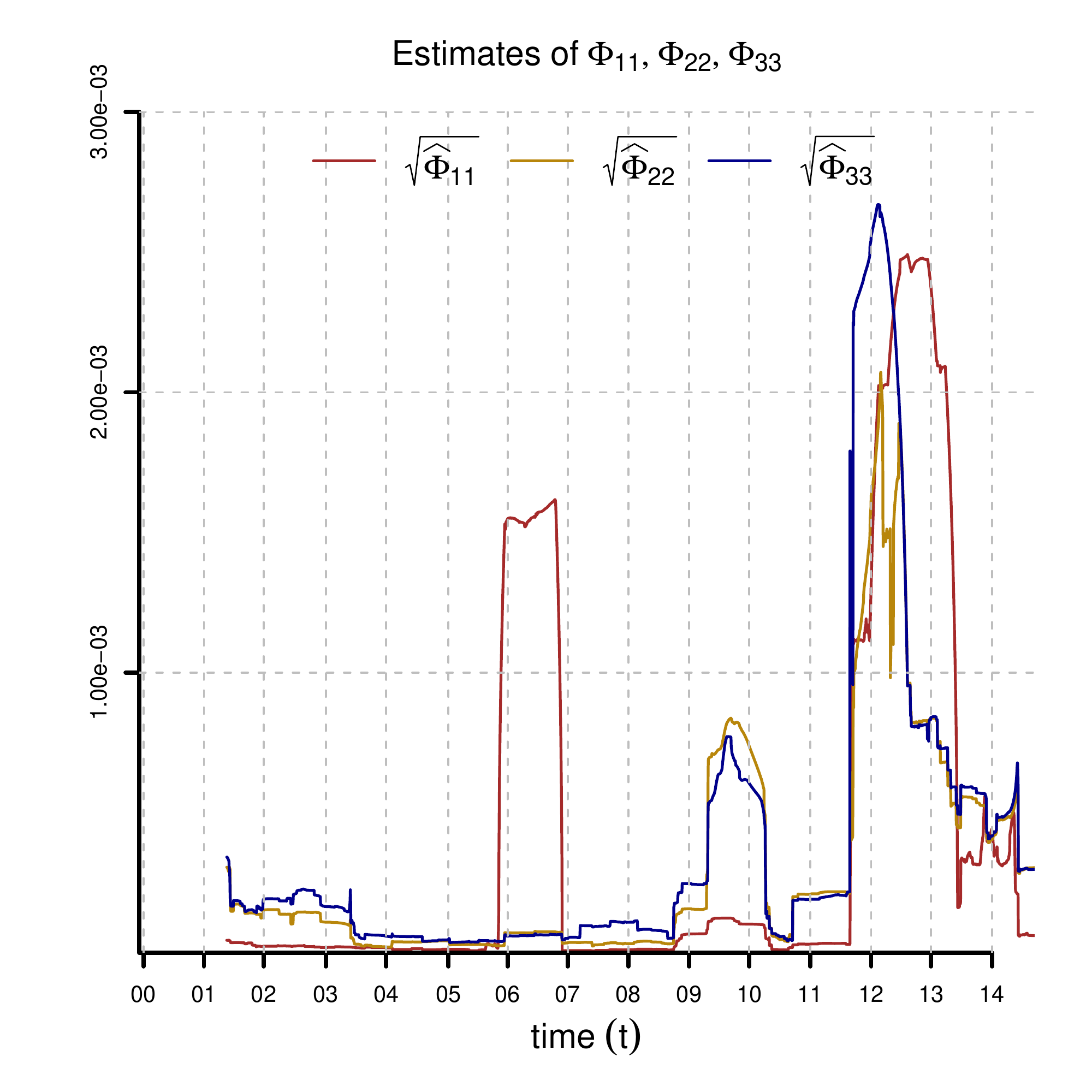}
\end{minipage}
\vspace{-15pt}
\caption{Estimation of $\phi_{11}$, $\phi_{22}$ and $\phi_{33}$ (lhs) and $\Phi_{11}$, $\Phi_{22}$ and $\Phi_{33}$ (rhs) by least square regression. We use a time window of 252 observations for the regression.}\vspace{-5pt}
\end{figure}
\begin{figure}[p]
\centering
\begin{minipage}[t]{0.42\textwidth}
\includegraphics[width=\textwidth]{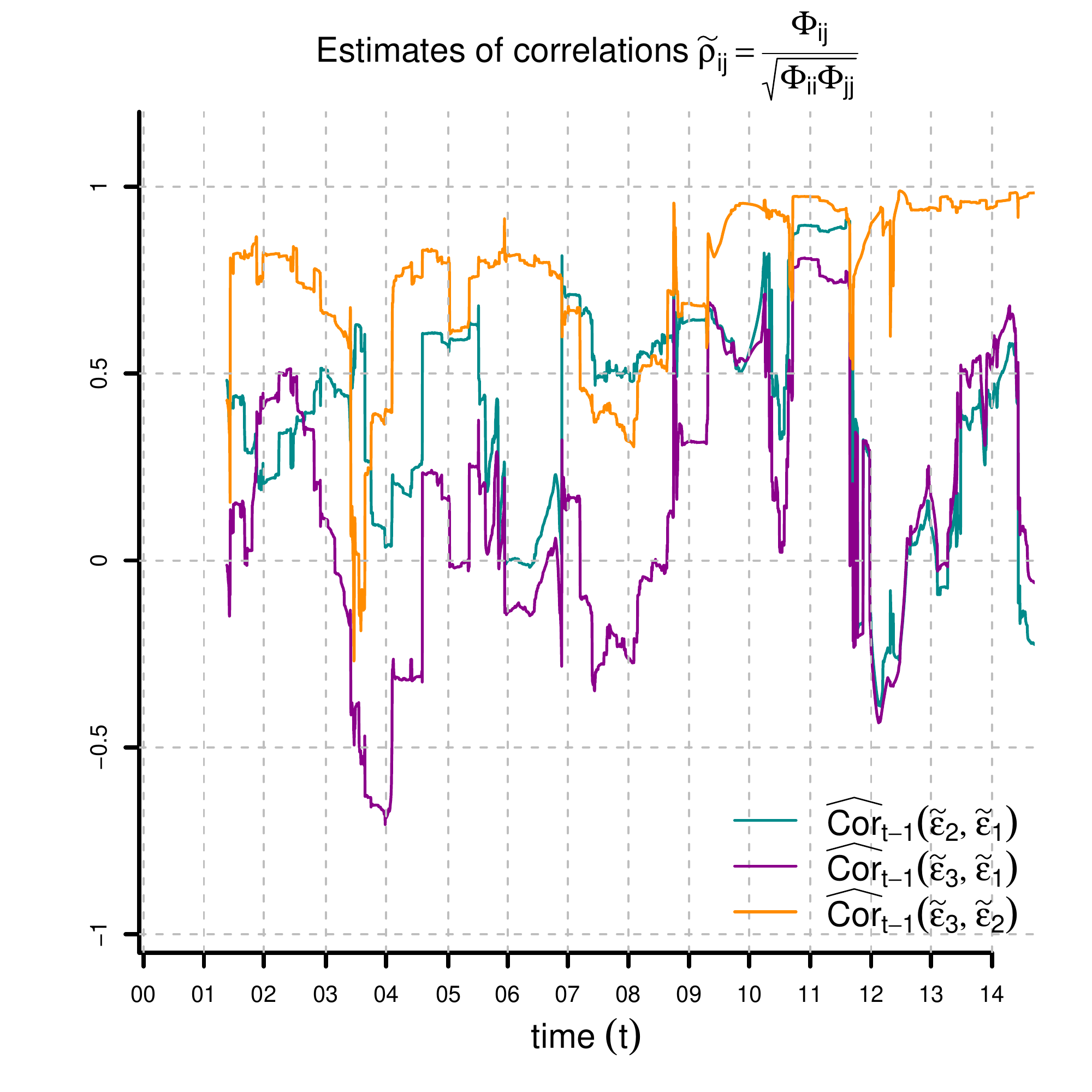}
\end{minipage}
\begin{minipage}[t]{0.42\textwidth}
\includegraphics[width=\textwidth]{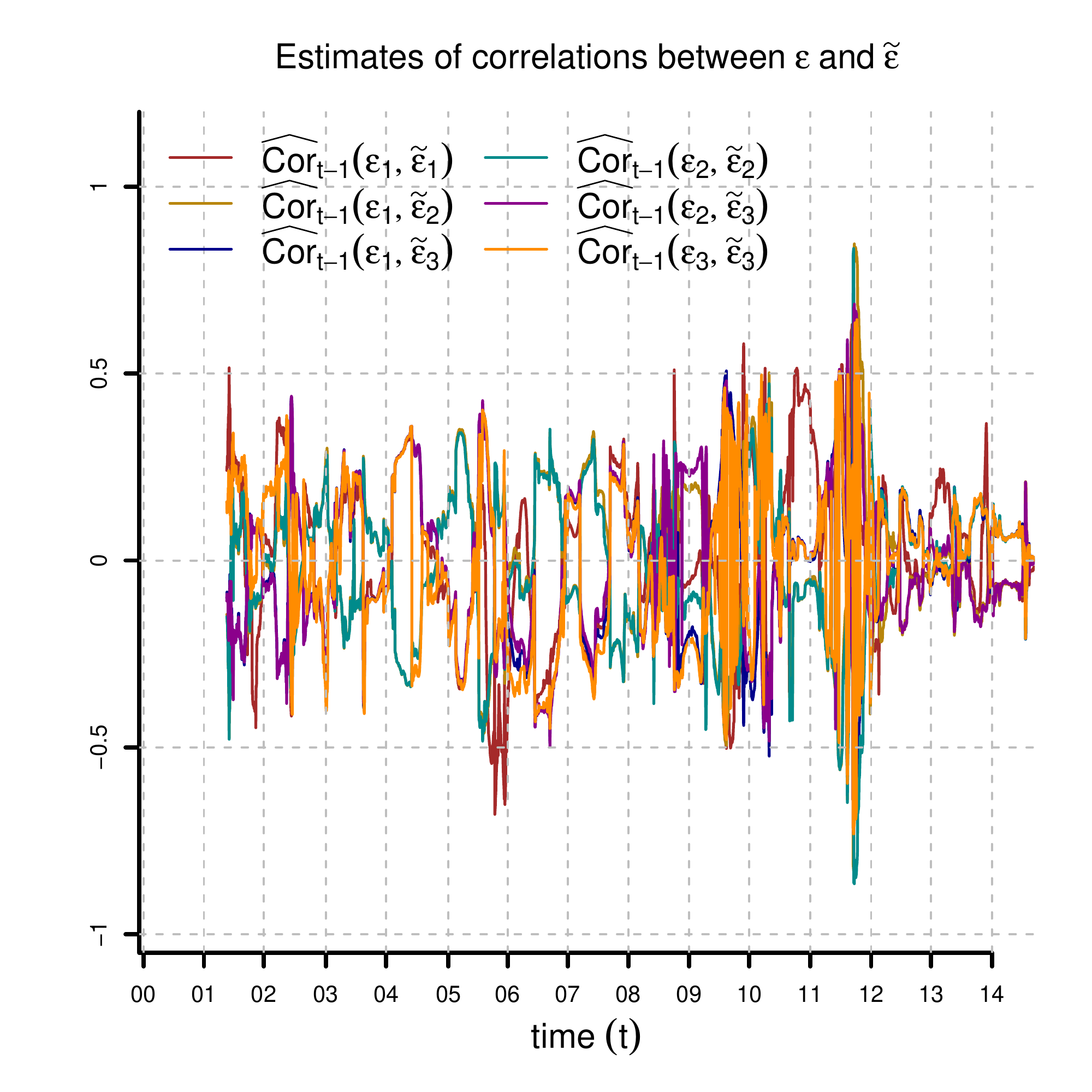}
\end{minipage}
\vspace{-15pt}
\caption{Estimation of correlations $\widetilde{\rho}_{21}$, $\widetilde{\rho}_{31}$ and $\widetilde{\rho}_{32}$~(lhs) and correlations $\mathrm{Cor}\left[\boldsymbol\varepsilon(t),\widetilde{\boldsymbol\varepsilon}(t)\mid\mathcal F(t-1)\right]$ (rhs). We use a time window of 252 observation for the regression. The residuals $\boldsymbol\varepsilon$ are calculated using the parameter estimates of Figure \ref{fig: parameters start}.}\label{fig: parameter crc end}
\end{figure}

\subsubsection{Simulation} The CRC approach has the remarkable property that yield curve increments can be simulated accurately and efficiently using Theorem \ref{theorem HJM view} and Corollary \ref{HJM under the real world measure}. In contrast, spot rate models with stochastic volatility without CRC have serious computational drawbacks. In such models, the calculation of the prevailing yield curve for given state variables requires Monte Carlo simulation. Therefore, the simulation of future yield curves requires nested simulations.

\subsubsection{Back-Testing}\label{subsec:bt} We backtest properties of the monthly returns of a buy and hold portfolio investing equal proportions of wealth in the zero-coupon bonds with times to maturity of 2, 3, 4, 5, 6 and 9 months and 1, 2, 3, 5, 7 and 10 years. We divide the sample into disjoint monthly periods and calculate the monthly return of this portfolio assuming that at the beginning of each period, we invest in the bonds with these times to maturity in equal proportions of wealth. The returns and some summary statistics are shown in Figure \ref{fig: returns}. We observe that the returns are positively skewed, leptokurtic and have heavier tails than the Gaussian distribution. These stylized facts are essential in applications. 

For each monthly period, we select a three-factor Vasi\v cek model and its CRC counterpart with stochastic volatility. Then, we simulate for each period realizations of the returns of the test portfolio. By construction, the Vasi\v cek model generates Gaussian log-returns and is unable to reproduce the stylized facts of the sample; see Tables \ref{tab: stats1} and \ref{tab: stats2} and Figure \ref{fig: stats}. Increasing the number of factors does not help much, because the log-returns remain Gaussian. On the other hand, CRC of the Vasi\v cek model with stochastic volatility provides additional modeling flexibility. In particular, we can see from the statistics in Table \ref{tab: stats2} and the confidence intervals in Figure \ref{fig: stats} that the model matches the return distribution better than the Vasi\v cek model. As explained in Figure \ref{fig: stats}, statistical tests assuming the independence of disjoint monthly periods show that the difference between the Vasi\v cek model and its CRC counterpart is statistically significant. We conclude that the three-factor CRC Vasi\v cek model is a parsimonious and tractable alternative that provides reasonable results.

\begin{table}[H]
\vspace{-10pt}
\centering
\includegraphics[width=\textwidth]{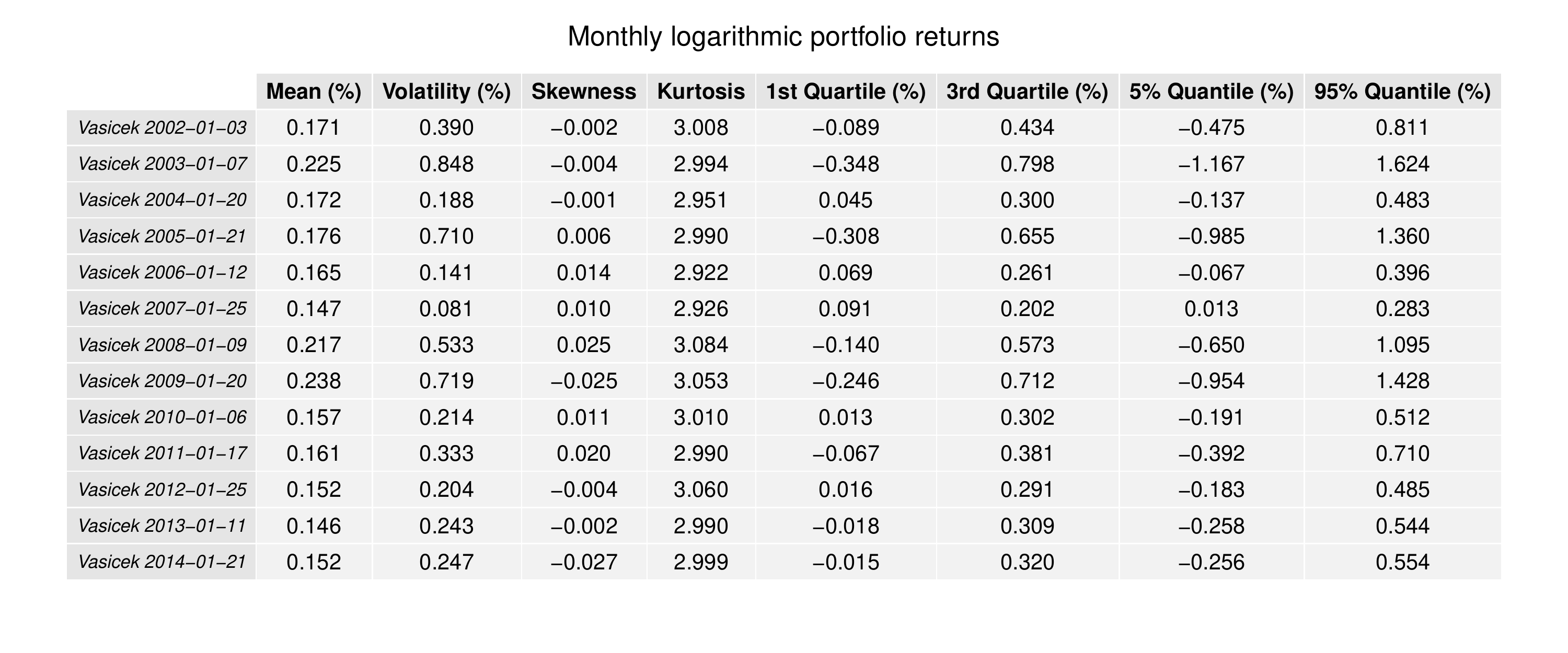}
\vspace{-35pt}
\caption{Statistics computed from simulations of the test portfolio returns for some of the monthly periods in the Vasi\v cek model. For each monthly period, we simulate $10^4$ realizations.}\label{tab: stats1}\vspace{-20pt}
\end{table}
\begin{table}[H]
\centering
\includegraphics[width=\textwidth]{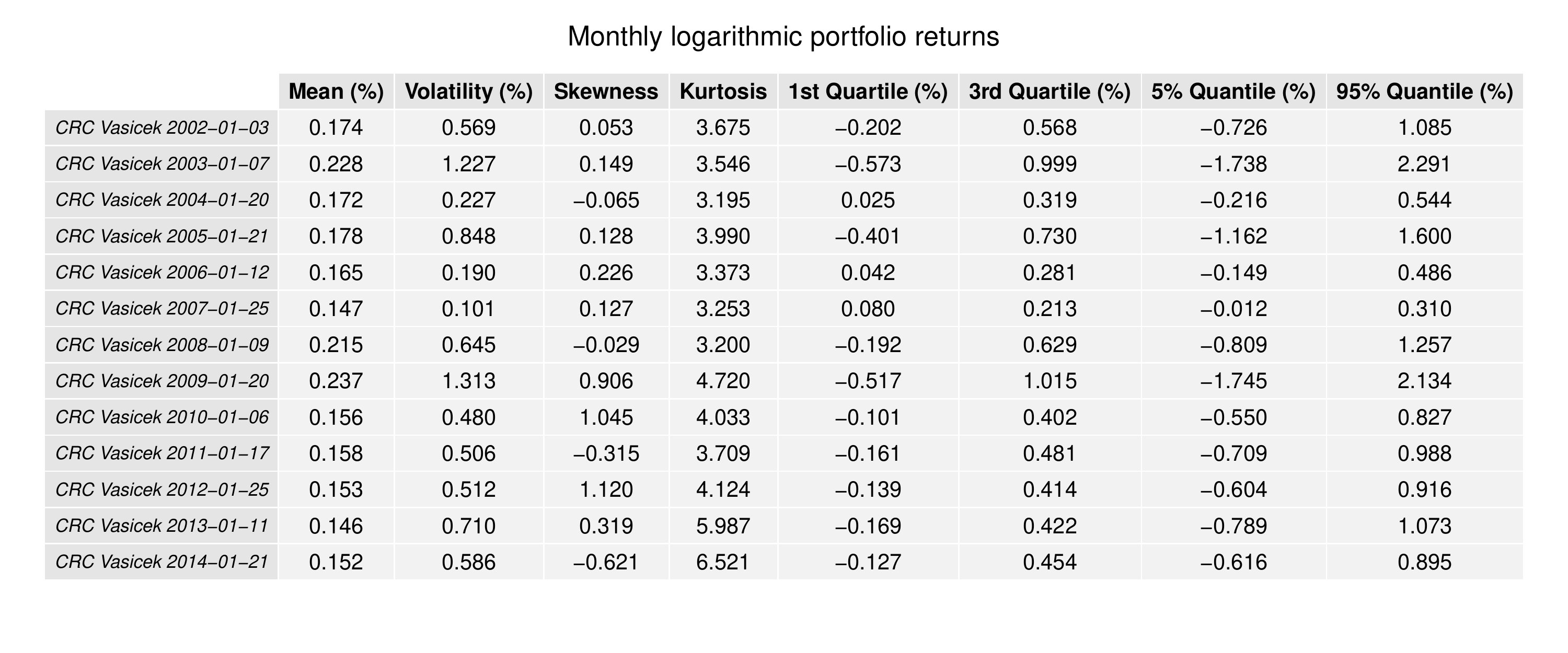}
\vspace{-35pt}
\caption{Statistics computed from the simulations of the test portfolio returns for some of the monthly periods in the consistent re-calibration (CRC) counterpart of the Vasi\v cek model with stochastic volatility. For each monthly period, we simulate $10^4$ realizations.}\label{tab: stats2}
\end{table}
\begin{figure}[H]
\vspace{-10pt}
\centering
\includegraphics[width=0.95\textwidth]{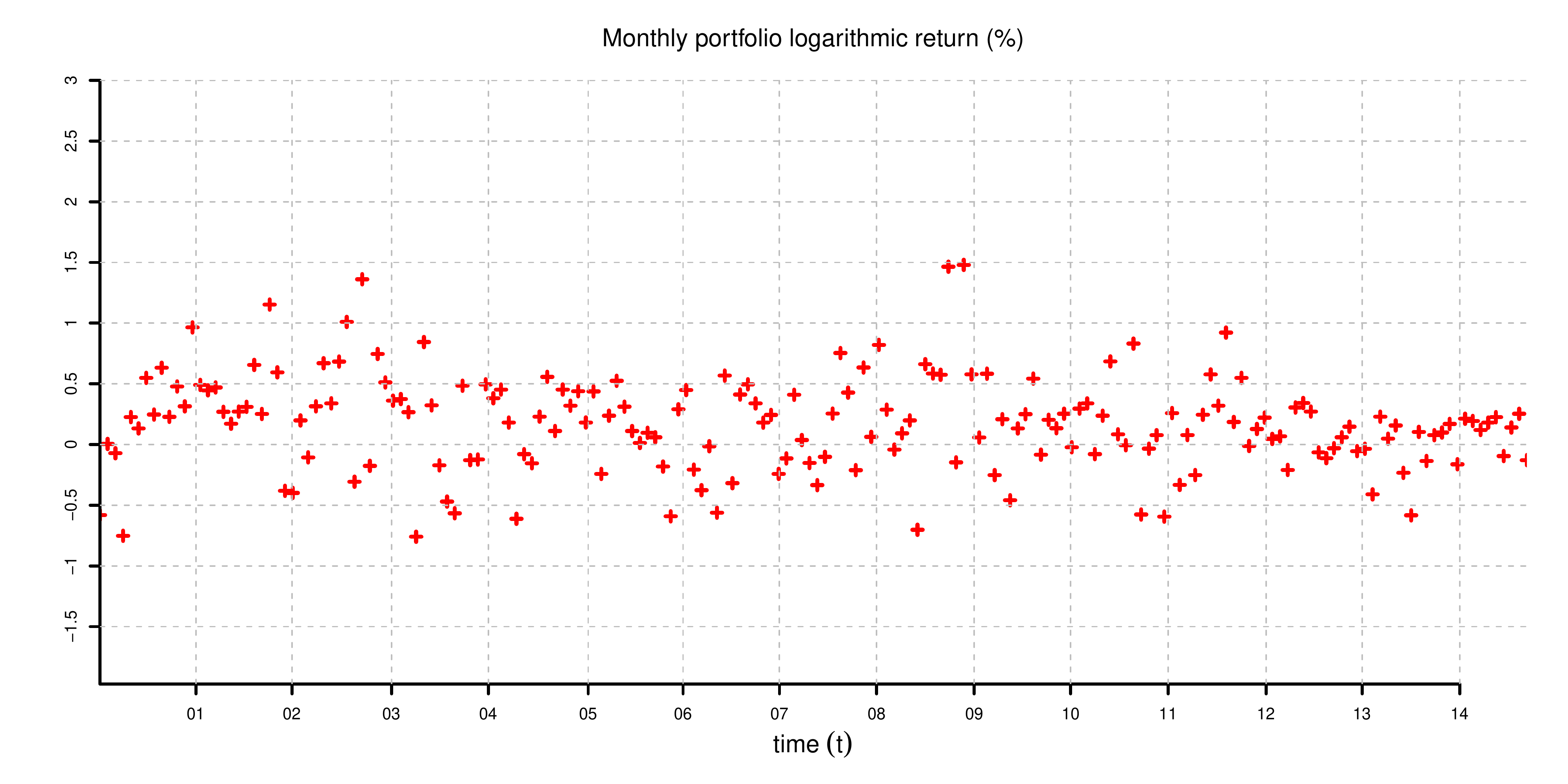}\\
\includegraphics[width=0.95\textwidth]{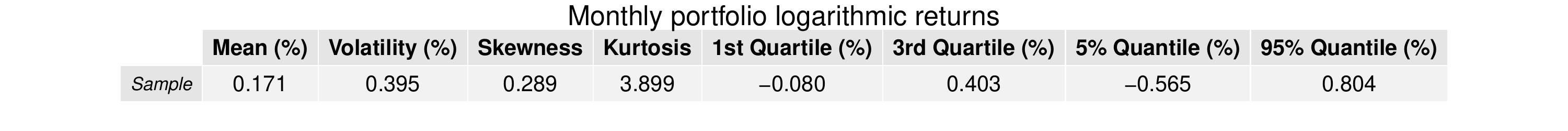}
\vspace{-10pt}
\caption{Logarithmic monthly returns of a buy and hold portfolio investing in equal wealth proportions in the zero-coupon bonds with times to maturity of 2, 3, 4, 5, 6 and 9 months and 1, 2, 3, 5, 7 and 10 years. For each monthly period, we calculate the logarithmic return of this portfolio assuming that at the beginning of each period, we are invested in the bonds with these times to maturity in equal proportions of wealth.}\label{fig: returns}\vspace{-20pt}
\end{figure}
\begin{figure}[H]
\centering
\begin{minipage}[t]{0.42\textwidth}
\includegraphics[width=\textwidth]{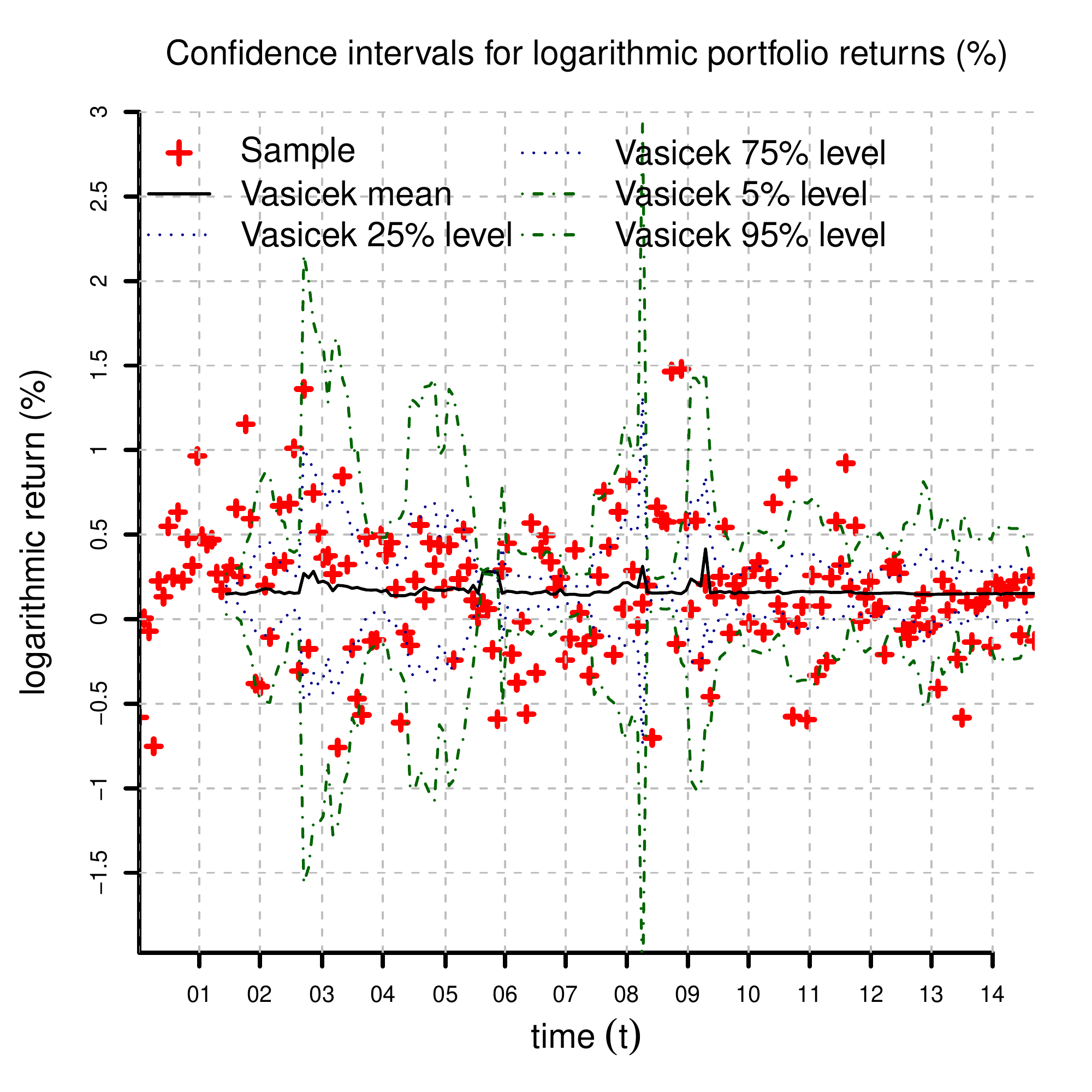}
\end{minipage}
\begin{minipage}[t]{0.42\textwidth}
\includegraphics[width=\textwidth]{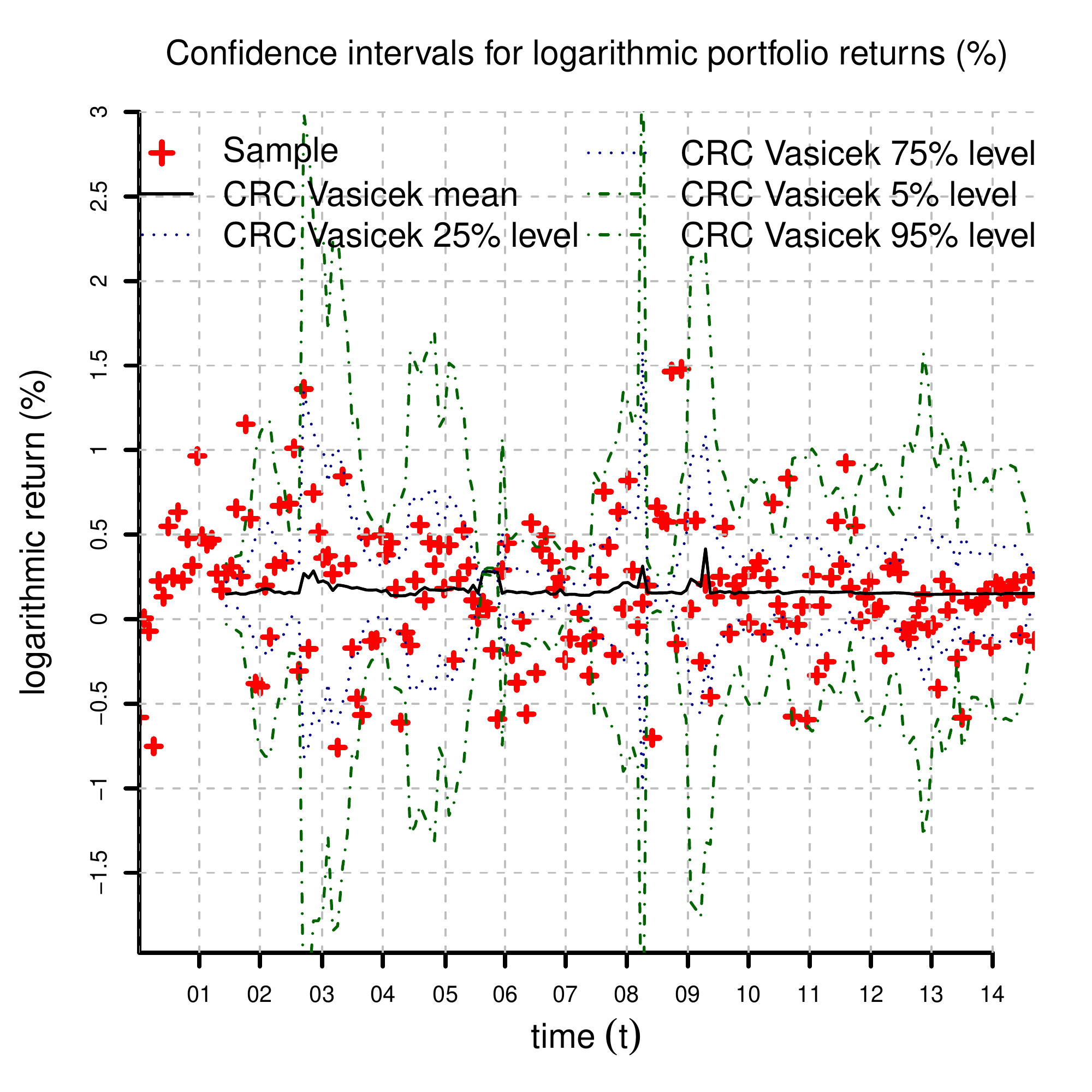}
\end{minipage}
\includegraphics[width=0.9\textwidth]{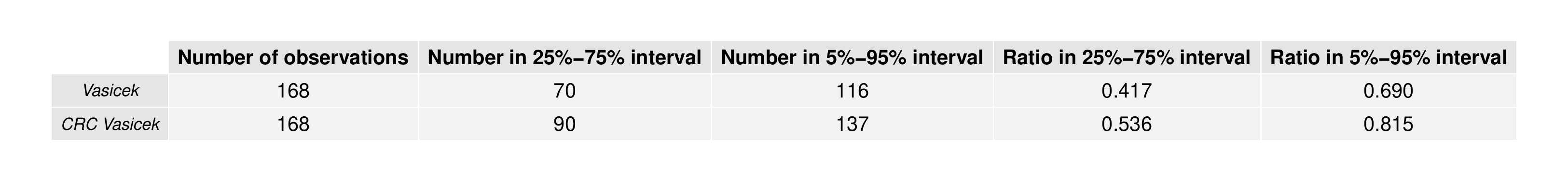}
\vspace{-10pt}
\caption{Confidence intervals computed from $10^4$ simulations of the test portfolio returns in the Vasi\v cek model and its CRC counterpart with stochastic volatility. For each monthly period, we check if the market return lies in the confidence interval. This is more often the case for the CRC than for the standard Vasi\v cek model. A one-sided binomial test assuming the independence of monthly periods shows that the difference is statistically significant ($p=0.0013$ for the $25\%$ and $p=$ 0.00017 for the $5\%$ quantiles). The result remains significant if every second month is discarded to account for dependencies ($p\approx 0.01$). This suggests that the CRC Vasic\v ek model is able to match the return distribution better than its counterpart with constant parameters.}\label{fig: stats}\vspace{-20pt}
\end{figure}

\subsubsection{Regulatory Framework}
The type of analysis that was performed in the previous section is an integral component of the present regulatory framework for risk management. In the Basel framework \cite{BIS}, the capital charge for the trading book is based on quantile risk measures. Under the internal model approach (\cite{BIS}, Section 2.VI.D), a bank calculates quantiles for the distribution of possible 10-day losses based on recent market data under the assumption that the trading book portfolio is held fixed over the time period. The approach relies on accurate modeling of the distribution of portfolio returns over holding periods of multiple days. A similar analysis is required by the Basel (\cite{BIS}, Section 2.VI.D) regulatory framework for model validation and stress testing: model validation is performed by backtesting the historical performance of the model, and stress tests are carried out using the same methodology by calibrating the model to historical periods of significant financial stress.

These tasks can be accomplished using the CRC approach by selecting suitable classes of affine models and parameter processes. The approach is fairly general, since there are few restrictions on the parameter processes. In particular, it allows for stochastic volatility and can be used to create realistic non-Gaussian distributions of multi-period bond returns (see Section~\ref{subsec:bt}). Nevertheless, computing these bond return distributions does not require nested simulations. This is crucial for reasons of efficiency. Moreover, the flexibility in the specification of the parameter processes makes the CRC approach well suited for stress testing, because it allows one to freely select and specify stress~scenarios.

\section{Conclusions}\label{sec: conclusion}

\begin{itemize}[leftmargin=*,labelsep=4mm]
\item {\it Flexibility and tractability.}
Consistent re-calibration of the multifactor Vasi\v cek model provides a tractable extension that allows parameters to follow stochastic processes. The additional flexibility can lead to better fits of yield curve dynamics and return distributions, as we demonstrated in our numerical example. Nevertheless, the model remains tractable. In particular, yield curves can be simulated efficiently using Theorem~\ref{theorem HJM view} and Corollary~\ref{HJM under the real world measure}. This allows one to efficiently calculate model quantities of interest in risk management, forecasting and pricing. 
\item {\it Model selection.} CRC models are selected from the data in accordance with the robust calibration principle of \cite{Harms}. First, historical parameters, market prices of risk and Hull--White extensions are inferred using a combination of volatility estimation, MLE and calibration to the prevailing yield curve via Formulas (\ref{eq: co-var estimate}--\ref{eq: lambda mat}, \ref{eq: re-calibration step 2}). The only choices in this inference procedure are the number of factors of the Vasi\v cek model and the window length $K$. Then, as a second step, the time series of estimated historical parameters are used to select a model for the parameter evolution. This results in a complete specification of the CRC model under the real world and the pricing measure. 
\item {\it Application to modeling of Swiss interest rates.} We fitted a three-factor Vasi\v cek CRC model with stochastic volatility to Swiss interest rate data. The model achieves a reasonably good fit in most time periods. The tractability of CRC allowed us to compute several model quantities by simulation. We looked at the historical performance of a representative buy and hold portfolio of Swiss bonds and concluded that a multifactor Vasi\v cek model is unable to describe the returns of this portfolio accurately. In contrast, the CRC version of the model provides the necessary flexibility for a good fit. 
\end{itemize}

\appendix 
\clearpage
\section*{\noindent Appendix A Proofs}\vspace{6pt} 
\renewcommand{\thesection}{\Alph{section}}
\refstepcounter{section}
\label{sec: proofs}
\renewcommand{\thesection}{}

\begin{proof}[Proof of Theorem \ref{theo: ARn prices}]
We prove Theorem \ref{theo: ARn prices} by induction as in (\cite{Wuethrich} Theorem 3.16) where ZCB prices are derived under the assumption that $\beta$ and $\Sigma$ are diagonal matrices. Note that we have the relation $P(m-1,m)=\exp\left(-\boldsymbol 1^\top\boldsymbol X(m-1)\Delta\right)$, which proves the claim for $t=m-1$. Assume that Theorem~\ref{theo: ARn prices} holds for $t+1\in\{2,\ldots,m-1\}$. We verify that it also holds for $t\in\{1,\ldots,m-2\}$. Under equivalent martingale measure $\p^\ast$, we have using the tower property for conditional expectations and the induction assumption:
\begin{myequation1}
\begin{aligned}
P(t,m)&=\exp\left\{-\boldsymbol 1^\top\boldsymbol X(t)\Delta\right\}\mathbb E^{*}\left[\mathbb E^{*}\left[\exp\left\{-\Delta\sum_{s=t+1}^{m-1}\boldsymbol 1^\top\boldsymbol X(s)\right\}\middle|{\cal F}(t+1)\right]\middle|{\cal F}(t)\right]\\&=\exp\left\{-\boldsymbol 1^\top\boldsymbol X(t)\Delta\right\}\mathbb E^{*}\left[P(t+1,m)\middle|{\cal F}(t)\right]\\&=\exp\left\{-\boldsymbol 1^\top\boldsymbol X(t)\Delta\right\}\mathbb E^{*}\left[\exp\left\{A(t+1,m)-\boldsymbol B(t+1,m)^\top\boldsymbol X(t+1)\right\}\middle|{\cal F}(t)\right]\\&=\exp\left\{-\boldsymbol 1^\top\boldsymbol X(t)\Delta+A(t+1,m)-\boldsymbol B(t+1,m)^\top\left(\boldsymbol b+\beta\boldsymbol X(t)\right)+\frac{1}{2}\boldsymbol B(t+1,m)^\top\Sigma\boldsymbol B(t+1,m)\right\}\\&=
\exp\left\{A(t+1,m)-\boldsymbol B(t+1,m)^\top\boldsymbol b+\frac{1}{2}\boldsymbol B(t+1,m)^\top\Sigma\boldsymbol B(t+1,m)-\left(\boldsymbol B(t+1,m)^\top\beta+\boldsymbol 1^\top\Delta\right)\boldsymbol X(t)\right\}.\nonumber
\end{aligned}
\end{myequation1}

This proves the following recursive formula for $m-1>t\geq0$:
\[
\begin{aligned}
A(t,m)&=A(t+1,m)-\boldsymbol B(t+1,m)^\top\boldsymbol b+\frac{1}{2}\boldsymbol B(t+1,m)^\top\Sigma\boldsymbol B(t+1,m),\\
\boldsymbol B(t,m)&=\beta^\top\boldsymbol B(t+1,m)+\boldsymbol 1\Delta.
\end{aligned}
\]

Finally, note that the recursive formula for $\boldsymbol B(\cdot,\cdot)$ implies:
\[
\boldsymbol B(t,m)=\sum_{s=0}^{m-t-1}\left(\beta^\top\right)^s\boldsymbol 1\Delta=\left(\mathds{1}-\beta^\top\right)^{-1}\left(\mathds{1}-\left(\beta^\top\right)^{m-t}\right)\boldsymbol 1\Delta.
\]

This concludes the proof.
\end{proof}

\begin{proof}[Proof of Theorem \ref{theo: ARn+ prices}.]
 The proof goes by induction as the proof of Theorem \ref{theo: ARn prices}.	
\end{proof}

\begin{proof}[Proof of Theorem \ref{theo: calibration}]
First, observe that the condition $\boldsymbol{y}^{(k)}(k)=\boldsymbol y$ imposes conditions only on the values $\theta(1),\ldots,\theta(M-1)$. Secondly, note that the vector $\boldsymbol\theta$, such that the condition is satisfied, can be calculated recursively in the following way.
\begin{enumerate}[leftmargin=*,labelsep=3mm]
\item \emph{First component $\theta_1$.} We have $A^{(k)}(k+1,k+2)=0$, $\boldsymbol B(k+1,k+2)=\boldsymbol 1\Delta$ and:
\begin{equation*}
A^{(k)}(k,k+2)=-\boldsymbol 1^\top\boldsymbol b\Delta-\theta_1\Delta+\frac{1}{2}\boldsymbol 1^\top\Sigma\boldsymbol 1\Delta^2,
\end{equation*}
see Theorem \ref{theo: ARn+ prices}.
Solving the last equation for $\theta_1$, we have:
\begin{equation*}
\theta_1
=\frac{1}{2}\boldsymbol 1^\top\Sigma\boldsymbol 1\Delta-\boldsymbol 1^\top\boldsymbol b-A^{(k)}(k,k+2)\Delta^{-1}.
\end{equation*}
From \eqref{eq: ARn+ yields} and the equation for $\boldsymbol B$ in Theorem~\ref{theo: ARn prices}, we obtain:
\begin{equation*}
A^{(k)}(k,k+2)=\boldsymbol 1^\top\left(\mathds{1}-\beta^2\right)\left(\mathds{1}-\beta\right)^{-1}\boldsymbol x\Delta-2 y_2\Delta.
\end{equation*}
This is equivalent to:
\begin{equation}\label{eq: first component}
\theta_1=\frac{1}{2}\boldsymbol 1^\top\Sigma\boldsymbol 1\Delta-\boldsymbol 1^\top\boldsymbol b-\boldsymbol 1^\top\left(\mathds{1}-\beta^2\right)\left(\mathds{1}-\beta\right)^{-1}\boldsymbol x+2y_2.
\end{equation}
\item \emph{Recursion $i\rightarrow i+1$.} Assume
we have determined $\theta_1,\ldots,\theta_i$ for $i=1,\ldots,M-2$. We want to determine $\theta_{i+1}$. We have $A^{(k)}(k+i+1,k+i+2)=0$, and iteration of the recursive formula for $A^{(k)}$ in Theorem~\ref{theo: ARn+ prices} implies:
\[
A^{(k)}(k,k+i+2)=-\sum_{s=k+1}^{k+i+1}\boldsymbol B(s,k+i+2)^\top\left(\boldsymbol b +\theta(s-k)\boldsymbol e_1\right)+\frac{1}{2}\sum_{s=k+1}^{k+i+1}\boldsymbol B(s,k+i+2)^\top\Sigma\boldsymbol B(s,k+i+2).
\]
Solving the last equation for $\theta_{i+1}$ and using $\boldsymbol B(k+i+1,k+i+2)=\boldsymbol 1\Delta$, we have:
\[
\begin{aligned}
\theta_{i+1}=&-\frac{1}{\Delta}A^{(k)}(k,k+i+2)-\frac{1}{\Delta}\sum_{s=k+1}^{k+i}\boldsymbol B(s,k+i+2)^\top\left(\boldsymbol b +\theta(s-k)\boldsymbol e_1\right)-\boldsymbol 1^\top\boldsymbol b\\&\quad+\frac{1}{2\Delta}\sum_{s=k+1}^{k+i+1}\boldsymbol B(s,k+i+2)^\top\Sigma\boldsymbol B(s,k+i+2).
\end{aligned}
\]
From \eqref{eq: ARn+ yields} and the equation for $\boldsymbol B$ in Theorem~\ref{theo: ARn prices}, we obtain:
\[
A^{(k)}(k,k+i+2)=\boldsymbol 1^\top\left(\mathds{1}-\beta^{i+2}\right)\left(\mathds{1}-\beta\right)^{-1}\boldsymbol x\Delta-(i+2)y_{i+2}\Delta.
\]
This is equivalent to:
\begin{equation}\label{eq: recursion}
\begin{aligned}
\theta_{i+1}&=(i+2)y_{i+2}-\boldsymbol 1^\top\left(\mathds{1}-\beta^{i+2}\right)\left(\mathds{1}-\beta\right)^{-1}\boldsymbol x-\frac{1}{\Delta}\sum_{s=k+1}^{k+i}\boldsymbol B(s,k+i+2)^\top\left(\boldsymbol b+\theta_{s-k}\boldsymbol e_1\right)\\&\quad-\boldsymbol 1^\top\boldsymbol b+\frac{1}{2\Delta}\sum_{s=k+1}^{k+i+1}\boldsymbol B(s,k+i+2)^\top\Sigma\boldsymbol B(s,k+i+2)\\&=(i+2)y_{i+2}-\boldsymbol 1^\top\left(\mathds{1}-\beta^{i+2}\right)\left(\mathds{1}-\beta\right)^{-1}\boldsymbol x-\frac{1}{\Delta}\sum_{s=k+1}^{k+i+1}\boldsymbol B(s,k+i+2)^\top\boldsymbol b\\&\quad-\frac{1}{\Delta}\sum_{s=k+1}^{k+i}B_1(s,k+i+2)\theta_{s-k}+\frac{1}{2\Delta}\sum_{s=k+1}^{k+i+1}\boldsymbol B(s,k+i+2)^\top\Sigma\boldsymbol B(s,k+i+2).
\end{aligned}
\end{equation}
\end{enumerate}

This recursion allows one to determine the components of $\boldsymbol\theta$. Note that Equation \eqref{eq: recursion} can be written~as: 
\[
\left({\cal C}(\beta)\boldsymbol\theta\right)_{i+1}=z_{i+1}\left(\boldsymbol b,\beta,\Sigma, \boldsymbol x,\boldsymbol y\right),\quad i=1,\ldots,M-2.
\] 

Observe that the lower triangular matrix ${\cal C}(\beta)$ is invertible since $\det{\cal C}(\beta)=\Delta^{M-1}>0$. Hence, Equations \eqref{eq: first component} and \eqref{eq: recursion} prove \eqref{eq: calibration matrix}.
\end{proof}

\begin{proof}[Proof of Theorem \ref{theorem HJM view}]
We add and subtract $-A^{(k)}(k,m)+\boldsymbol B^{(k)}(k,m)^\top\boldsymbol{\cal X}(k)$ to the right hand side of Equation \eqref{start HJM} and obtain:
\begin{myequation2}\label{eq: identity1}
\begin{aligned}
{\cal Y}(k+1,m)\left(m-(k+1)\right)\Delta&=A^{(k)}(k,m)-A^{(k)}(k+1,m)-A^{(k)}(k,m)\\&\quad+\boldsymbol B^{(k)}(k,m)^\top\boldsymbol{\cal X}(k)-\boldsymbol B^{(k)}(k,m)^\top\boldsymbol{\cal X}(k)\\&\quad+\boldsymbol B^{(k)}(k+1,m)^\top \left(\boldsymbol b(k)+\theta^{(k)}(1)\boldsymbol e_1+\beta(k)\boldsymbol {\cal X}(k)+\Sigma(k)^{\frac{1}{2}}\boldsymbol\varepsilon^\ast(k+1)
\right).
\end{aligned}
\end{myequation2}

We have the following two identities from Section~\ref{subsubsec crc step 2}:
\begin{myequation1}
\begin{aligned}
-A^{(k)}(k,m)+\boldsymbol B^{(k)}(k,m)^\top\boldsymbol{\cal X}(k)&={\cal Y}(k,m)(m-k)\Delta,\\
A^{(k)}(k,m) -A^{(k)}(k+1,m)&=-\boldsymbol B^{(k)}(k+1,m)^\top\left(\boldsymbol b(k)+\theta^{(k)}(1)\boldsymbol e_1\right)+\frac{1}{2}\boldsymbol B^{(k)}(k+1,m)^\top\Sigma(k)\boldsymbol B^{(k)}(k+1,m). 
\end{aligned}
\end{myequation1}

Therefore, the right hand side of \eqref{eq: identity1} is rewritten as:
\begin{myequation1}
\begin{aligned}
{\cal Y}(k+1,m)
{(m-(k+1))\Delta}
&={\cal Y}(k,m){(m-k)\Delta}+\left(\boldsymbol B^{(k)}(k+1,m)^\top\beta(k)-\boldsymbol B^{(k)}(k,m)^\top\right)\boldsymbol{\cal X}(k)\\&\quad+\frac{1}{2}\boldsymbol B^{(k)}(k+1,m)^\top\Sigma(k)\boldsymbol B^{(k)}(k+1,m)+\boldsymbol B^{(k)}(k+1,m)^\top\Sigma(k)^{\frac{1}{2}}\boldsymbol\varepsilon^\ast(k+1).
\end{aligned}
\end{myequation1}

Observe that:
\[
\begin{aligned}
\boldsymbol B^{(k)}(k+1,m)^\top\beta(k)=\left(\sum_{s=0}^{m-k-2}\left(\beta^\top(k)\right)^s\boldsymbol 1\right)^\top\beta(k)\Delta=\boldsymbol 1^\top\sum_{s=1}^{m-k-1}\beta(k)^{s}\Delta=\boldsymbol B^{(k)}(k,m)^\top-\boldsymbol 1^\top\Delta,
\end{aligned}
\]
and that $Y(k,k+1)=\boldsymbol 1^\top\boldsymbol {\cal X}(k)$. This proves the claim.	
\end{proof}

\end{document}